%% file: MINE_estimation.tex
\documentclass[11pt]{article}
\usepackage[margin=1.13in]{geometry}

\date{}

\usepackage[numbers,sort&compress]{natbib}
\usepackage{graphicx}
\usepackage[space]{grffile}
\usepackage[caption=false]{subfig}
\usepackage{amsmath, amssymb, amsthm}
\usepackage[colorlinks=true, urlcolor=blue, citecolor=blue, linkcolor=blue, filecolor=blue]{hyperref}
\usepackage{comment}
\usepackage{soul} 
\usepackage{sidecap}
\usepackage[textsize=footnotesize]{todonotes}

\theoremstyle{plain}
\newtheorem{theorem}{Theorem}
\newtheorem{proposition}{Proposition}
\newtheorem*{corollary}{Corollary}

\theoremstyle{remark}
\newtheorem*{remark}{Remark}

\usepackage{setspace}
\makeatletter
\renewcommand{\todo}[2][]{
    \@todo[caption={#2}, #1]{\begin{spacing}{0.75}#2\end{spacing}}
}
\makeatother

\newcommand{\pathToCommon}{.}
\newcommand{\pathToCommonFigs}{.}
\newcommand{\pathToFigures}{.}
\input{\pathToCommon/commands.tex}

\makeatletter
\renewcommand*{\@fnsymbol}[1]{\ensuremath{\ifcase#1\or 1\or *\or \dagger\or 2\or 3\or 4\or 5\or **\or 6\or \mathsection\or \mathparagraph\or \|\or \ddagger\or \dagger\dagger
   \or \ddagger\ddagger \else\@ctrerr\fi}}
\makeatother

\title{Measuring dependence powerfully and equitably}
\author{
Yakir A.\ Reshef\footnote{School of Engineering and Applied Sciences, Harvard University.} \footnote{Co-first author.} \footnote{To whom correspondence should be addressed. Email: \url{yakir@seas.harvard.edu}}
\and
David N. Reshef\footnote{Department of Computer Science, Massachusetts Institute of Technology.} \footnotemark[2]
\and
Hilary K. Finucane\footnote{Department of Mathematics, Massachusetts Institute of Technology.}
\and
Pardis C. Sabeti\footnote{Department of Organismic and Evolutionary Biology, Harvard University.} \footnote{Broad Institute of MIT and Harvard.} \footnote{Co-last author.}
\and
Michael M. Mitzenmacher\footnotemark[1] \footnotemark[8]
}

\begin{document}

\maketitle

\begin{abstract}
Given a high-dimensional data set we often wish to find the strongest relationships within it. A common strategy is to evaluate a measure of dependence on every variable pair and retain the highest-scoring pairs for follow-up. This strategy works well if the statistic used is {\em equitable} \cite{reshef2015equitability}, i.e., if, for some measure of noise, it assigns similar scores to equally noisy relationships regardless of relationship type (e.g., linear, exponential, periodic).

In this paper, we introduce and characterize a population measure of dependence called $\popMIC$. We show three ways that $\popMIC$ can be viewed: as the population value of $\MIC$, a highly equitable statistic from \cite{MINE}, as a canonical ``smoothing'' of mutual information, and as the supremum of an infinite sequence defined in terms of optimal one-dimensional partitions of the marginals of the joint distribution. Based on this theory, we introduce an efficient approach for computing $\popMIC$ from the density of a pair of random variables, and we define a new consistent estimator $\MICestE$ for $\popMIC$ that is efficiently computable. In contrast, there is no known polynomial-time algorithm for computing the original equitable statistic $\MIC$. We show through simulations that $\MICestE$ has better bias-variance properties than $\MIC$. We then introduce and prove the consistency of a second statistic, $\TICestE$, that is a trivial side-product of the computation of $\MICestE$ and whose goal is powerful independence testing rather than equitability.

We show in simulations that $\MICestE$ and $\TICestE$ have good equitability and power against independence respectively. The analyses here complement a more in-depth empirical evaluation of several leading measures of dependence \citep{reshef2015comparisons} that shows state-of-the-art performance for $\MICestE$ and $\TICestE$.
\end{abstract}

\section{Introduction}
The growing dimensionality of today's data sets has popularized the idea of {\em hypothesis-generating science}, whereby a data set is used not to test existing hypotheses but rather to help a researcher formulate new ones. A common approach among practitioners is to evaluate some statistic on many candidate variable pairs in a data set, sort the variable pairs from highest-scoring to lowest, and manually examine all the pairs above a threshold score~\citep{storey2003statistical,emilsson2008genetics}.

A popular class of statistics used for such analyses is {\em measures of dependence}, i.e., statistics whose population value is 0 in cases of statistical independence and non-zero otherwise. Measures of dependence are attractive because they guarantee that asymptotically no non-trivial relationship will erroneously be declared trivial. In the setting of continuous-valued data, which is our focus, there is a long line of fruitful research on such statistics including, e.g.,~\cite{hoeffding1948non, renyi1959measures, breiman1985estimating, paninski2003estimation, szekely2007measuring, gretton2005measuring, MINE, gretton2012kernel, lopez2013randomized, heller2013consistent, jiang2015nonparametric}.

The utility of a measure of dependence $\hvphi$ can be assessed in two ways. The first is {\em power against independence}, i.e., the power of independence testing based on $\hvphi$ to detect various types of non-trivial relationships. This is an important goal for datasets that have very few non-trivial relationships, or only very weak relationships that are difficult to detect. Often, however, the number of relationships declared statistically significant by a measure of dependence greatly exceeds the number of relationships that can then be explored further. For example, biological datasets often contain many non-trivial relationships, but testing a preliminary finding for further corroboration may take extensive manual lab work, or a study on human or animal subjects. In this case, it is tempting to restrict follow-up to relationships with high values of $\hvphi$, but this can skew the direction of follow-up work: if $\hvphi$ systematically assigns higher scores to, say, linear relationships than to non-linear ones, relatively noisy linear relationships might crowd out strong non-linear relationships from the top-scoring set.

Motivated by this problem, in a companion paper \citep{reshef2015equitability} we define a second way of assessing a measure of dependence called {\em equitability}. Informally, an equitable statistic is one that, for some measure of relationship strength, assigns similar scores to equally strong relationships regardless of relationship type. For instance, we may want our measure of dependence to also have the property that on noisy functional relationships it assigns similar scores to relationships with the same $R^2$, i.e., the squared Pearson correlation between the observed y-values and the x-values passed through the underlying function in question \citep{MINE}. Or, alternatively, we may want the value of our statistic to tell us about the proportion of points coming from the deterministic component of a mixture containing part signal and part uniform noise~\citep{ding2013copula}. Defining measures of dependence that achieve good equitability with respect to interesting measures of relationship strength is a new and challenging problem, with a number of different formalizations. (See, e.g., \cite{reshef2015equitability} and \cite{ding2013copula} cited above, as well as \cite{kinney2014equitability} along with associated technical comments \cite{reshef2014comment} and \cite{Murrell2014comment}.)

In this paper, we introduce and theoretically characterize two new measures of dependence that we empirically show to have good equitability with respect to $R^2$ and power against independence, respectively. We begin by introducing a new population measure of dependence called $\popMIC$. Given a pair of jointly distributed random variables $(X,Y)$, $\popMIC(X,Y)$ is the supremum, over all finite grids $G$ imposed on the support of $(X,Y)$, of the mutual information of the discrete distribution induced by $(X,Y)$ on the cells of $G$, subject to a regularization based on the resolution of $G$. We prove three results, each of which gives a different way that this population quantity can be viewed.

\begin{enumerate}
\item $\popMIC$ is the population value of the maximal information coefficient ($\MIC$), a statistic introduced in \cite{MINE} that is highly equitable with respect to $R^2$ on a large class of noisy functional relationships. Simple corollaries of this result simplify and strengthen many of the theoretical results proven in \cite{MINE} about $\MIC$.
\item $\popMIC$ is a minimal ``smoothing'' of mutual information, in the sense that the regularization in the definition of $\popMIC$ renders it uniformly continuous as a function of random variables, and no smaller regularization achieves continuity. A corollary of this is that $\popMIC$ is uniformly continuous while mutual information is not continuous.
\item $\popMIC$ is the supremum of an infinite sequence defined in terms of optimal partitions of the marginal distributions of $(X,Y)$ rather than optimal (two-dimensional) grids imposed on the joint distribution. This characterization greatly simplifies the computation of $\popMIC$ and associated quantities.
\end{enumerate}

After proving these three results, we leverage them to introduce efficient algorithms both for approximately computing $\popMIC$ and for estimating it from finite samples. We first provide an efficient algorithm that in many cases allows for computation to arbitrary precision of the $\popMIC$ of a pair of random variables whose joint density is known. We then introduce a statistic, called $\MICestE$, that we prove is a consistent estimator of $\popMIC$. In contrast to the $\MIC$ statistic from \cite{MINE}, for which no efficient algorithm is known and a heuristic algorithm is used in practice, $\MICestE$ is efficiently computable. It has a better runtime complexity than the heuristic algorithm currently in use for computing the original $\MIC$ statistic, and is orders of magnitude faster in practice.

With a consistent and fast estimator for $\popMIC$ in hand, we turn to empirical analysis of its performance. Specifically, we show through simulation that $\MICestE$ has better bias/variance properties than the heuristic algorithm used in \cite{MINE} for computing $\MIC$, which has no theoretical convergence guarantees. Our analysis also reveals that the main parameter of $\MICestE$ can be used to tune statistical performance toward either stronger or weaker relationships in general. After studying the bias/variance properties of $\MICestE$, we then demonstrate via simulation that it outperforms currently available methods in terms of equitability with respect to $R^2$. Notably, we show this performance advantage both on the set of functional relationships analyzed in \cite{MINE} as well as on a large set of randomly chosen noisy functional relationships.

We choose in this paper to analyze equitability specifically with respect to $R^2$, rather than some other notion of relationship strength, because $R^2$ on noisy functional relationships is a simple measure with broad familiarity and intuitive interpretation among practitioners. Of course, it is also important to develop measures of dependence that are equitable with respect to notions of relationship strength besides $R^2$ or on families of relationships besides noisy functional relationships; however, our focus here remains on the ``simple'' case of $R^2$ on noisy functional relationships.

Importantly, we note that although there are methods for directly estimating the $R^2$ of a noisy functional relationship via nonparametric regression (see, e.g., \cite{cleveland1988locally, stone1977consistent}), those methods are not applicable in the context of equitability because they are not measures of dependence. That is, because non-parametric regression methods \textit{assume} a functional form for the relationship in question, they can give trivial scores to non-functional relationships, even in the large-sample limit. A simple example of this is when a distribution is supported on a circle, such that the regression function is constant. In contrast, a \textit{measure of dependence} is guaranteed never to make this ``mistake''. A measure of dependence that is equitable with respect to $R^2$ can therefore be viewed either as an ``upgraded'' measure of dependence that also comes with some of the interpretability properties of non-parametric regression, or as an ``upgraded'' approximate non-parametric regression method that also has the robustness properties of a measure of dependence.

The main strength of $\MICestE$ is equitability rather than power to reject a null hypothesis of independence. In some settings, though, it may be important to have good power against independence. We therefore introduce here a statistic closely related to $\MICestE$ called the total information coefficient $\TICestE$. We prove the consistency of testing for independence using $\TICestE$, and show via simulations that it achieves excellent power in practice, performing comparably to or better than current methods. Because $\TICestE$ arises naturally as a side-product of the computation of $\MICestE$, it is available ``for free'' once $\MICestE$ has been computed. This leads us to propose a data analysis strategy consisting of first using $\TICestE$ to filter out non-significant relationships, and then ranking the remaining ones using the simultaneously computed values of $\MICestE$.

In addition to the companion paper \cite{reshef2015equitability}, which focuses on the theory behind equitability, this paper is accompanied by a second companion work \citep{reshef2015comparisons} that explores in detail the empirical performance of the methods introduced here.
That paper shows, by comparing $\MICestE$ and $\TICestE$ to several leading measures of dependence under many different sampling and noise models, that the equitability of $\MICestE$ on noisy functional relationships and the power of independence testing using $\TICestE$ are both state-of-the-art. It also shows that these methods can be computed very fast in practice.

Taken together, our results shed significant light on the theory behind the maximal information coefficient, and suggest that $\TICestE$ and $\MICestE$ are a useful pair of methods for data exploration. Specifically, they point to joint use of these two statistics to filter and then rank relationships as a fast and practical way to explore large data sets by measuring dependence both powerfully and equitably.

\section{Preliminaries}
We work extensively in this paper with grids and discrete distributions over their cells. Given a grid $G$ and a point $(x,y)$, we define the function $\mbox{row}_G(y)$ to be the row of $G$ containing $y$ and we define $\mbox{col}_G(x)$ analogously. For a pair $(X,Y)$ of jointly distributed random variables, we write $(X,Y)|_G$ to denote $\left( \mbox{col}_G(X), \mbox{row}_G(Y) \right)$, and we use $I((X,Y)|_G)$ to denote the discrete mutual information \citep{Cover2006,csiszar2004information,csiszar2008axiomatic} between $\mbox{col}_G(X)$ and $\mbox{row}_G(Y)$. Given a finite sample $D$ from the distribution of $(X,Y)$, we sometimes use $D$ to refer both to the set of points in the sample as well as to a point chosen uniformly at random from $D$. In the latter case, it will then make sense to talk about, e.g., $D|_G$ and $I(D|_G)$.

For natural numbers $k$ and $\ell$, we use $G(k, \ell)$ to denote the set of all $k$-by-$\ell$ grids (possibly with empty rows/columns). A grid $G$ is an equipartition of $(X,Y)$ if all the rows of $(X,Y)|_G$ have the same probability mass, and all the columns do as well. We also use the term equipartition in the analogous way for one-dimensional partitions into just rows or columns. For a one-dimensional partition $P$ into rows and a one-dimensional partition $Q$ into columns, we write $(P,Q)$ to refer to the grid constructed from these two partitions. When a partition $P$ can be obtained from a partition $P'$ by addition of separators alone, we write $P' \subset P$.

Finally, let us establish some notation for infinite matrices. We use $m^\infty$ to denote the space of infinite matrices equipped with the supremum norm. Given a matrix $A \in m^\infty$, we often examine only the $k,\ell$-th entries of $A$ for which $k\ell \leq i$ for some $i$. Thus, for $i \in \Z^+$, we define the projection $r_i : m^\infty \rightarrow m^\infty$ via
\[ r_i(A)_{k,\ell} = \left\{
        \begin{array}{ll}
            A_{k,\ell} & \quad k\ell \leq i \\
            0 & \quad k\ell > i
        \end{array}
    \right. . \]

\section{The population maximal information coefficient \texorpdfstring{$\popMIC$}{MIC*}}
In this section, we define and characterize the population maximal information coefficient $\popMIC$. We begin by defining the population quantity $\popMIC(X,Y)$ for a pair of jointly distributed random variables $(X,Y)$. We then show three different ways to characterize this population quantity: first, as the large-sample limit of the statistic $\MIC$ from \cite{MINE}; second, as a minimally smoothed version of mutual information; and third, as the supremum of an infinite sequence defined in terms of optimal one-dimensional partitions of the marginals of the joint distribution of $(X,Y)$. We conclude the section by showing how the third characterization leads to an efficient approach for computing $\popMIC$ from the density of $(X,Y)$.

\subsection{Defining \texorpdfstring{$\popMIC$}{MIC*}}
\label{subsec:defining_popMIC}
The population maximal information can be defined in several equivalent ways, as we will see later. For now, we begin with the simplest definition.

\begin{definition}
Let $(X,Y)$ be jointly distributed random variables. The {\em population maximal information coefficient} ($\popMIC$) of $(X,Y)$ is defined by
\[
\popMIC(X, Y) = \sup_G \frac{I((X,Y)|_G)}{\log \|G\| }
\]
where $\|G\|$ denotes the minimum of the number of rows of $G$ and the number of columns of $G$.
\end{definition}
Given that $I(X,Y) = \sup_G I( (X,Y)|_G )$ (see, e.g., Chapter 8 of \cite{Cover2006}), this can be viewed as a regularized version of mutual information that penalizes complicated grids and ensures that the result falls between 0 and 1.

Before we continue, we state one simple equivalent definition of $\popMIC$ that is useful for the results in this section. This definition views $\popMIC$ as the supremum of a matrix called the {\em population characteristic matrix}, defined below.

\begin{definition}
\label{def:charmatrix}
Let $(X,Y)$ be jointly distributed random variables. Let
\[ I^*((X, Y), k, \ell) = \max_{G \in G(k,\ell)} I((X,Y)|_G) . \]
The {\em population characteristic matrix} of $(X,Y)$, denoted by $M(X,Y)$, is defined by
\[ M(X,Y)_{k,\ell} = \frac{I^*((X,Y),k,\ell)}{\log \min \{k, \ell\}} \]
for $k, \ell > 1$.
\end{definition}

It is easy to see the following:
\begin{proposition}
Let $(X,Y)$ be jointly distributed random variables. We have
\[ \popMIC(X, Y) = \sup M(X,Y) \]
where $M(X,Y)$ is the population characteristic matrix of $(X,Y)$.
\end{proposition}

The population characteristic matrix is so named because just as $\popMIC$, the supremum of this matrix, captures a sense of relationship strength, other properties of this matrix correspond to different properties of relationships. For instance, later in this paper we introduce an additional property of the characteristic matrix, the total information coefficient, that is useful for testing for the presence or absence of a relationship rather than quantifying relationship strength.

\subsection{First alternate characterization: \texorpdfstring{$\popMIC$}{MIC*} is the population value of \texorpdfstring{$\MIC$}{MIC}}
With $\popMIC$ defined, we now state our first alternate characterization of it, as the large-sample limit of the statistic $\MIC$ introduced in \cite{MINE}. We begin by first reproducing a description of $\MIC$ from~\cite{MINE}, via the two definitions below.

\begin{definition}[\citealp{MINE}]
\label{def:sample_char_matrix}
Let $D \subset \R^2$ be a set of ordered pairs. The {\em sample characteristic matrix} $\widehat{M}(D)$ of $D$ is defined by
\[
\widehat{M}(D)_{k,\ell} = \frac{I^*(D, k, \ell)}{ \log \min \{k, \ell \}}  .
\]
\end{definition}

\begin{definition}[\citealp{MINE}]
Let $D \subset \R^2$ be a set of $n$ ordered pairs, and let $B : \Z^+ \rightarrow \Z^+$. We define
\[ \MIC_B(D) = \max_{k,\ell \leq B(n)} \widehat{M}(D)_{k,\ell} . \]
\end{definition}
where the function $B(n)$ is specified by the user. In \cite{MINE}, it was suggested that $B(n)$ be chosen to be $n^\alpha$ for some constant $\alpha$ in the range of $0.5$ to $0.8$. (The statistics we introduce later will have an analogous parameter. See Section~\ref{subsubsec:choosingB}.)

We have shown the following result about convergence of functions of the sample characteristic matrix to their population counterparts, a consequence of which is the convergence of $\MIC$ to $\popMIC$. (In the theorem statement below, recall that $m^\infty$ is the space of infinite matrices equipped with the supremum norm, and given a matrix $A$ the projection $r_i$ zeros out all the entries $A_{k,\ell}$ for which $k\ell > i$.)

\begin{theorem}
\label{thm:estimator}
Let $f : m^\infty \rightarrow \R$ be uniformly continuous, and assume that $f \circ r_i \rightarrow f$ pointwise. Then for every random variable $(X,Y)$, we have
\[
\left( f \circ r_{B(n)} \right) \left(\widehat{M}(D_n) \right) \rightarrow f(M(X,Y))
\]
in probability where $D_n$ is a sample of size $n$ from the distribution of $(X,Y)$, provided $\omega(1) < B(n) \leq O(n^{1-\ep})$ for some $\ep > 0$.
\end{theorem}

Since the supremum of a matrix is uniformly continuous as a function on $m^\infty$ and can be realized as the limit of maxima of larger and larger segments of the matrix, this theorem yields our claim about $\popMIC$ as a corollary.
\begin{corollary}
$\MIC_B$ is a consistent estimator of $\popMIC$ provided $\omega(1) < B(n) \leq O(n^{1-\ep})$ for some $\ep > 0$.
\end{corollary}

We prove Theorem~\ref{thm:estimator} in Appendix~\ref{appendix:popmicproof} and provide here some intuition for why it should hold as well as a description of the obstacles that must be overcome in the proof.

To see why the theorem should hold, fix a random variable $(X,Y)$ and let $D$ be a sample of size $n$ from its distribution. It is known that, for a fixed grid $G$, $I(D|_G)$ is a consistent estimator of $I((X,Y)|_G)$ \citep{roulston1999estimating, paninski2003estimation}. We might therefore expect $I^*(D, k, \ell)$ to be a consistent estimator of $I^*((X, Y), k, \ell)$ as well. And if $I^*(D, k, \ell)$ is a consistent estimator of $I^*((X, Y), k, \ell)$, then we might expect the maximum of the sample characteristic matrix (which just consists of normalized $I^*$ terms) to be a consistent estimator of the supremum of the true characteristic matrix.

These intuitions turn out to be true, but there are two reasons they are non-trivial to prove. First, consistency for $I^*$ does not follow from abstract considerations since the maximum of an infinite set of estimators is not necessarily a consistent estimator of the supremum of the estimands\footnote{
If $\hat \theta_1, \ldots, \hat \theta_k$ is a finite set of estimators, then a union bound shows that the random variable $(\hat \theta_1(D), \ldots, \hat \theta_k(D))$ converges in probability to $(\theta_1, \ldots, \theta_k)$ with respect to the supremum metric. The continuous mapping theorem then gives the desired result. However, if the set of estimators is infinite, the union bound cannot be employed. And indeed, if we let $\theta_1 = \cdots = \theta_k = 0$, and let $\hat \theta_i(D_n) = i/n$ deterministically, then each $\hat \theta_i$ is a consistent estimator of $\theta_i$, but since the set $\{\hat \theta_1(D_n), \hat \theta_2(D_n), \ldots \} = \{1/n, 2/n, \ldots \}$ is unbounded, $\sup_i \hat \theta_i(D_n) = \infty$ for every $n$.
}. Second, consistency of $I^*$ alone does not suffice to show that the maximum of the sample characteristic matrix converges to $\popMIC$. In particular, if $B(n)$ grows too quickly, and the convergence of $I^*(D, k, \ell)$ to $I^*((X, Y), k, \ell)$ is slow, inflated values of $\MIC$ can result. To see this, notice that if $B(n) = \infty$ then $\MIC = 1$ always, even though each individual entry of the sample characteristic matrix converges to its true value eventually.

The technical heart of the proof is overcoming these obstacles by using the dependencies between the quantities $I(D|_G)$ for different grids $G$ to not only show the consistency of $I^*(D, k, \ell)$ but then to quantify how quickly $I^*(D, k, \ell)$ actually converges to $I^*((X, Y), k, \ell)$.

\subsection{Second alternate characterization: \texorpdfstring{$\popMIC$}{MIC*} is a minimally smoothed mutual information}
We now describe a second equivalent view of $\popMIC$. Recall that for a pair of jointly distributed random variables $(X,Y)$, we defined $\popMIC(X,Y)$ as
\[
\popMIC(X, Y) = \sup_G \frac{I((X,Y)|_G)}{\log \|G\| }
\]
where $\|G\|$ denotes the minimum of the number of rows of $G$ and the number of columns of $G$. As we discussed in Section~\ref{subsec:defining_popMIC}, the mutual information $I(X,Y)$ is also a supremum, namely
\[
I(X,Y) = \sup_G I((X,Y)|_G) .
\]
and so $\popMIC$ can be viewed as a regularized version of $I$. It is natural to ask whether the regularization in the definition of $\popMIC$ has any smoothing effect on $I$. In this sub-section we show first that it does, in the sense that $\popMIC$ is uniformly continuous as a function of random variables with respect to the metric of statistical distance\footnote{
Recall that the statistical distance between random variables $A$ and $B$ is defined as $\sup_T \left| \Pr{A \in T} - \Pr{B \in T} \right|$. When $A$ and $B$ have probability density functions or probability mass functions, this equals one-half of the $L^1$ distance between those functions.},
and second that the regularization by $\log \| G \|$ is in fact the minimal one necessary for achieving any sort of continuity. As a corollary, we obtain that $I$ by itself is not continuous as a function of random variables with respect to the metric of statistical distance. This yields a view of $\popMIC$ as a canonical smoothing of $I$ that yields continuity.

Formally, let $\P(\R^2)$ denote the space of random variables supported on $\R^2$ equipped with the metric of statistical distance. Our first claim is that as a function defined on $\P(\R^2)$, $\popMIC$ is uniformly continuous. We prove this claim by establishing a stronger result: the uniform continuity of the characteristic matrix $M(X,Y)$. Specifically, by showing that the family of maps corresponding to each individual entry of the characteristic matrix is uniformly equicontinuous, we establish the following result.

\begin{theorem}
\label{thm:continuityofM}
The map from $\P(\R^2)$ to $m^\infty$ defined by $(X,Y) \mapsto M(X,Y)$ is uniformly continuous.
\end{theorem}
\begin{proof}
See Appendix~\ref{app:continuityofM}.
\end{proof}

Since the supremum is a continuous function on $m^\infty$, Theorem~\ref{thm:continuityofM} yields the following corollary.

\begin{corollary}
The map $(X,Y) \mapsto \popMIC(X,Y)$ is uniformly continuous.
\end{corollary}
Similar corollaries exist for any continuous function of the characteristic matrix.

Interestingly, Theorem~\ref{thm:continuityofM} relies crucially on the normalization in the definition of the characteristic matrix. This is not a coincidence: as the following proposition shows, any normalization that is meaningfully smaller than the one in the definition of the characteristic matrix will cause the matrix to contain an infinite discontinuity as a function on $\P(\R^2)$.

\begin{proposition}
\label{prop:MIdiscontinuous}
For some function $N(k, \ell$), let $M^N$ be the characteristic matrix with normalization $N$, i.e.,
\[
M^N(X,Y) = \frac{I^*((X,Y), k, \ell)}{N(k, \ell)} .
\]
If $N(k, \ell) = o(\log \min \{k,\ell\})$ along some infinite path in $\mathbb{N} \times \mathbb{N}$, then $M^N$ and $\sup M^N$ are not continuous as functions of $\mathcal{P}([0,1]\times[0,1]) \subset \P(\R^2)$.
\end{proposition}
\begin{proof}
See Appendix~\ref{app:MIdiscontinuous}
\end{proof}

The above proposition implies that the ``smoothing'' that $\popMIC$ applies to mutual information is necessary in some sense. In particular, one corollary of the proposition is that mutual information with no smoothing will contain a disconuity.
\begin{corollary}
\label{cor:MIdiscontinuous}
Mutual information is not continuous on $\mathcal{P}([0,1]\times[0,1]) \subset \P(\R^2)$.
\end{corollary}
\begin{proof}
Mutual information is the supremum of $M^N$ with $N \equiv 1$.
\end{proof}
The same result can also be shown for the squared Linfoot correlation~\citep{speed2011correlation,linfoot1957informational}, which equals $1 - 2^{-2I}$ where $I$ represents mutual information. Thus, though the Linfoot correlation smoothes the mutual information enough to cause it to lie in the unit interval, it does not smooth the mutual information sufficiently to cause it to be continuous.

As we remarked previously, these results, when contrasted with the uniform continuity of $\popMIC$, allow us to view the latter as a canonical ``minimally smoothed'' version of mutual information that is uniformly continuous. This view gives a meaningful interpretation to the normalization used in $\popMIC$. Understanding $\popMIC$ as having smoothness properties not shared by mutual information also suggests that estimators of $\popMIC$ may have better statistical properties than estimators of ordinary mutual information. This is consistent with the hardness-of-estimation result in \cite{ding2013copula} and is also borne out empirically in \cite{reshef2015comparisons}.

\subsection{Third alternate characterization: \texorpdfstring{$\popMIC$}{MIC*} is the supremum of the boundary of the characteristic matrix}
We now show the third alternate view of $\popMIC$: that it can be equivalently defined as the supremum over a \textit{boundary} of the characteristic matrix rather than as a supremum over all of the entries of the matrix. This characterization of $\popMIC$ will serve as the foundation both for our approach to computing $\popMIC(X,Y)$ as well as the new estimator of $\popMIC$ that we introduce later in this paper.

We begin by defining what we mean by the boundary of the characteristic matrix. Our definition rests on the following observation.
\begin{proposition}
Let $M$ be a population characteristic matrix. Then for $\ell \geq k$, $M_{k,\ell} \leq M_{k, \ell+1}$.
\end{proposition}
\begin{proof}
Let $(X,Y)$ be the random variable in question. Since we can always let a row/column be empty, we know that $I^*((X,Y),k,\ell) \leq I^*((X,Y), k, \ell+1)$. And since $\ell, \ell+1 \geq k$, we know that $M_{k,\ell} = I^*((X,Y),k,\ell)/\log k \leq I^*((X,Y),k,\ell+1)/\log k = M_{k, \ell+1}$.
\end{proof}

Since the entries of the characteristic matrix are bounded, the monotone convergence theorem then gives the following corollary. In the corollary and henceforth, we let $M_{k,\uparrow} = \lim_{\ell \rightarrow \infty} M_{k,\ell}$ and define $M_{\uparrow, \ell}$ similarly.
\begin{corollary}
\label{cor:boundaryIsSup}
Let $M$ be a population characteristic matrix. Then $M_{k,\uparrow}$ exists, is finite, and equals $\sup_{\ell \geq k} M_{k,\ell}$. The same is true for $M_{\uparrow,\ell}$.
\end{corollary}

The above corollary allows us to define the {\em boundary} of the characteristic matrix.
\begin{definition}
\label{def:boundary}
Let $M$ be a population characteristic matrix. The {\em boundary} of $M$ is the set
\[
\partial M = \{ M_{k, \uparrow} : 1 < k < \infty \} \bigcup \{ M_{\uparrow, \ell} : 1 < \ell < \infty \}.
\]
\end{definition}

The theorem below then gives a relationship between the boundary of the characteristic matrix and $\popMIC$.
\begin{theorem}
\label{thm:MICinTermsOfBoundary}
Let $(X,Y)$ be a random variable. We have
\[ \popMIC(X,Y) = \sup \partial M(X,Y) \]
where $M(X,Y)$ is the population characteristic matrix of $(X, Y)$.
\end{theorem}
\begin{proof}
The following argument shows that every entry of $M$ is at most $\sup \partial M$: fix a pair $(k, \ell)$ and notice that either $k \leq \ell$, in which case $M_{k,\ell} \leq M_{k,\uparrow}$, or $\ell \leq k$, in which case $M_{k,\ell} \leq M_{\uparrow, \ell}$. Thus, $\popMIC \leq \sup \{M_{\uparrow, \ell}\} \cup \{M_{k, \uparrow}\} = \sup \partial M$.

On the other hand, Corollary~\ref{cor:boundaryIsSup} shows that each element of $\partial M$ is a supremum over some elements of $M$. Therefore, $\sup \partial M$, being a supremum over suprema of elements of $M$, cannot exceed $\sup M = \popMIC$.
\end{proof}

\subsection{Computing \texorpdfstring{$\popMIC$}{MIC*} efficiently}
\label{subsec:alg_infinite_data}
The importance of the characterization in Theorem~\ref{thm:MICinTermsOfBoundary} from the previous sub-section is computational. Specifically, elements of the boundary of the characteristic matrix can be expressed in terms of a maximization over (one-dimensional) partitions rather than (two-dimensional) grids, the former being much quicker to compute exactly. This is stated in the theorem below.
\begin{theorem}
\label{thm:explicitValueOfBoundary}
Let $M$ be a population characteristic matrix. Then $M_{k, \uparrow}$ equals
\[ \max_{P \in P(k)} \frac{I(X, Y|_P)}{\log{k}} \]
where $P(k)$ denotes the set of all partitions of size at most $k$.
\end{theorem}
\begin{proof}
See Appendix~\ref{app:explicitValueOfBoundary}.
\end{proof}

To formally state how this will help us compute $\popMIC$, we note that Theorems~\ref{thm:MICinTermsOfBoundary} and~\ref{thm:explicitValueOfBoundary} above together give the following corollary.

\begin{corollary}
\label{cor:altDef1}
Let $(X,Y)$ be a random variable, and let $\mathbb{P}$ be the set of finite-size partitions. Then
\[ \popMIC(X,Y) = \sup \left\{ \frac{I(X, Y|_P)}{\log |P|} : P \in \mathbb{P} \right\} \bigcup \left\{ \frac{I(X|_P, Y)}{\log |P|} : P \in \mathbb{P} \right\} \]
where $|P|$ is the number of bins in the partition $P$.
\end{corollary}

The expressions in the above corollary involve maximization only over one-dimensional partitions rather than two-dimensional grids. We can exploit this fact to give an algorithm for computing elements of the boundary of the characteristic matrix to arbitrary precision. To do so, we utilize as a subroutine a dynamic programming algorithm from \cite{MINE} called \algname{OptimizeXAxis}. Before continuing, we therefore give a brief overview of that algorithm.

\paragraph{Overview of \algname{OptimizeXAxis} algorithm from \cite{MINE}}
The \algname{OptimizeXAxis} algorithm takes as input a set $D$ of $n$ data points, a fixed partition into columns\footnote{
Despite its name, the \algname{OptimizeXAxis} algorithm can be used to optimize a partition of either axis. In our description of the algorithm here, we choose to describe the algorithm as it would work for optimizing a partition of the \textit{y-axis} rather than the x-axis. This is for notational coherence of this paper only.}
$Q$ of size $\ell$, a ``master'' partition into rows $\Pi$, and a number $k$. The algorithm returns, for $2 \leq i \leq k$, the partition into rows $P_i \subset \Pi$ that maximizes the mutual information of $D|_{(P_i, Q)}$ among all sub-partitions of $\Pi$ of size at most $i$. The algorithm works by exploiting the fact that, conditioned on the location $y$ of the top-most line of $P_i$, the optimization of the rest of $P_i$ can be formulated as a sub-problem that depends only on the data points below $y$. The algorithm uses dynamic programming to store and reuse solutions to these subproblems, resulting in a runtime of $O(|\Pi|^2k\ell)$. If a black-box algorithm is used to compute each required mutual information in time at most $T$, then the runtime of the algorithm can be shown to be $O(T k |\Pi|)$.

The following theorem shows that the theory developed about the boundary of the characteristic matrix, together with \algname{OptimizeXAxis}, yields an efficient algorithm for computing entries of the boundary to arbitrary precision.

\begin{theorem}
\label{thm:alg_infinite_data}
Given a random variable $(X,Y)$, $M_{k, \uparrow}$ (resp. $M_{\uparrow, \ell}$) is computable to within an additive error of $O(k \ep \log(1/(k\ep))) + E$ (resp. $O(\ell \ep \log(1/(\ell\ep))) + E$) in time $O(kT(E)/\ep)$ (resp. $O(\ell T(E) / \ep)$), where $T(E)$ is the time required to numerically compute the mutual information of a continuous distribution to within an additive error of $E$.
\end{theorem}
\begin{proof}
See Appendix~\ref{app:alg_infinite_data}.
\end{proof}

The algorithm proposed in Theorem~\ref{thm:alg_infinite_data} gives us a polynomial-time method for computing any finite subset of the boundary $\partial M$ of the population characteristic matrix $M(X,Y)$ of a random variable $(X,Y)$. Thus, if we have some $k_0, \ell_0$ such that the maximum of the finite subset $\{ M_{k, \uparrow}, M_{\uparrow, \ell} : k \leq k_0, \ell \leq \ell_0 \}$ of $\partial M$ will be $\ep$-close to the supremum of the entire set $\partial M$, we can compute $\popMIC(X,Y)$ to within an error of $\ep$. Though we usually do not have precise knowledge of $k_0$ and $\ell_0$, for simple distributions it is often easy to make very conservative educated guesses for them. This algorithm therefore allows us to approximate $\popMIC(X,Y)$ very well in practice.

Being able to compute $\popMIC(X,Y)$ has two main advantages. The first is that it allows us to assess in simulations the large-sample properties of $\popMIC$ independent of any estimator. This is done in the companion paper \citep{reshef2015comparisons}, which shows that $\popMIC$ achieves high equitability with respect to $R^2$ on a set of noisy functional relationships thereby confirming that statistically efficient estimation of $\popMIC$ is a worthwhile goal.

The second advantage of being able to compute $\popMIC(X,Y)$ is that we can empirically assess the bias, variance, and expected squared error of estimators of $\popMIC$ by taking a distribution, computing $\popMIC$, and then comparing the result to estimates of it based on finite samples. In the next section, we introduce a new estimator $\MICestE$ of $\popMIC$ and carry out such an analysis to compare its statistical properties to those of the statistic $\MIC$ from \cite{MINE}.

\section{Estimating \texorpdfstring{$\popMIC$}{MIC*} with \texorpdfstring{$\MICestE$}{MICe}}
As we have shown, $\popMIC$ is actually the population value of the statistic $\MIC$ introduced in \cite{MINE}. However, though consistent, the statistic $\MIC$ is not known to be efficiently computable and in \cite{MINE} a heuristic approximation algorithm called \algname{Approx-MIC} was computed instead. In this section, we leverage the theory we have developed here to introduce a new estimator of $\popMIC$ that is both consistent and efficiently computable. The new estimator, called $\MICestE$, actually has better runtime complexity even than the heuristic \algname{Approx-MIC} algorithm, and runs orders of magnitude faster in practice.

The estimator $\MICestE$ is based on one of the alternate characterizations of $\popMIC$ proven in the previous section. Namely, if $\popMIC$ can be viewed as the supremum of the {\em boundary} of the characteristic matrix rather than of the entire matrix, then only the boundary of the matrix must be accurately estimated in order to estimate $\popMIC$. This has the advantage that, whereas computing individual entries of the sample characteristic matrix involves finding optimal (two-dimensional) grids, estimating entries of the boundary requires us only to find optimal (one-dimensional) partitions. While the former problem is computationally difficult, the latter can be solved using the dynamic programming algorithm from \cite{MINE} that we also employed in Section~\ref{subsec:alg_infinite_data} to compute $\popMIC$ in the large-sample limit.

We formalize this idea via a new object called the {\em equicharacteristic matrix}, which we denote by $[M]$. The difference between $[M]$ and the characteristic matrix $M$ is as follows: while the $k, \ell$-th entry of $M$ is computed from the maximal achievable mutual information using any $k$-by-$\ell$ grid, the $k, \ell$-th entry of $[M]$ is computed from the maximal achievable mutual information using any $k$-by-$\ell$ grid that equipartitions the dimension with more rows/columns. (See Figure~\ref{fig:equicharmatrix_schematic}.) Despite this difference, as the equipartition in question gets finer and finer it becomes indistinguishable from an optimal partition of the same size. This intuition can be formalized to show that the boundary of $[M]$ equals the boundary of $M$, and therefore that $\sup [M] = \sup M = \popMIC$. It will then follow that estimating $[M]$ and taking the supremum \--- as we did with $M$ in the case of $\MIC$ \--- yields a consistent estimate of $\popMIC$.

\subsection{The equicharacteristic matrix}
We now define the equicharacteristic matrix and show that its supremum is indeed $\popMIC$. To do so, we first define a version of $I^*$ that equipartitions the dimension with more rows/columns. Note that in the definition, brackets are used to indicate the presence of an equipartition.

\begin{figure}[t]
\centering
\begin{tabular}{ccc}
\includegraphics[clip=true, trim = 0in 8.85in 6.35in 0in, width=1in]{\pathToCommonFigs/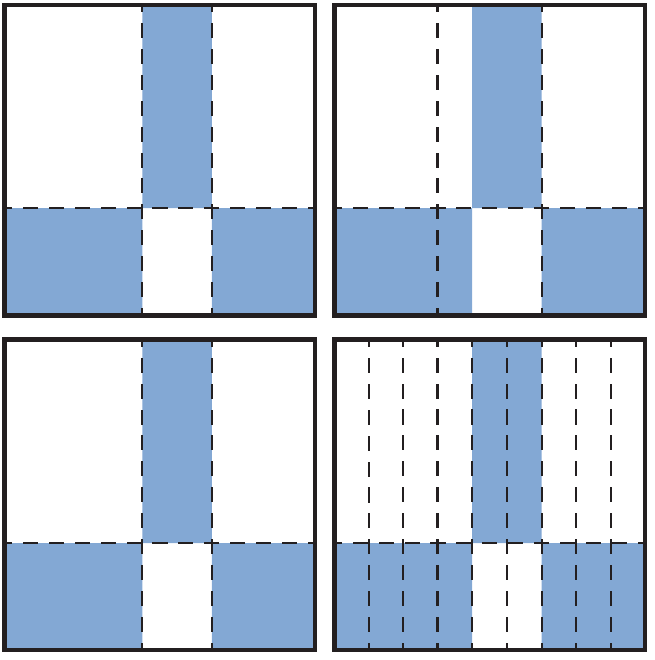} & & \includegraphics[clip=true, trim = 2.1875in 8.85in 4.15in 0in, width=1in]{\pathToCommonFigs/Equicharmatrix_Figure.pdf} \\
$M(X,Y)_{2,3}$ & & $[M](X,Y)_{2,3}$ \\
$I^* = 0.918$   & & $I^{[*]} = 0.613$ \\
\hfill & \hfill \\
\includegraphics[clip=true, trim = 0in 6.625in 6.35in 2.1875in, width=1in]{\pathToCommonFigs/Equicharmatrix_Figure.pdf} & & \includegraphics[clip=true, trim = 2.1875in 6.625in 4.15in 2.1875in, width=1in]{\pathToCommonFigs/Equicharmatrix_Figure.pdf} \\
$M(X,Y)_{2,9}$ & & $[M](X,Y)_{2,9}$ \\
$I^* = 0.918$   & & $I^{[*]} = 0.918$
\end{tabular}
\caption{A schematic illustrating the difference between the characteristic matrix $M$ and the equicharacteristic matrix $[M]$.
\figpart{Top} When restricted to 2 rows and 3 columns, the characteristic matrix $M$ is computed from the optimal 2-by-3 grid. In contrast, the equicharacteristic matrix $[M]$ still optimizes the smaller partition of size 2 but is restricted to have the larger partition be an equipartition of size 3. This results in a lower mutual information of $0.613$.
\figpart{Bottom} When 9 columns are allowed instead of 3, the grid found by the equicharacteristic matrix does not change, since the grid with 3 columns was already optimal. However, now the equicharacteristic matrix uses an equipartition into columns of size 9, whose resolution is able to fully capture the dependence between $X$ and $Y$.}
\label{fig:equicharmatrix_schematic}
\end{figure}

\begin{definition}
Let $(X,Y)$ be jointly distributed random variables. Define
\[ I^* \left( (X,Y), k, [\ell] \right) = \max_{G \in G(k, [\ell])} I \left( (X,Y)|_G \right) \]
where $G(k,[\ell])$ is the set of $k$-by-$\ell$ grids whose y-axis partition is an equipartition of size $\ell$. Define $I^*((X,Y), [k], \ell)$ analogously.

Define $I^{[*]}((X,Y), k, \ell)$ to equal $I^*((X,Y), k, [\ell])$ if $k < \ell$ and $I^*((X,Y), [k], \ell)$ otherwise.
\end{definition}

We now define the equicharacteristic matrix in terms of $I^{[*]}$. In the definition below, we continue our convention of using brackets around a quantity to denote the presence of equipartitions.
\begin{definition}
Let $(X,Y)$ be jointly distributed random variables. The {\em population equicharacteristic matrix} of $(X,Y)$, denoted by $[M](X,Y)$, is defined by
\[ [M](X,Y)_{k,\ell} = \frac{I^{[*]}((X,Y), k, \ell)}{\log \min \{k, \ell\}} \]
for $k, \ell > 1$.
\end{definition}

The boundary of the equicharacteristic matrix can be defined via a limit in the same way as the characteristic matrix. We then have the following theorem.

\begin{theorem}
\label{thm:same_boundary}
Let $(X,Y)$ be jointly distributed random variables. Then $\partial [M] = \partial M$.
\end{theorem}
\begin{proof}
See Appendix~\ref{app:same_boundary}.
\end{proof}

Since every entry of the equicharacteristic matrix is dominated by some entry on its boundary, the equivalence of $\partial [M]$ and $\partial M$ yields the following corollary as a simple consequence.
\begin{corollary}
\label{cor:boundary_equichar}
Let $(X,Y)$ be jointly distributed random variables. Then $\sup [M](X,Y) = \popMIC(X,Y)$.
\end{corollary}

\subsection{The estimator \texorpdfstring{$\MICestE$}{MICe}}
With the equicharacteristic matrix defined, we can now define our new estimator $\MICestE$ in terms of the sample equicharacteristic matrix, analogously to the way we defined $\MIC$ in terms of the sample characteristic matrix.

\begin{definition}
Let $D \subset \R^2$ be a set of ordered pairs. The {\em sample equicharacteristic matrix} $\widehat{[M]}(D)$ of $D$ is defined by
\[
\widehat{[M]}(D)_{k,\ell} = \frac{I^{[*]}(D, k, \ell)}{ \log \min \{k, \ell \}} .\
\]
\end{definition}

\begin{definition}
Let $D \subset \R^2$ be a set of $n$ ordered pairs, and let $B : \Z^+ \rightarrow \Z^+$. We define
\[ \MICestE_{,B}(D) = \max_{k\ell \leq B(n)} \widehat{[M]}(D)_{k,\ell} .\]
\end{definition}

With the equivalence between the boundaries of the characteristic matrix and the equicharacteristic matrix established, it is straightforward to show that $\MICestE$ is a consistent estimator of $\popMIC$ via arguments similar to those we applied in the case of $\MIC$. (See Appendix~\ref{app:MICeconsistent}.) Specifically, we show the following theorem, an analogue of Theorem~\ref{thm:estimator}.

\begin{theorem}
\label{thm:new_estimator}
Let $f : m^\infty \rightarrow \R$ be uniformly continuous, and assume that $f \circ r_i \rightarrow f$ pointwise. Then for every random variable $(X,Y)$, we have
\[
\left( f \circ r_{B(n)} \right) \left(\widehat{[M]}(D_n) \right) \rightarrow f([M](X,Y))
\]
in probability where $D_n$ is a sample of size $n$ from the distribution of $(X,Y)$, provided $\omega(1) < B(n) \leq O(n^{1-\ep})$ for some $\ep > 0$.
\end{theorem}

By setting $f([M]) = \sup [M]$, we then obtain as a corollary the consistency of $\MICestE$.
\begin{corollary}
\label{cor:mice_consistent}
$\MICestE_{,B}$ is a consistent estimator of $\popMIC$ provided $\omega(1) < B(n) \leq O(n^{1-\ep})$ for some $\ep > 0$.
\end{corollary}

\subsubsection{Choosing \texorpdfstring{$B(n)$}{B(n)}}
\label{subsubsec:choosingB}
As with the statistic $\MIC$, the statistic $\MICestE$ requires the user to specify a function $B(n)$ to use. While the theory suggests that any function of the form $B(n) = n^\alpha$ suffices provided $0 < \alpha < 1$, different values of $\alpha$ may yield different finite-sample properties. We study the empirical performance of $\MICestE$ for different choices of $B(n)$ in Section~\ref{sec:bias_variance}.

\cite{reshef2015comparisons} provides simple, empirical recommendations about appropriate values of $\alpha$ for different settings. Those recommendations are formulated by choosing a set of representative relationships (e.g., a set of noisy functional relationships), as well as a ``ground truth'' population quantity $\Phi$ (e.g., $R^2$) that can be used to quantify the strength of each of those relationships, and then assessing which values of $\alpha$ maximize the equitability of $\MICestE$ with respect to $\Phi$ at a given sample size. This approach is applied to an analysis of real data from the World Health Organization in \cite{reshef2015comparisons}, and the parameters chosen for that analysis are the ones used for all analyses in this paper.

We remark that if the goal of the user is only detection of non-trivial relationships rather discovery of the strongest such relationships, $\alpha$ can also be chosen in a more straightforward manner: the user can sub-sample a small random set of relationships on which to compare the power of $\MICestE$ for different values of $\alpha$. Those relationships can then be discarded and the rest of the relationships analyzed with the optimal value of $\alpha$. However, if the user's primary goal is power against independence, the statistic $\TICestE$ introduced in Section~\ref{sec:TIC} of this paper should be used with this strategy rather than $\MICestE$.

\subsection{Computing \texorpdfstring{$\MICestE$}{MICe}}
\label{subsec:computing_MICe}
Both $\MIC$ and $\MICestE$ are consistent estimators of $\popMIC$. The difference between them is that while $\MIC$ can currently be computed efficiently only via a heuristic approximation, $\MICestE$ can be computed exactly very efficiently via an approach similar to the one used for computing $\popMIC$ involving the \algname{OptimizeXAxis} subroutine. We now describe the details of this approach.

Recall that, given a fixed x-axis partition $Q$ into $\ell$ columns, a set of $n$ data points, a ``master'' y-axis partition $\Pi$, and a number $k$, the \algname{OptimizeXAxis} subroutine finds, for every $2 \leq i \leq k$, a y-axis partition $P_i \subset \Pi$ of size at most $i$ that maximizes the mutual information induced by the grid $(P_i,Q)$. The algorithm does this in time $O(|\Pi|^2 k\ell)$. (For more discussion of \algname{OptimizeXAxis}, see Section~\ref{subsec:alg_infinite_data}, where it is also used to give an algorithm for computing $\popMIC$.)

In the pair of theorems below, we show two ways that \algname{OptimizeXAxis} can be used to compute $\MICestE$ efficiently. In the proofs of both theorems, we neglect issues of divisibility, i.e., we often write $B/2$ rather than $\lfloor B / 2 \rfloor$. This does not affect the results.
\begin{theorem}
\label{thm:algorithm_naive}
There exists an algorithm \algname{Equichar} that, given a sample $D$ of size $n$ and some $B \in \Z^+$, computes the portion $r_{B(n)}(\widehat{[M]}(D))$ of the sample equicharacteristic matrix in time $O(n^2B^2)$, which equals $O(n^{4 - 2\ep})$ for $B(n) = O(n^{1-\ep})$ with $\ep > 0$.
\end{theorem}
\begin{proof}
We describe the algorithm and simultaneously bound its runtime. We do so only for the $k,\ell$-th entries of $\widehat{[M]}(D)$ satisfying $k \leq \ell, k\ell \leq B$. This suffices, since by symmetry computing the rest of the required entries at most doubles the runtime.

To compute $\widehat{[M]}(D)_{k,\ell}$ with $k \leq \ell$, we must fix an equipartition into $\ell$ columns on the x-axis and then find the optimal partition of the y-axis of size at most $k$. If we set the master partition $\Pi$ of the \algname{OptimizeXAxis} algorithm to be an equipartition into rows of size $n$, then it performs precisely the required optimization. Moreover, for fixed $\ell$ it can carry out the optimization simultaneously for all of the pairs $\{(2, \ell), \ldots, (B/\ell, \ell) \}$ in time $O(|\Pi|^2 (B/\ell) \ell) = O(n^2 B)$. For fixed $\ell$, this set contains all the pairs $(k, \ell)$ satisfying $k \leq \ell, kl \leq B$. Therefore, to compute all the required entries of $\widehat{[M]}(D)$ we need only apply this algorithm for each $\ell = 2, \ldots, B/2$. Doing so gives a runtime of $O(n^2 B^2)$.
\end{proof}

The algorithm above, while polynomial-time, is nonetheless not efficient enough for use in practice. However, a simple modification solves this problem without affecting the consistency of the resulting estimates. The modification hinges on the fact that \algname{OptimizeXAxis} can use master partitions $\Pi$ besides the equipartition of size $n$ that we used above. Spefically, setting $\Pi$ in the above algorithm to be an equipartition into $c k$ ``clumps'', where $k$ is the size of the largest optimal partition being sought, speeds up the computation significantly. This modification does give a slightly different statistic. However, the result is still a consistent estimator of $\popMIC$ because the size of the master partition $\Pi$ grows as $k$ grows, and so the optimal sub-partition of $\Pi$ approaches the true optimal partition eventually. These ideas, first about improved runtime and second about preserved consistency, are formalized in the following theorem.

\begin{theorem}
\label{thm:algorithm_with_c}
Let $(X,Y)$ be a pair of jointly distributed random variables, and let $D_n$ be a sample of size $n$ from the distribution of $(X,Y)$. For every $c \geq 1$, there exists a matrix $\{\widehat{M}\}^c(D_n)$ such that
\begin{enumerate}
\item The function
\[
\widetilde{\MICestE}_{,B}(\cdot) = \max_{k\ell \leq B(n)} \{\widehat{M}\}^c(\cdot)_{k,\ell}
\]
is a consistent estimator of $\popMIC$ provided $\omega(1) < B(n) \leq O(n^{1-\ep})$ for some $\ep >0$.
\item There exists an algorithm \algname{EquicharClump} for computing $r_B(\{\widehat{M}\}^c(D_n))$ in time $O(n + B^{5/2})$, which equals $O(n + n^{5(1-\ep)/2})$ when $B(n) = O(n^{1-\ep})$.
\end{enumerate}
\end{theorem}
\begin{proof}
See Appendix~\ref{app:computation}.
\end{proof}

For an analysis of the effect of the parameter $c$ in the above theorem on the results of the \algname{EquicharClump} algorithm, see Appendix~\ref{app:empirical_char_of_c}.

Setting $\ep = 0.6$ in the above theorem yields the following corollary.
\begin{corollary}
$\popMIC$ can be estimated consistently in linear time.
\end{corollary}
Of course, at low sample sizes, setting $\ep=0.6$ would be undesirable. However, our companion paper \citep{reshef2015comparisons} shows empirically that at large sample sizes this strategy works very well on typical relationships.

We remark that the \algname{EquicharClump} algorithm given above is asymptotically faster even than the heuristic \algname{Approx-MIC} algorithm used to calculate $\MIC$ in practice, which runs in time $O(B(n)^4)$. As demonstrated in our companion paper \citep{reshef2015comparisons}, this difference translates into a substantial difference in runtimes for similar performance at a range of realistic sample sizes, ranging from a $30$-fold speedup at $n=500$ to over a $350$-fold speedup at $n=10,000$.

For readability, in the rest of this paper we do not distinguish between the two versions of $\MICestE$ computed by the \algname{Equichar} and \algname{EquicharClump} algorithms described above. Wherever we present simulation data about $\MICestE$ in simulations though, we use the version of the statistic computed by \algname{EquicharClump}.

\begin{figure}[h!]
	(a) \vspace{-0.5cm} \\
	\hspace*{\fill} \includegraphics[clip=true, trim = 0in 6.45in 0in 0in, height=0.36\textheight]{\pathToFigures/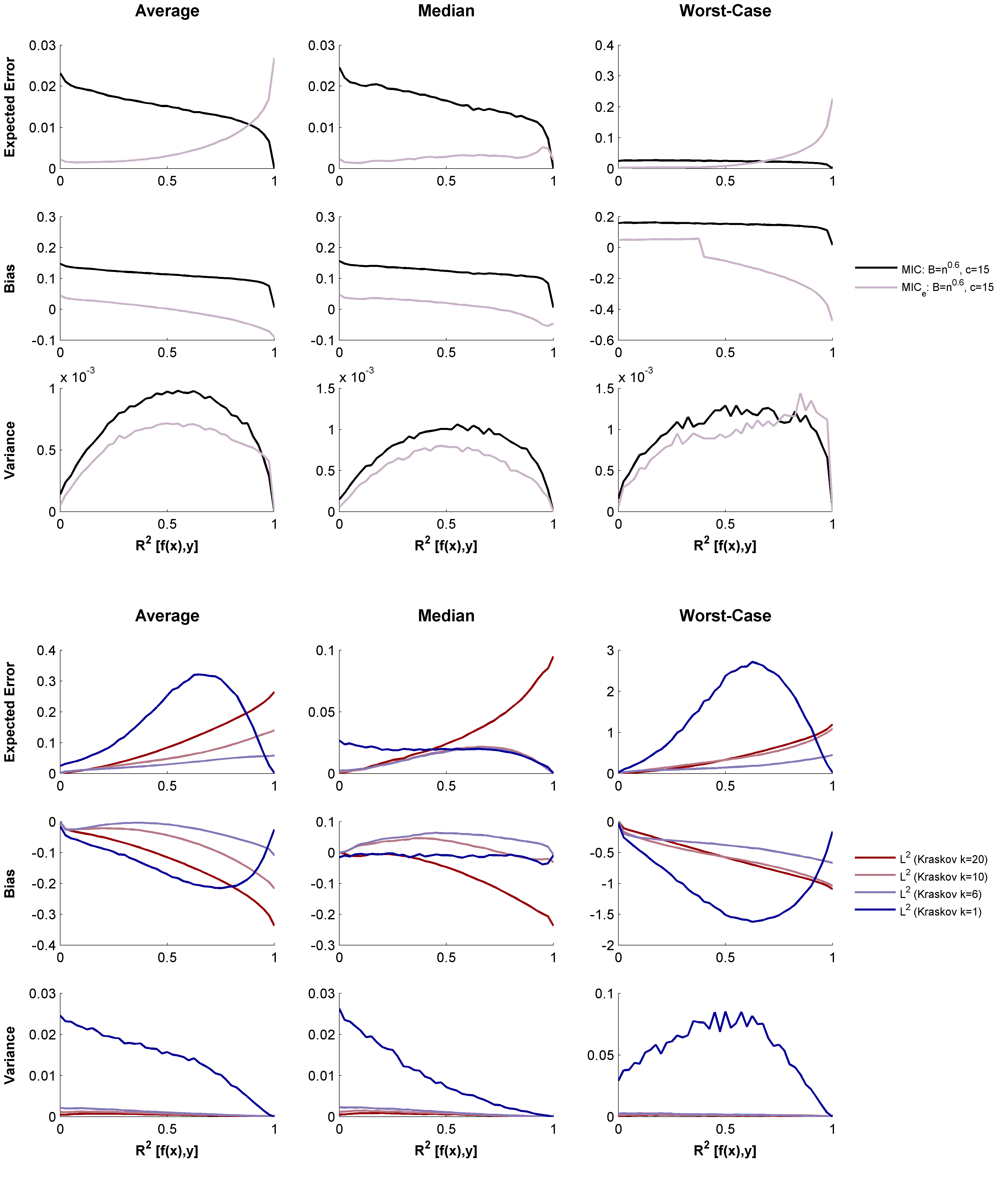} \\
	\vspace{0cm} \\
	(b) \vspace{-0.5cm} \\
	\hspace*{\fill} \includegraphics[clip=true, trim = 0in 6.45in 0in 0in, height=0.36\textheight]{\pathToFigures/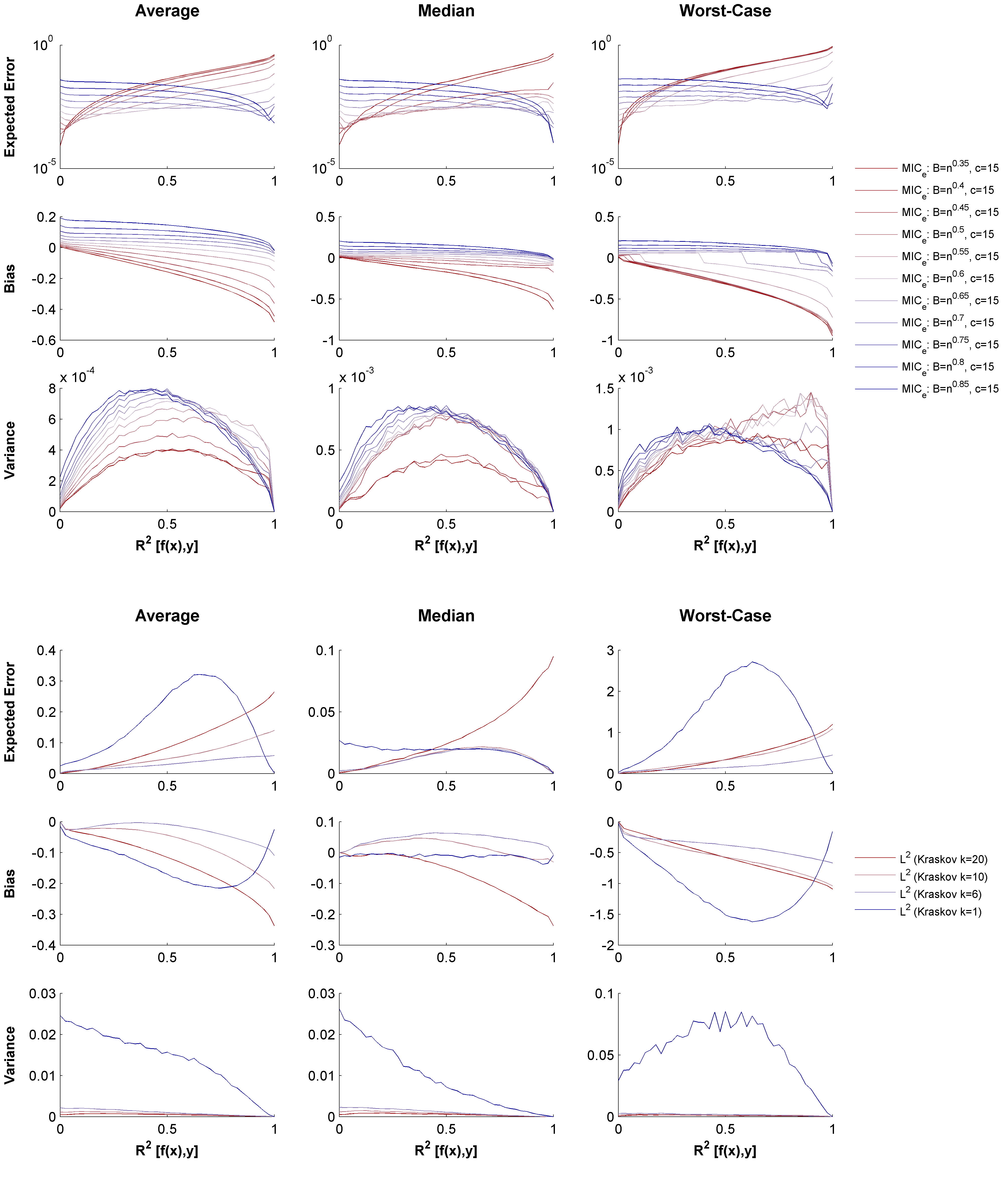}

	\caption[Bias and variance of \algname{Approx-MIC} and $\MICestE$]{Bias/variance characterization of \algname{Approx-MIC} and $\MICestE$.
Each plot shows expected squared error, bias, or variance across the set of noisy functional relationships described in Section~\ref{sec:bias_variance} as a function of the $R^2$ of the relationships. The results are aggregated across the 16 function types analyzed by either the average, median, or worst result at every value of $R^2$.
\figpart{a} A comparison between $\MICestE$ (light purple) and $\MIC$ (black) as computed via the heuristic \algname{Approx-MIC} algorithm, at a typical usage parameter.
\figpart{b} Performance of $\MICestE$ with $B(n) = n^\alpha$ for various values of $\alpha$.}
	\label{fig:biasVarianceMICe_evenCurve_XYNoise}
\end{figure}

\subsection{Bias/variance characterization of \texorpdfstring{$\MICestE$}{MICe}}
\label{sec:bias_variance}
The algorithm we presented in Section~\ref{subsec:alg_infinite_data} for computing $\popMIC$ in the large-sample limit allows us to examine the bias/variance properties of estimators of $\popMIC$. Here, we use it to examine the bias and variance of both $\MIC$ as computed by the heuristic \algname{Approx-MIC} algorithm from \cite{MINE}, and $\MICestE$ as computed by the \algname{EquicharClump} algorithm given above. To do this, we performed a simulation analysis on the following set of relationships
\[
\Q = \{ (x + \ep_\sigma, f(x) + \ep'_\sigma) : x \in X_f, \ep_\sigma, \ep'_\sigma \sim \mathcal{N}(0, \sigma^2), f \in F, \sigma \in \R_{\geq 0} \}
\]
where $\ep_\sigma$ and $\ep'_\sigma$ are i.i.d., $F$ is the set of 16 functions analyzed in \cite{MINE}, and $X_f$ is the set of $n$ x-values that result in the points $(x_i, f(x_i))$ being equally spaced along the graph of $f$.

For each relationship $\mcZ \in \Q$ that we examined, we used the algorithm from Theorem~\ref{thm:alg_infinite_data} to compute $\popMIC$. We then simulated 500 independent samples from $\mcZ$, each of size $n=500$, and computed both $\algname{Approx-MIC}$ and $\MICestE$ on each one to obtain estimates of the sampling distributions of the two statistics. From each of the two sampling distributions, we estimated the bias and variance of either statistic on $\mcZ$. We then analysed the bias, variance, and expected squared error of the two statistics as a function of relationship strength, which we quantified using the coefficient of determination ($R^2$) with respect to the generating function.

The results, presented in Figure~\ref{fig:biasVarianceMICe_evenCurve_XYNoise}, are interesting for two reasons. First, they demonstrate that for a typical usage parameter of $B(n) = n^{0.6}$, $\MICestE$ performs substantially better than \algname{Approx-MIC} overall. Specifically, the median of the expected squared error of $\MICestE$ across the set $F$ of functions is uniformly lower across $R^2$ values than that of \algname{Approx-MIC}. When average expected squared error is used instead of median, $\MICestE$ still performs better on all but the strongest of relationships ($R^2$ above $\sim$0.9). The superior performance of $\MICestE$ is consistent with the fact that we have theoretical guarantees about its statistical properties whereas \algname{Approx-MIC} is a heuristic.

Second, the results show that different values of the exponent in $B(n) = n^\alpha$ give good performance in different signal-to-noise regimes due to a bias-variance trade-off represented by this parameter. Large values of $\alpha$ lead to increased expected error in lower-signal regimes (low $R^2$) through both a positive bias in those regimes and a general increase in variance that predominantly affects those regimes. On the other hand, small values of $\alpha$ lead to an increased expected error in higher-signal regimes (high $R^2$) by leading to a negative bias in those regimes and by shifting the variance of the estimator toward those regimes. In other words, lower values of $\alpha$ are better-suited for detecting weaker signals, while higher values of $\alpha$ are better suited for distinguishing among stronger signals. This is consistent with the results seen in our companion paper \citep{reshef2015comparisons}, which show that low values of $\alpha$ cause $\MICestE$ to yield better powered independence tests while high values of $\alpha$ cause $\MICestE$ to have better equitability. For a detailed discussion of this trade-off and of specific recommendations for how to set $\alpha$ in practice, see \cite{reshef2015comparisons}.

\subsection{Equitability of \texorpdfstring{$\MICestE$}{MICe}}
\label{sec:equitability_analysis}
As mentioned previously, one of the main motivations for the introduction of $\MIC$ was equitability, the extent to which a measure of dependence usefully captures some notion of relationship strength on some set of standard relationships. We therefore carried out an empirical analysis of the equitability of $\MICestE$ with respect to $R^2$ and compared its performance to distance correlation \citep{szekely2007measuring, szekely2009brownian}, mutual information estimation \citep{Kraskov}, and maximal correlation estimation \citep{breiman1985estimating}.

We began by assessing equitability on the set of relationships $\Q$ defined above, a set that has been analyzed in previous work \citep{MINE, reshef2015comparisons, kinney2014equitability}. The results, shown in Figure~\ref{fig:equitabilityAnalysis}, confirm the superior equitability of the new estimator $\MICestE$ on this set of relationships.

\begin{figure}
    \centering
    \begin{tabular}{ccc}
    \includegraphics[clip=true, trim = 0in 0in 0.5in 0in, height=1.7in]{\pathToFigures/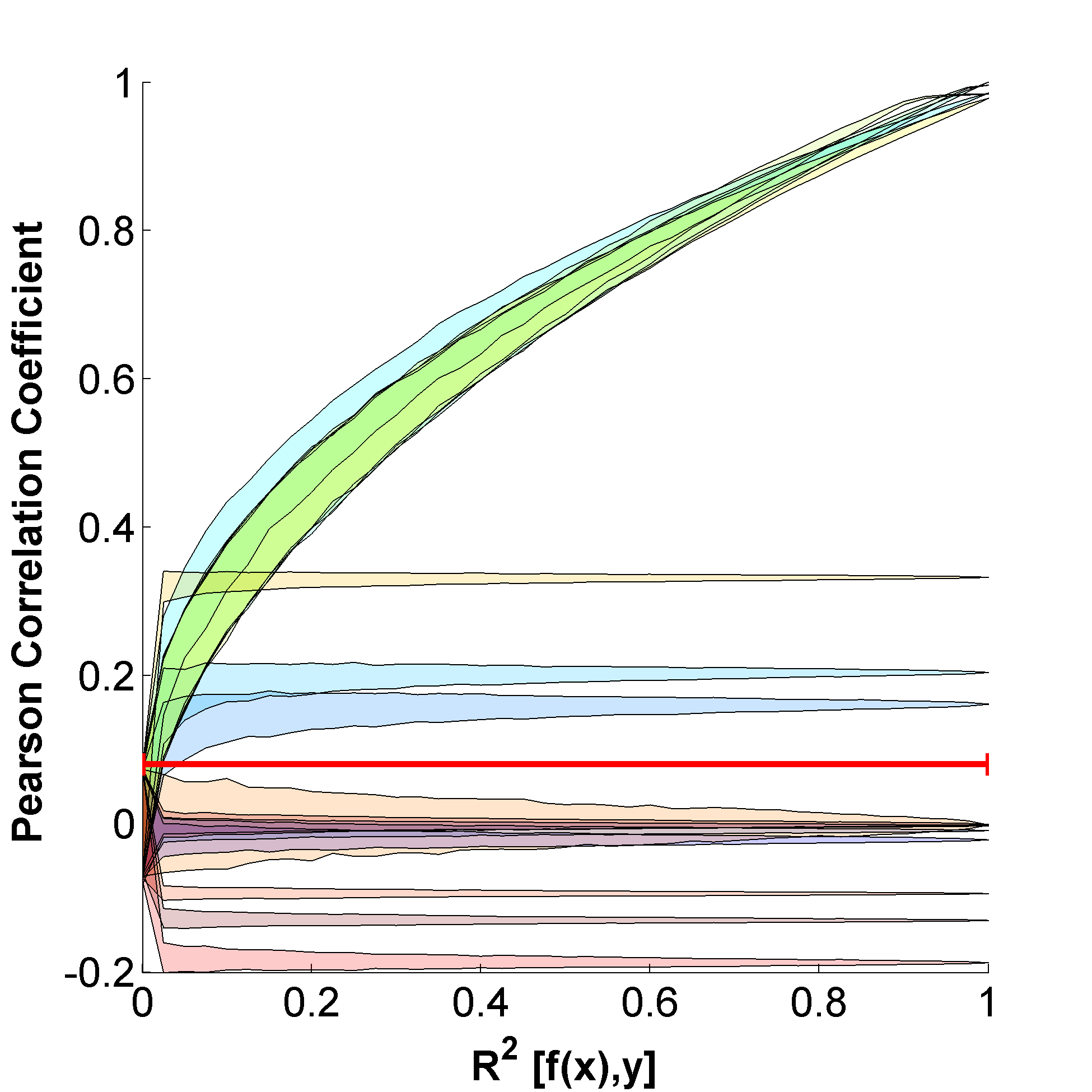} & & \includegraphics[clip=true, trim = 4.4in 6.125in 0in 2.2in, height=1.7in]{\pathToCommonFigs/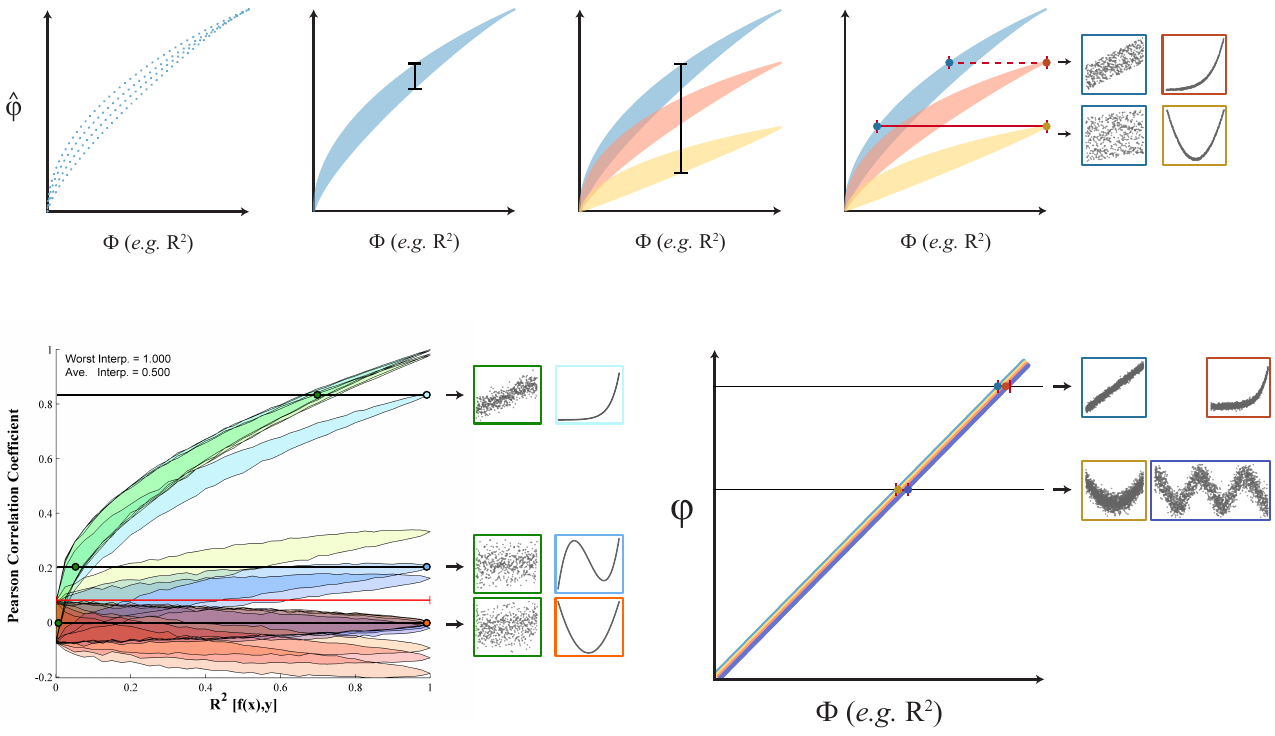} \hspace{0.275in} \\
    \hspace{0.17in} (a) & & (b) \hspace{1.25in} \\
    
    \includegraphics[clip=true, trim = 0in 0in 0.5in 0in, height=1.7in]{\pathToFigures/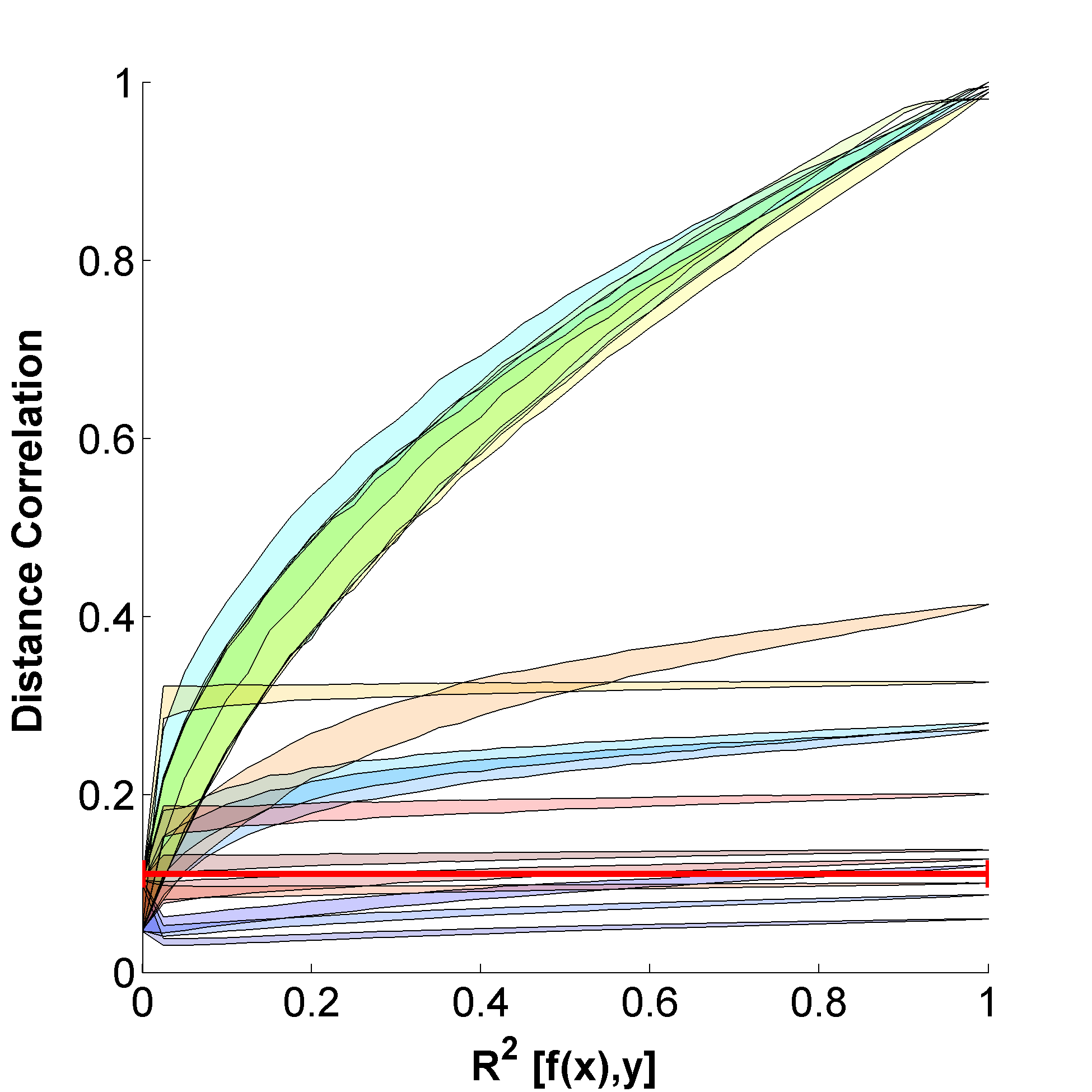} & & \includegraphics[clip=true, trim = 0in 0in 0.5in 0in, height=1.7in]{\pathToFigures/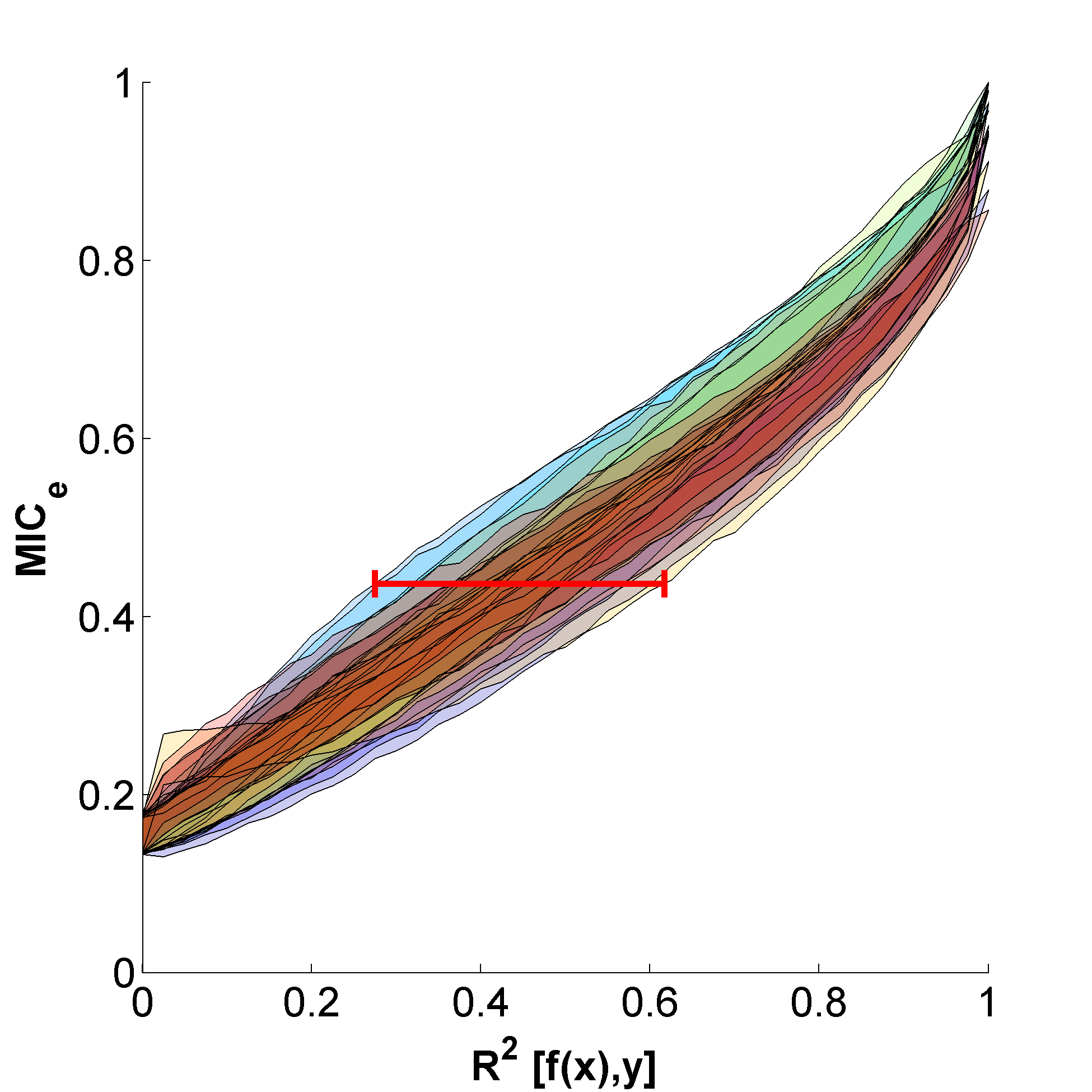} \quad \includegraphics[clip=true, trim = 4.125in 1.35in 1.75in 1.675in,height=1.7in]{\pathToFigures/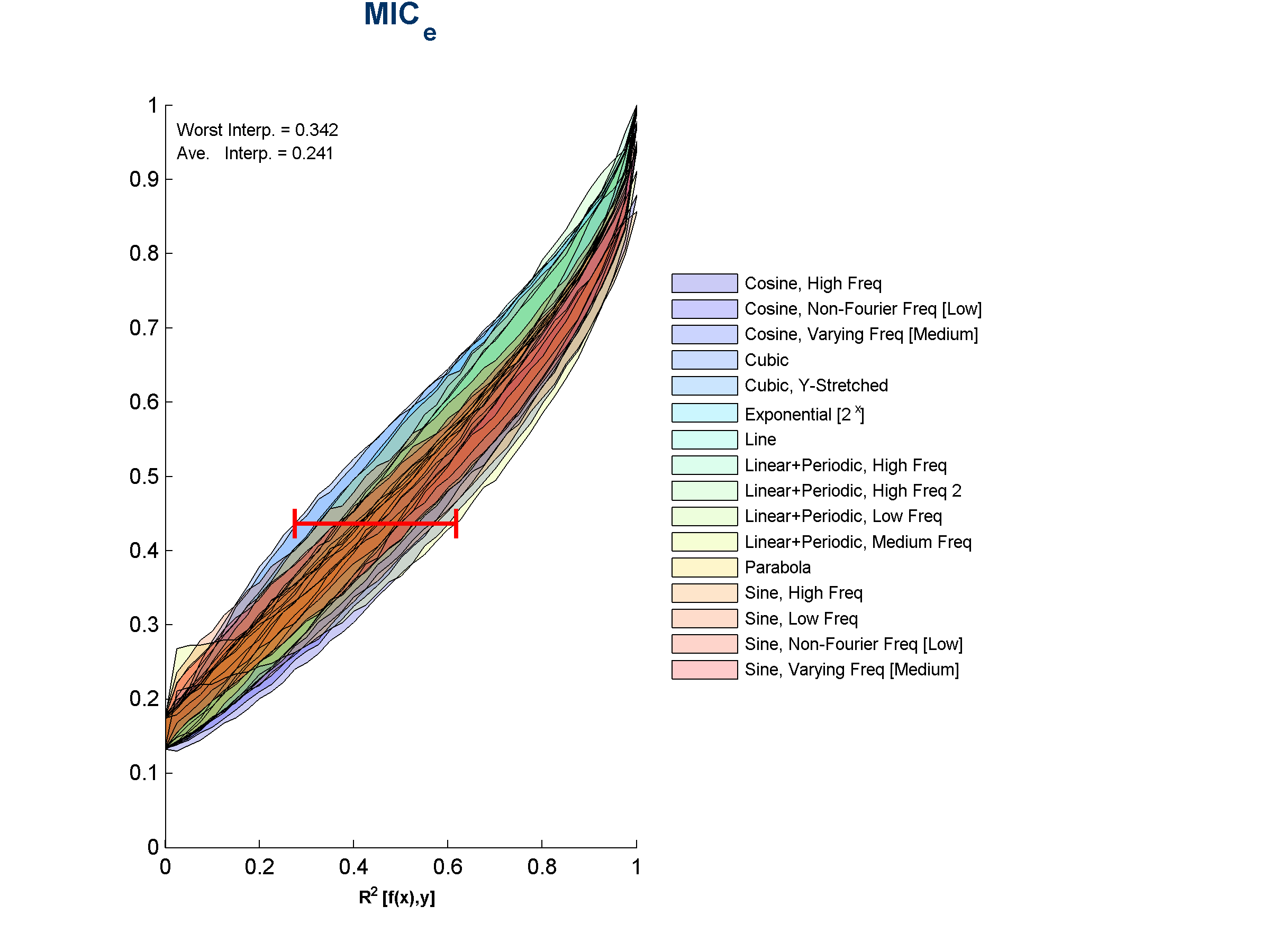} \\
    \hspace{0.17in} (c) & & (d) \hspace{1.25in} \\
    
    \includegraphics[clip=true, trim = 0in 0in 0.5in 0in, height=1.7in]{\pathToFigures/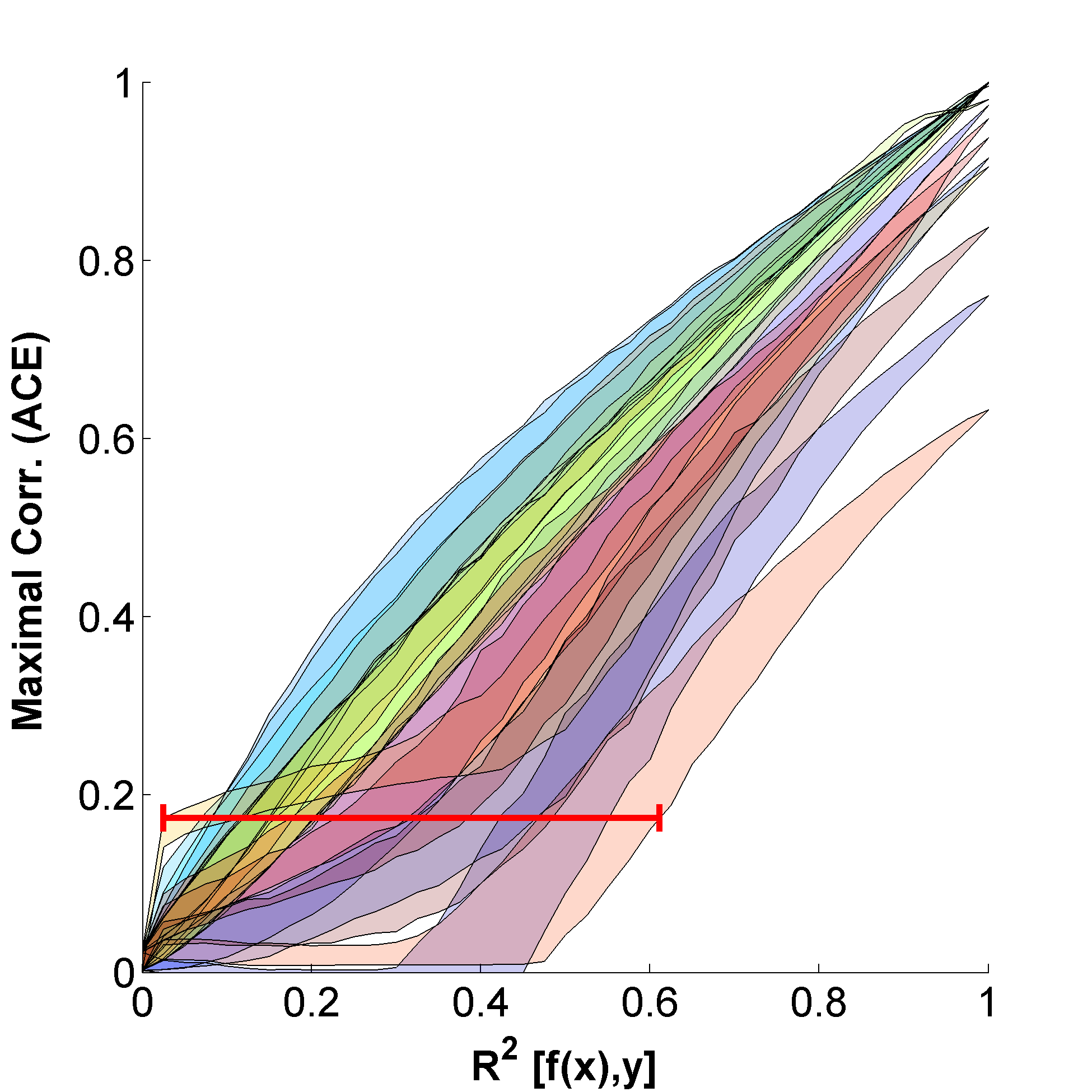} & & \includegraphics[clip=true, trim = 0in 0in 0.5in 0in, height=1.7in]{\pathToFigures/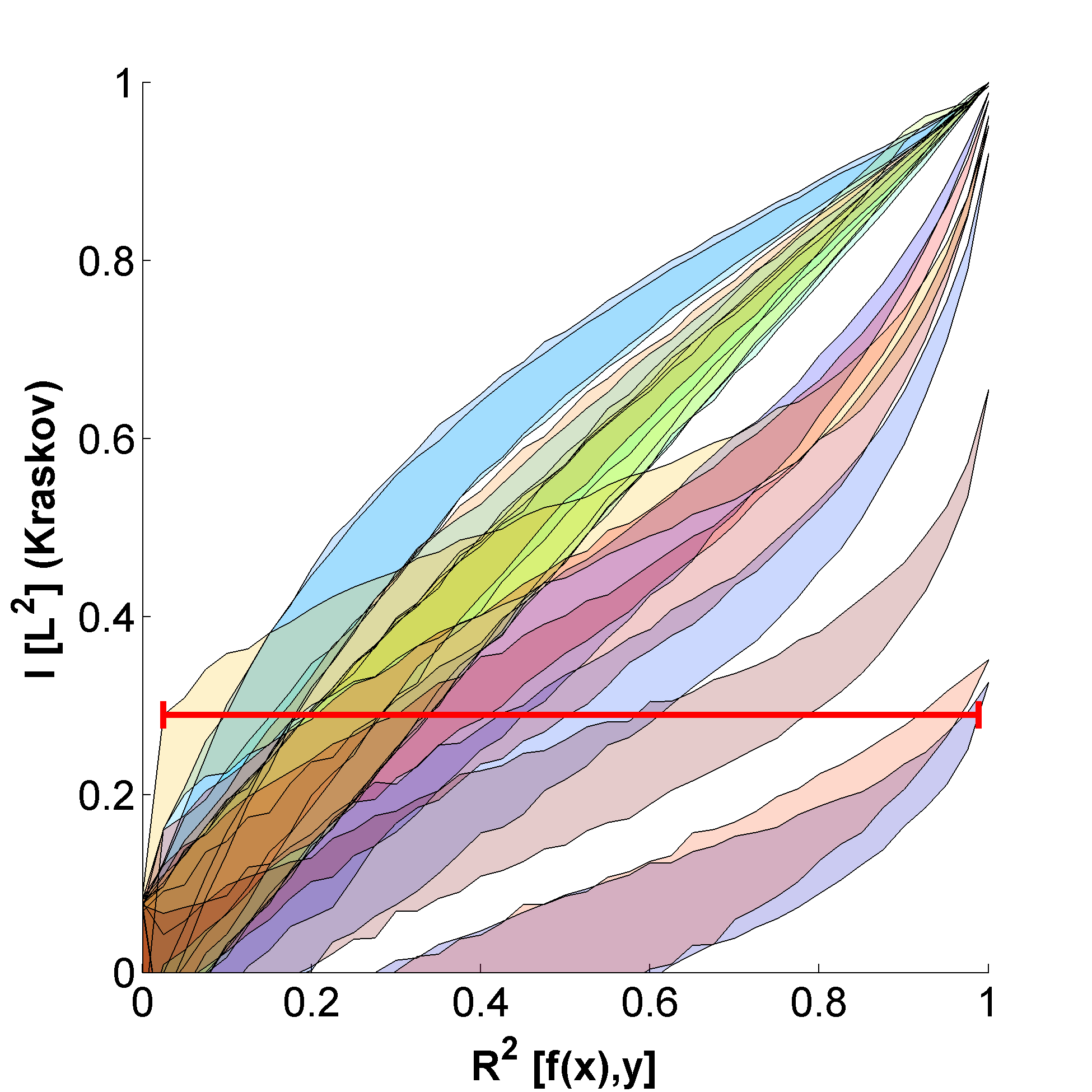} \hspace*{\fill} \\
    \hspace{0.16in} (e) & & (f) \hspace{1.25in} \\
    \end{tabular}
    \caption{Equitability with respect to $R^2$ on a set of noisy functional relationships of \figpart{a} the Pearson correlation coefficient, \figpart{b} a hypothetical measure of dependence $\varphi$ with perfect equitability, \figpart{c} distance correlation, \figpart{d} $\MICestE$, \figpart{e} maximal correlation estimation, and \figpart{f} mutual information estimation. For each relationship, a shaded region denotes 5th and 95th percentile values of the sampling distribution of the statistic in question on that relationship at every $R^2$. The resulting plot shows which values of $R^2$ correspond to a given value of each statistic. The red interval on each plot indicates the widest range of $R^2$ values corresponding to any one value of the statistic; the narrower the red interval, the higher the equitability. A red interval with width 0, as in \figpart{b}, means that the statistic reflects only $R^2$ with no dependence on relationship type, as demonstrated by the pairs of thumbnails of relationships of different types with identical $R^2$ values.
    }
    \label{fig:equitabilityAnalysis}
\end{figure}

To assess equitability more objectively without relying on a manually curated set of functions, we then analyzed 160 random functions drawn from a Gaussian process distribution with a radial basis function kernel with one of eight possible bandwidths in the set
\[
\{0.01, 0.025, 0.05, 0.1, 0.2, 0.25, 0.5, 1\}
\]
to represent a range of possible relationship complexities. The results, shown in Figure~\ref{fig:equitabilityAnalysis_evenCurve_GP_XYNoise}, show that $\MICestE$ outperforms currently existing methods in terms of equitability with respect to $R^2$ on these functions as well. Appendix Figure~\ref*{fig:equitabilityAnalysis_evenCurve_GP_YNoise} shows a version of this analysis under a different noise model that yields the same conclusion. We also examined the effect of outlier relationships on our results by repeatedly subsampling random subsets of 20 functions from this large set of relationships and measuring the equitability of each method on average over the subsets; results were similar.

One feature of the performance of $\MICestE$ on these randomly chosen relationships that is demonstrated in Figure~\ref{fig:equitabilityAnalysis_evenCurve_GP_XYNoise} is that it appears minimally sensitive to the bandwidth of the Gaussian process from which a given relationship is drawn. This puts it in contrast to, e.g., mutual information estimation, which shows a pronounced sensitivity to this parameter that prevents it from being highly equitable when relationships with different bandwidths are present in the same dataset.

In our companion paper \citep{reshef2015comparisons}, we perform more in-depth analyses of the equitability with respect to $R^2$ of $\MICestE$, $\MIC$, and the four measures of dependence described above as well as the Hilbert-Schmidt independence criterion (HSIC) \citep{gretton2005measuring, gretton2007kernel}, the Heller-Heller-Gorfine (HHG) test \citep{heller2013consistent}, the data-derived partitions (DDP) test \citep{heller2014consistent}, and the randomized dependence coefficient (RDC) \citep{lopez2013randomized}. These analyses consider a range of sample sizes, noise models, marginal distributions, and parameter settings. They conclude that, in terms of equitability with respect to $R^2$ on the sets of noisy functional relationships studied, a) $\MICestE$ uniformly outperforms $\MIC$, and b) $\MICestE$ outperforms all the methods tested in the vast majority of settings examined. Appendix Figure~\ref*{fig:equitabilityAnalysis_evenCurve_XYNoise} contains a reproduction of a representative equitability analysis from that paper for the reader's reference.

\begin{figure}
    \centering
    \includegraphics[clip=true, trim = 0in 0in 0in 0in, height=7in]{\pathToFigures/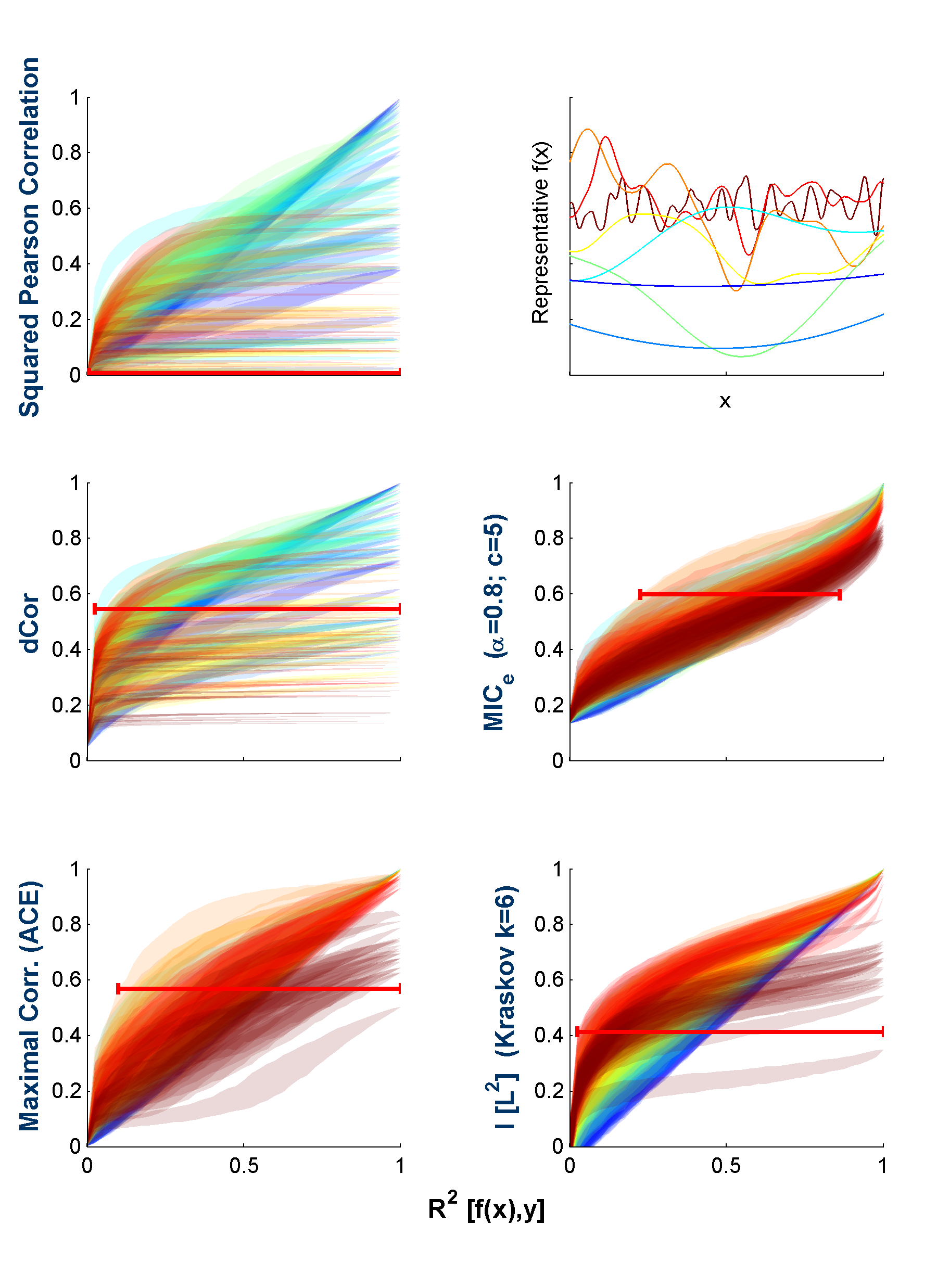}
    \caption{Equitability of methods examined on functions randomly drawn from a Gaussian process distribution. Each method is assessed as in Figure~\ref{fig:equitabilityAnalysis}, with a red interval indicating the widest range of $R^2$ values corresponding to any one value of the statistic; the narrower the red interval, the higher the equitability. Each shaded region corresponds to one relationship, and the regions are colored by the bandwidth of the Gaussian process from which they were sampled. Sample relationships for each bandwidth are shown in the top right with matching colors.
    }
    \label{fig:equitabilityAnalysis_evenCurve_GP_XYNoise}
\end{figure}

\section{The total information coefficient}
\label{sec:TIC}
So far we have presented results about estimators of the population maximal information coefficient, a quantity for which equitability is the primary motivation. We now introduce and analyze a new measure of dependence, the {\em total information coefficient} (TIC). In contrast to the maximal information coefficient, the total information coefficient is designed not for equitability but rather as a test statistic for testing a null hypothesis of independence.

We begin by giving some intuition. Recall that the maximal information coefficient is the supremum of the characteristic matrix. While estimating the supremum of this matrix has many advantages, this estimation involves taking a maximum over many estimates of individual entries of the characteristic matrix. Since maxima of sets of random variables tend to become large as the number of variables grows, one can imagine that this procedure will lead to an undesirable positive bias in the case of statistical independence, when the population characteristic matrix equals 0. This is detrimental for independence testing, when the sampling distribution of a statistic under a null hypothesis of independence is crucial.

The intuition behind the total information coefficient is that if we instead consider a more stable property, such as the sum of the entries in the characteristic matrix, we might expect to obtain a statistic with a smaller bias in the case of independence and therefore better power. Stated differently, if our only goal is to distinguish any dependence at all from complete noise, then disregarding all of the sample characteristic matrix except for its maximal value may throw away useful signal, and the total information coefficient avoids this by summing all the entries.

We remark that in \cite{MINE} it is suggested that other properties of the characteristic matrix may allow us to measure other aspects of a given relationship besides its strength, and several such properties were defined. The total information coefficient fits within this conceptual framework.

In this section we define the total information coefficient in the case of both the characteristic matrix ($\TIC$) and the equicharacteristic matrix ($\TICestE$). We then prove that both $\TIC$ and $\TICestE$ yield independence tests that are consistent against all dependent alternatives. Finally, we present a simulation study of the power of independence testing based on $\TICestE$, showing that it outperforms other common measures of dependence.

\subsection{Definition and consistency of the total information coefficient}
We begin by defining the two versions of the total information coefficient. In the definition below, recall that $\widehat{M}$ denotes a sample characteristic matrix whereas $\widehat{[M]}$ denotes a sample equicharacteristic matrix.
\begin{definition}
Let $D \subset \R^2$ be a set of $n$ ordered pairs, and let $B : \Z^+ \rightarrow \Z^+$. We define
\[
\TIC_B(D) = \sum_{k\ell \leq B(n)} \widehat{M}(D)_{k, \ell}
\]
and
\[
\TICestE_{,B}(D) = \sum_{k\ell \leq B(n)} \widehat{[M]}(D)_{k, \ell} .
\]
\end{definition}

To show that these two statistics lead to consistent independence tests, we must take a step back and analyze the behavior of the analogous population quantities. We take some care with the limits involved in defining these quantities, since the fine-grained behavior of these limits will be a key part of our proof of consistency.

\begin{definition}
For a matrix $A$ and a positive number $B$, the {\em $B$-partial sum} of $A$, denoted by $S_B(A)$, is
\[
S_B(A) = \sum_{k\ell \leq B} A_{k,\ell} .
\]
\end{definition}
When $A$ is an (equi)characteristic matrix, $S_B(A)$ is the sum over all entries corresponding to grids with at most $B$ total cells. Thus, if $\widehat{M}(D)$ is a sample characteristic matrix of a sample $D$, $S_B(\widehat{M}(D)) = \TIC_B(D)$, and the same holds for $\widehat{[M]}(D)$ and $\TICestE_{,B}(D)$.

It is clear that if $X$ and $Y$ are statistically independent random variables, then both the characteristic matrix $M(X,Y)$ and the equicharacteristic matrix $[M](X,Y)$ are identically 0, so that $S_B(M(X,Y)) = S_B([M](X,Y)) = 0$ for all $B$. However, we are also interested in how these quantities behave when $X$ and $Y$ are dependent. The following pair of propositions helps us understand this. The first proposition shows a lower bound on the values of entries in both $M(X,Y)$ and $[M](X,Y)$. The second proposition translates this into an asymptotic characterization of how quickly $S_B(M)$ and $S_B([M])$ grow as functions of $B$. These two propositions are the technical heart of why the total information coefficient yields a consistent independence test.

\begin{proposition}
\label{prop:nonzero_submatrix}
Let $(X,Y)$ be a pair of jointly distributed random variables. If $X$ and $Y$ are statistically independent, then $M(X,Y) \equiv [M](X,Y) \equiv 0$. If not, then there exists some $a > 0$ and some integer $\ell_0 \geq 2$ such that
\[
M(X,Y)_{k,\ell}, [M](X,Y)_{k,\ell} \geq \frac{a}{\log \min \{k, \ell\}}
\]
either for all $k \geq \ell \geq \ell_0$, or for all $\ell \geq k \geq \ell_0$.
\end{proposition}
\begin{proof}
See Appendix~\ref{app:nonzero_submatrix}
\end{proof}

\begin{proposition}
\label{prop:growth_of_S}
Let $(X,Y)$ be a pair of jointly distributed random variables. If $X$ and $Y$ are statistically independent, then $S_B(M(X,Y)) = S_B([M](X,Y)) = 0$ for all $B > 0$. If not, then $S_B(M(X,Y))$ and $S_B([M](X,Y))$ are both $\Omega(B \log \log B)$.
\end{proposition}
\begin{proof}
See Appendix~\ref{app:growth_of_S}
\end{proof}

The propositions above, together with reasoning analogous to the convergence arguments presented earlier, can be used to show the main result of this section, namely that the statistics $\TIC$ and $\TICestE$ yield consistent independence tests.
\begin{theorem}
\label{thm:consistency_TICe}
The statistics $\TIC_B$ and $\TICestE_{,B}$ yield consistent right-tailed tests of independence, provided $\omega(1) < B(n) \leq O(n^{1-\ep})$ for some $\ep > 0$.
\end{theorem}
\begin{proof}
See Appendix~\ref{app:consistency_TICe}.
\end{proof}

In practice, we often use the \algname{EquicharClump} algorithm (see Section~\ref{subsec:computing_MICe}) to compute the equicharacteristic matrix from which we calculate $\TICestE$. This algorithm does not compute the sample equicharacteristic matrix exactly. However, as in the case of $\MICestE$, the use of the algorithm does not affect the consistency of our approach for independence testing. This is proven in Appendix~\ref{app:computation}.

\subsection{Power of independence tests based on \texorpdfstring{$\TICestE$}{TICe}}
\label{subsec:power_analysis}
With the consistency of independence tests based on $\TIC$ and $\TICestE$ established, we turn now to empirical evaluation of the power of independence testing based on $\TICestE$ as computed using the \algname{EquicharClump} algorithm.

To evaluate the power of $\TICestE$-based tests, we reproduced the analysis performed in \cite{simon2012comment}. Namely, for the set of functions $F$ chosen by Simon and Tibshirani, we considered the set of relationships they analyzed:
\[
\Q = \left\{ (X, f(X) + \ep') : X \sim \mbox{Unif}, f \in F, \ep' \sim \mathcal{N}(0, \sigma^2), \sigma \in \R_{\geq 0} \right\} .
\]

For each relationship $\mcZ$ in this set that we examined, we simulated a null hypothesis of independence with the same marginal distributions, and generated $1,000$ independent samples, each with a sample size of $n=500$, from both $\mcZ$ and from the null distribution. These were used to estimate the power of the size-$\alpha$ right-tailed independence test based on each statistic being evaluated. Following Simon and Tibshirani, we compared $\TICestE$ to the distance correlation \citep{szekely2007measuring, szekely2009brownian}, the original maximal information coefficient~\citep{MINE} as approximated using $\algname{Approx-MIC}$, and to the Pearson correlation. (Though it is not a measure of dependence, the Pearson correlation was presumably included by Simon and Tibshirani as an intuitive benchmark for what is achievable under a linear model.) We also compared to $\MICestE$ using identical parameters to those of $\TICestE$ to examine whether the summation performed by $\TICestE$ is better than maximization when all other things are equal. Note that we do not compare to methods of analyzing contingency tables, such as Pearson's chi-squared test. This is because our data are real-valued rather than discrete, and so contingency-based methods are not applicable. However, when data are discrete, those methods can be very well powered.

\begin{figure}
	\centering
	\includegraphics[clip=true, trim = 0.85in 0.8in 0.85in 0.6in, width=0.975\textwidth]{\pathToFigures/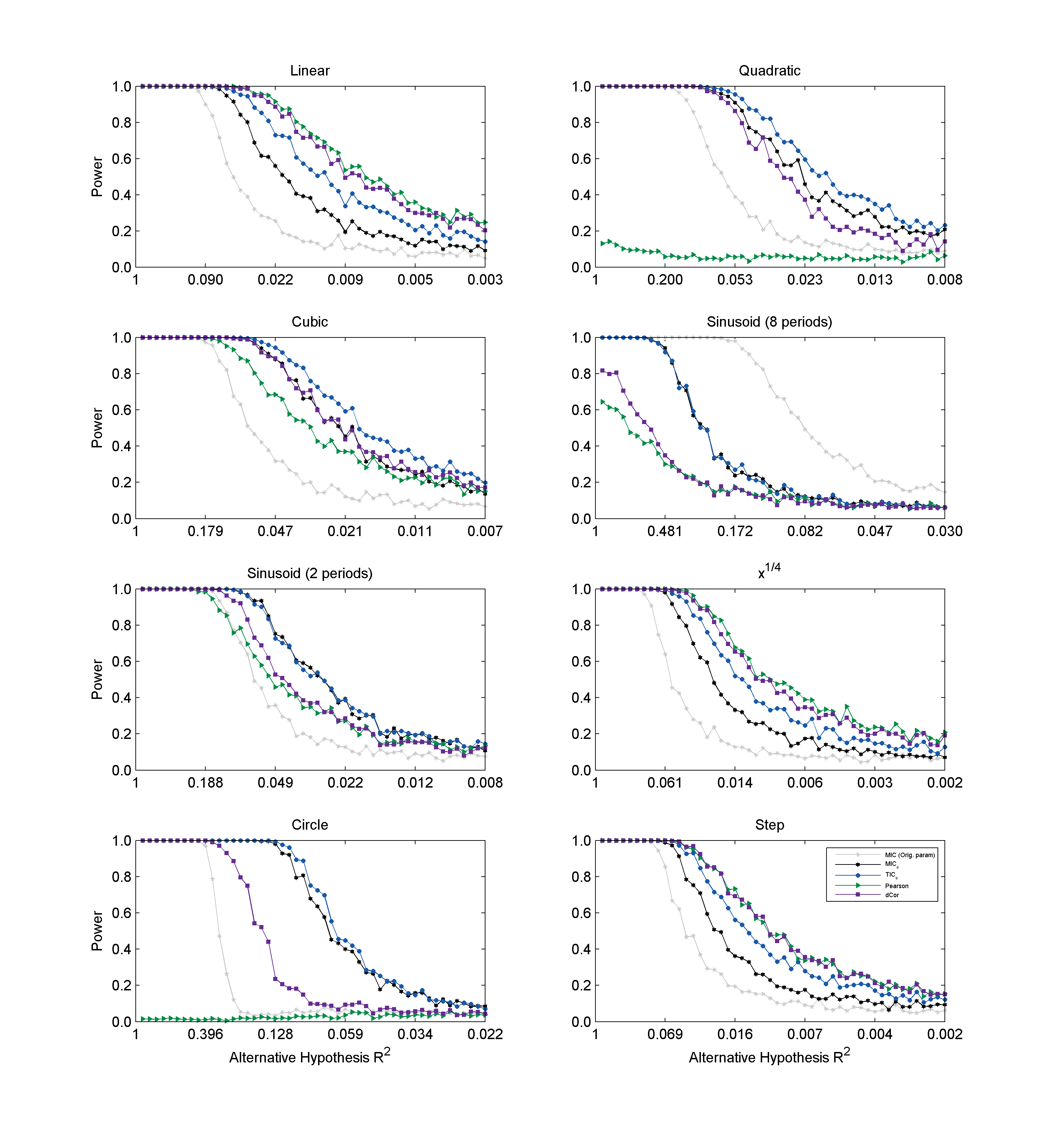}
	\caption[Power of independence tests based on measures of dependence across several relationship types]{Comparison of power of independence testing based on $\TICestE$ (blue) to $\MIC$ with default parameters (gray), $\MICestE$ with the same parameters as $\TICestE$ (black), distance correlation (purple), and the Pearson correlation coefficient (green) across several alternative hypothesis relationship types chosen by \cite{simon2012comment}. The relationships analyzed are described in Section~\ref{subsec:power_analysis}.
}
\label{fig:power_analysis}
\end{figure}

The results of our analysis are presented in Figure~\ref{fig:power_analysis}. First, the figure shows that $\TICestE$ compares quite favorably with distance correlation, a method considered to have state-of-the-art power \citep{simon2012comment}. Specifically, $\TICestE$ uniformly outperforms distance correlation on 5 of the 8 relationship types examined, and performs comparably to it on the other three relationship types. Distance correlation has many advantages over $\TICestE$, including the fact that it easily generalizes to higher-dimensional relationships. However, in the bivariate setting $\TICestE$ appears to perform as well as distance correlation if not better in terms of statistical power against independence.

The analysis also shows that $\TICestE$ outperforms the original maximal information coefficient by a very large margin, and outperformed $\MICestE$ as well, supporting the intuition that the summation performed by the former can indeed lead to substantial gains in power against independence over the maximization performed by the latter. (We note that in both Simon and Tibshirani's analysis and in this one, the original maximal information coefficient was run with default parameters that were optimized for equitability rather than power against independence. When run with different parameters, its power improves substantially, though it still does not match the power of $\MICestE$. See Appendix Figure~\ref*{fig:indivRelPowerOptimalParam} and the discussion in \cite{reshef2015comparisons}.)

Our companion paper \citep{reshef2015comparisons} expands on this analysis, conducting in-depth evaluation of the the power against independence of the tests described above as well as tests based on mutual information estimation \citep{Kraskov}, maximal correlation estimation \citep{breiman1985estimating}, HSIC \citep{gretton2005measuring, gretton2007kernel}, HHG \citep{heller2013consistent}, DDP \citep{heller2014consistent}, and RDC \citep{lopez2013randomized}. These analyses consider a range of sample sizes and parameter settings, as well as a variety of ways of quantifying power across different alternative hypothesis relationship types and noise levels. They conclude that in most settings $\TICestE$ either outperforms all the methods tested or performs comparably to the best ones. Appendix Figure~\ref*{fig:indivRelPowerOptimalParam} contains a reproduction of one detailed set of power curves from the main analysis in that paper for the reader's reference.

\section{Conclusion}
As high-dimensional data sets become increasingly common, data exploration requires not only statistics that can accurately detect a large number of non-trivial relationships in a data set, but also ones that can identify a smaller number of strongest relationships. The former property is achieved by measures of dependence that yield independence tests with high power; the latter is achieved by measures of dependence that are equitable with respect to some measure of relationship strength. In this paper, we introduced two related measures of dependence that achieve these two goals, respectively, through the following theoretical contributions.

\begin{itemize}
\item A new population measure of dependence, $\popMIC$, that we proved can be viewed in three different ways: as the population value of the maximal information coefficient ($\MIC$) from \cite{MINE}, as a ``minimal smoothing'' of mutual information that makes it uniformly continuous, or as the supremum of an infinite sequence defined in terms of optimal partitions of one marginal at a time of a given joint distribution.
\item An efficient algorithm for approximating the $\popMIC$ of a given joint distribution.
\item A statistic $\MICestE$ that is a consistent estimator of $\popMIC$, is efficiently computable, and has good equitability with respect to $R^2$ both on a manually chosen set of noisy functional relationships as well as on randomly chosen noisy functional relationships.
\item The total information coefficient ($\TICestE$), a statistic that arises as a trivial side-product of the computation of $\MICestE$ and yields a consistent and powerful independence test.
\end{itemize}

Though we presented here some empirical results for $\popMIC$, $\MICestE$, and $\TICestE$, our focus was on theoretical considerations; the performance of these methods is analyzed in detail in our companion paper \citep{reshef2015comparisons}. That paper shows that on a large set of noisy functional relationships with varying noise and sampling properties, the asymptotic equitability with respect to $R^2$ of $\popMIC$ is quite high and the equitability with respect to $R^2$ of $\MICestE$ is state-of-the-art. It also shows that the power of the independence test based on $\TICestE$ is state-of-the-art across a wide variety of dependent alternative hypotheses. Finally, it demonstrates that the algorithms presented here allow for $\MICestE$ and $\TICestE$ to be computed simultaneously very quickly, enabling analysis of extremely large data sets using both statistics together.

Our contributions are of both theoretical and practical importance for several reasons. First, our characterization of $\popMIC$ as the large-sample limit of $\MIC$ sheds light on the latter statistic. For example, while $\MIC$ is parametrized, $\popMIC$ is not. Knowing that $\MIC$ converges in probability to $\popMIC$ tells us that this parametrization is statistical only: it controls the bias/variance properties of the statistic, but not its asymptotic behavior.

Second, the normalization in the definition of $\MIC$, while empirically seen to yield good performance, had previously not been theoretically understood. Our result that this normalization is the minimal smoothing necessary to make mutual information uniformly continuous provides for the first time a lens through which the normalization is canonical. In doing so, it constitutes an initial step toward understanding the role of the normalization in the performance of $\popMIC$ and $\MIC$. The uniform continuity of $\popMIC$ and the lack of continuity of ordinary mutual information also suggest that estimation of the former may be easier in some sense than estimation of the latter. This is consonant with the result concerning difficulty of estimation of mutual information shown in \cite{ding2013copula}. It is also borne out empirically by the substantial finite-sample bias and variance observed in \cite{reshef2015comparisons} of the Kraskov mutual information estimator \citep{Kraskov} compared to $\MICestE$.

Third, our alternate characterization of $\popMIC$ in terms of one-dimensional optimization over partitions rather than two-dimensional optimization over grids enhances our understanding of how to efficiently compute it in the large-sample limit and estimate it from finite samples using $\MICestE$. This is a significant improvement over the previous state of affairs, in which the statistic $\MIC$ could only be approximated heuristically, and orders of magnitude slower than the results in this paper now allow.

Finally, the introduction of the total information coefficient provides evidence that the basic approach of considering the set of normalized mutual information values achievable by applying different grids to a joint distribution is of fundamental value in characterizing dependence. Interestingly, a statistic introduced in \cite{heller2014consistent} follows a similar approach by considering the (non-normalized) sum of the mutual information values achieved by all possible finite grids. Consistent with our demonstration here that an aggregative grid-based approach works well, that statistic also achieves excellent power. ($\TICestE$ is compared to the statistic from \cite{heller2014consistent} in our companion paper, \cite{reshef2015comparisons}.)

Taken together, our results point to joint use of the statistics $\MICestE$ and $\TICestE$ as a theoretically grounded, computationally efficient, and highly practical approach to data exploration. Specifically, since the two statistics can be computed simultaneously with little extra cost beyond that of computing either individually, we propose computing both of them on all variable pairs in a data set, using $\TICestE$ to filter out non-significant associations, and then using $\MICestE$ to rank the remaining variable pairs. Such a strategy would have the advantage of leveraging the state-of-the-art power of $\TICestE$ to substantially reduce the multiple-testing burden on $\MICestE$, while utilizing the latter statistic's state-of-the-art equitability to effectively rank relationships for follow-up by the practitioner.

Of course our results, while useful, nevertheless have limitations that warrant exploration in future work. First, for a sample $D$ from the distribution of some random $(X,Y)$, all of the sample quantities we define here use the naive estimate $I(D|_G)$ of the quantity $I((X,Y)|_G)$ for various grids $G$. There is a long and fruitful line of work on more sophisticated estimators of the discrete mutual information~\cite{paninski2003estimation} whose use instead of $I(D|_G)$ could improve the statistics introduced here. Second, our approach to approximating the $\popMIC$ of a given joint density consists of computing a finite subset of an infinite set whose supremum we seek to calculate. However, the choice of how large a finite set we should compute in order to approximate the supremum to a given precision remains heuristic. Finally, though empirical characterization of the equitability of $\MICestE$ on representative sets of relationships is important and promising, we are still missing a theoretical characterization of its equitability in the large-sample limit. A clear theoretical demarcation of the set of relationships on which $\popMIC$ achieves good equitability with respect to $R^2$, and an understanding of why that is, would greatly advance our understanding of both $\popMIC$ and equitability.

\section{Acknowledgments}
We would like to acknowledge R Adams, T Broderick, E Airoldi, A Gelman, M Gorfine, R Heller, J Huggins, J Mueller, and R Tibshirani for constructive conversations and useful feedback.

\appendix
\setcounter{figure}{0}
\renewcommand{\thefigure}{\thesection\arabic{figure}}

\section{Proof of Theorem~\ref{thm:estimator}}
\setcounter{figure}{0}
\label{appendix:popmicproof}
This appendix is devoted to proving Theorem~\ref{thm:estimator}, restated below.

\vspace{0.3cm}
\noindent \textbf{Theorem}
\textit{
Let $f : m^\infty \rightarrow \R$ be uniformly continuous, and assume that $f \circ r_i \rightarrow f$ pointwise. Then for every random variable $(X,Y)$, we have
\[
\left( f \circ r_{B(n)} \right) \left(\widehat{M}(D_n) \right) \rightarrow f(M(X,Y))
\]
in probability where $D_n$ is a sample of size $n$ from the distribution of $(X,Y)$, provided $\omega(1) < B(n) \leq O(n^{1-\ep})$ for some $\ep > 0$.}
\vspace{0.3cm}

We prove the theorem by a sequence of lemmas that build on each other to bound the bias of $I^*(D, k, \ell)$. The general strategy is to capture the dependencies between different $k$-by-$\ell$ grids $G$ by considering a ``master grid'' $\Gamma$ that contains many more than $k\ell$ cells. Given this master grid, we first bound the difference between $I(D|_G)$ and $I((X,Y)|_G)$ only for sub-grids $G$ of $\Gamma$. The bound is in terms of the difference between $D|_\Gamma$ and $(X,Y)|_\Gamma$. We then show that this bound can be extended without too much loss to all $k$-by-$\ell$ grids. This gives what we seek, because then the difference between $I(D|_G)$ and $I((X,Y)|_G)$ is uniformly bounded for all grids $G$ in terms of the same random variable: $D|_\Gamma$. Once this is done, standard arguments give the consistency we seek.

In our argument we occasionally require technical facts about entropy and mutual information that are self-contained and unrelated to the central ideas. These lemmas are consolidated in Appendix~\ref{appendix:technical}.

We begin by using one of these technical lemmas to prove a bound on the difference between $I(D|_G)$ and $I((X,Y)|_G)$ that is uniform over all grids $G$ that are sub-grids of a much denser grid $\Gamma$. The common structure imposed by $\Gamma$ will allow us to capture the dependence between the quantities $\left| I(D|_G) - I((X,Y)|_G) \right|$ for different grids $G$.

\begin{lemma}
\label{lem:boundMIOnSubgrids}
Let $\Pi = (\Pi_X, \Pi_Y)$ and $\Psi = (\Psi_X, \Psi_Y)$ be random variables distributed over the cells of a grid $\Gamma$, and let $(\pi_{i,j})$ and $(\psi_{i,j})$ be their respective distributions. Define
\[ \ep_{i,j} = \frac{\psi_{i,j} - \pi_{i, j}}{\pi_{i,j}} . \]
Let $G$ be a sub-grid of $\Gamma$ with $B$ cells. Then for every fixed $0 < a < 1$ we have
\[ \left| I(\Psi|_G) - I(\Pi|_G) \right| \leq \O{ \left( \log B \right) \sum_{i,j} |\ep_{i,j}| } \]
when $|\ep_{i,j}| \leq 1 - a$ for all $i$ and $j$.
\end{lemma}

\begin{proof}
Let $P = \Pi|_G$ and $Q = \Psi|_G$ be the random variables induced by $\Pi$ and $\Psi$ respectively on the cells of $G$. Using the fact that $I(X,Y) = H(X) + H(Y) - H(X,Y)$, we write
\[ \left| I(Q) - I(P) \right| \leq \left| H(Q_X) - H(P_X) \right| + \left| H(Q_Y) - H(P_Y) \right| + \left| H(Q) - H(P) \right| \]
where $Q_X$ and $P_X$ denote the marginal distributions on the columns of $G$ and $Q_Y$ and $P_Y$ denote the marginal distributions on the rows. We can bound each of the terms on the right-hand side of the equation above using a Taylor expansion argument given in Lemma~\ref{lem:boundEntropy}, whose proof is found in the appendix. Doing so gives
\[ \left| I(Q) - I(P) \right| \leq (\ln B) \left( \sum_i \O{|\ep_{i,*}|} + \sum_j \O{|\ep_{*,j}|} + \sum_{i,j} \O{|\ep_{i,j}|} \right) \]
where
\[ \ep_{i,*} = \frac{\sum_j (\psi_{i,j} - \pi_{i,j})}{\sum_j \pi_{i,j}} \]
and $\ep_{*,j}$ is defined analogously.

To obtain the result, we observe that
\[ \left| \ep_{i,*} \right|
	= \left| \frac{\sum_j \pi_{i,j} \ep_{i,j}}{\sum_j \pi_{i,j}} \right|
	\leq \frac{ \sum_j \pi_{i,j} \left| \ep_{i,j} \right|}{\sum_j \pi_{i,j}}
	\leq \sum_j \left| \ep_{i,j} \right|
\]
since $\pi_{i,j} / \sum_j \pi_{i,j} \leq 1$, and the analogous bound holds for $\left| \ep_{*,j} \right|$.
\end{proof}

We now extend Lemma~\ref{lem:boundMIOnSubgrids} to all grids with $B$ cells rather than just those that are sub-grids of the master grid $\Gamma$. The proof of this lemma relies on an information-theoretic result proven in Appendix~\ref{app:continuityofM} that bounds the difference in mutual information between two distributions that can be obtained from each other by moving a small amount of probability mass.

\begin{lemma}
\label{lem:boundMI}
Let $\Pi = (\Pi_X, \Pi_Y)$ and $\Psi = (\Psi_X, \Psi_Y)$ be random variables, and let $\Gamma$ be a grid. Define $\ep_{i,j}$ on $\Pi|_\Gamma$ and $\Psi|_\Gamma$ as in Lemma~\ref{lem:boundMIOnSubgrids}. Let $G$ be any grid with $B$ cells, and let $\delta$ (resp. $d$) represent the total probability mass of $\Pi|_\Gamma$ (resp. $\Psi|_\Gamma$) falling in cells of $\Gamma$ that are not contained in individual cells of $G$. We have that
\[ \left| I(\Psi|_G) - I(\Pi|_G) \right|
\leq \O{ \left( \sum_{i,j} \left| \ep_{i,j} \right| + \delta + d \right) \log B + \delta \log(1/\delta) + d \log(1/d) }
\]
provided that the $|\ep_{i,j}|$ are bounded away from 1 and that $d, \delta \leq 1/2$.
\end{lemma}
\begin{proof}
In the proof below, we use the convention that for any two grids $G$ and $G'$ and any random variable $Z$, the expression $\Delta^Z(G, G')$ denotes $| I(Z|_G) - I(Z|_{G'})|$.

Consider the grid $G'$ obtained by replacing every horizontal or vertical line in $G$ that is not in $\Gamma$ with a closest line in $\Gamma$. The grid $G'$ is clearly a sub-grid of $\Gamma$. Moreover, $\Pi|_{G'}$ (resp. $\Psi|_{G'}$) can be obtained from $\Pi|_G$ (resp. $\Pi|_G$) by moving at most $\delta$ (resp. $d$) probability mass. This can be shown to imply that
\[ \Delta^\Pi(G, G') \leq \O{ \delta \log(1/\delta) + \delta \log B }
\quad \mbox{and} \quad
\Delta^{\Psi}(G', G) \leq \O{ d \log(1/d) + d \log B } .\]
The proof of this information-theoretic fact is self-contained and so we defer it to Proposition~\ref{prop:boundedVariationOfI} in Appendix~\ref{app:continuityofM}, as it is more central to the arguments presented there.

With $\Delta^\Phi(G,G')$ and $\Delta^\Psi(G', G)$ bounded in terms of $\delta$ and $d$, we can bound $|I(\Psi|_G) - I(\Phi|_G)|$ using the triangle inequality by comparing it with
\[ \Delta^\Pi(G, G')
+ \left| I \left( \Pi|_{G'} \right) - I \left( \Psi|_{G'} \right) \right|
+ \Delta^{\Psi}(G', G) \]
and bounding the middle term using Lemma~\ref{lem:boundMIOnSubgrids}, since $G' \subset \Gamma$.
\end{proof}

We now use the fact that the variables $\ep_{i,j}$ defined in Lemma~\ref{lem:boundMIOnSubgrids} are small with high probability to give a concrete bound on the bias of $I(D|_G)$ that is uniform over all $k$-by-$\ell$ grids $G$ and that holds with high probability. It is useful at this point to recall that, given a distribution $(X,Y)$, an {\em equipartition} of $(X,Y)$ is a grid $G$ such that all the rows of $(X,Y)|_G$ have the same probability mass, and all the columns do as well.

\begin{lemma}
\label{lem:unionBoundOverGrids}
Let $D_n$ be a sample of size $n$ from the distribution of a pair $(X,Y)$ of jointly distributed random variables. For any $\alpha \geq 0$, any $\ep > 0$, and any integers $k, \ell > 1$, we have that for all $n$
\[ \left| I(D_n|_G) - I((X,Y)|_G) \right|
	\leq \O{ \frac{\log (k \ell)}{C(n)^\alpha} + \frac{\log (k\ell n)}{n^{\ep/4}} } \]
for every $k$-by-$\ell$ grid $G$ with probability at least $1 - C(n) e^{-\Omega(n/C(n)^{1+2\alpha})}$, where $C(n) = k\ell n^{\ep/2}$.
\end{lemma}

\begin{proof}
Fix $n$, and let $\Gamma$ be an equipartition of $(X,Y)$ into $kn^{\ep/4}$ rows and $\ell n^{\ep/4}$ columns. $C(n)$ is now the number of cells in $\Gamma$. Lemma~\ref{lem:boundMI}, with $\Pi = (X,Y)$ and $\Psi = D$, shows that $\left| I(D|_G) - I((X,Y)|_G) \right|$ is at most
\[
\O{ \left( \sum_{i,j} \left| \ep_{i,j} \right| + \delta + d \right) \log (k\ell) + \delta \log(1/\delta) + d \log(1/d) }
\]
provided the $\ep_{i,j}$ have absolute value bounded away from 1, and provided that $d, \delta \leq 1/2$.

The remainder of the proof proceeds as follows. We first show that the $\ep_{i,j}$ are small with high probability. This will both show that the lemma's requirement on the $\ep_{i,j}$ holds and allow us to bound the sum in the inequality above. We will then use our bound on the $\ep_{i,j}$ to bound $d$ in terms of $\delta$. Finally, we will bound $\delta$ using the fact that the number of rows and columns in $\Gamma$ increases with $n$. This will give us that $d, \delta \leq 1/2$ and allow us to bound the rest of the terms in the expression above.

\proofwaypoint{Bounding the $\ep_{i,j}$}
We bound the $\ep_{i,j}$ using a multiplicative Chernoff bound. Let $\pi_{i,j}$ and $\psi_{i,j}$ represent the probability mass functions of $(X,Y)|_\Gamma$ and $D|_\Gamma$ respectively. We write
\begin{eqnarray*}
\Pr{ \left| \ep_{i,j} \right| \geq \delta} &=& \Pr{\pi_{i,j}(1 - \delta) \leq \psi_{i,j} \leq \pi_{i,j}(1 + \delta)} \\
	&\leq& e^{-\Omega(n\pi_{i,j} \delta^2)}
\end{eqnarray*}
since $\psi_{i,j}$ is a sum of $n$ i.i.d Bernoulli random variables and $\E{\psi_{i,j}} = n\pi_{i,j}$. (See, e.g., \cite{mitzenmacher2005probability}.) Setting $\delta = \sqrt{\pi_{i,j}}/C(n)^{1/2 + \alpha}$ yields
\[
\Pr{|\ep_{i,j}| \geq \frac{\sqrt{\pi_{i,j}}}{C(n)^{1/2 + \alpha}}} \leq e^{-\Omega(n/C(n)^{1+2\alpha})} .
\]
A union bound over the pairs $(i,j)$ then gives that, with the desired probability, the above bound on $|\ep_{i,j}|$ holds for all $i,j$.

\proofwaypoint{Bounding $\sum \left| \ep_{i,j} \right|$}
The bound on the $\ep_{i,j}$ implies that
\begin{eqnarray*}
\sum_i |\ep_{i,j}| &\leq& \frac{1}{C(n)^{1/2 + \alpha}} \sum_{i,j} \sqrt{\pi_{i,j}} \\
	&\leq& \frac{1}{C(n)^{1/2 + \alpha}} \sqrt{C(n)} \\
	&\leq& \frac{1}{C(n)^{\alpha}}
\end{eqnarray*}
where the second line follows from the fact that the function $\sum \sqrt{\pi_{i,j}}$ is symmetric and concave and therefore, when restricted to the hyperplane $\sum \pi_{i,j }= 1$, must achieve its maximum when $\pi_{i,j} = 1/C(n)$ for all $i,j$.

\proofwaypoint{Bounding $d$ in terms of $\delta$}
We use our bound on the $\ep_{i,j}$ to bound $d$. We do so by observing that it implies
\[
\psi_{i,j} \leq \pi_{i,j} \left( 1 + \frac{\sqrt{\pi_{i,j}}}{C(n)^{1/2 + \alpha}} \right)
	= \pi_{i,j} + \frac{\pi_{i,j}^{3/2}}{C(n)^{1/2 + \alpha}}
	\leq \pi_{i,j} + \frac{\pi_{i,j}}{C(n)^{1/2 + \alpha}}
	\leq 2\pi_{i,j} \]
since $\pi_{i,j} \leq 1$ and $C(n) \geq 1$.

The connection to $d$ comes from the fact that for any column $j$ of $\Gamma$, this means that
\[
\psi_{*,j} = \sum_i \psi_{i,j} \leq 2 \sum_i \pi_{i,j} = 2 \pi_{*,j} .
\]
This also applies to the sums across rows. Since $d$ is a sum of terms of the form $\psi_{*,j}$ and $\psi_{i,*}$ for $j$ in some index set $J$ and $i$ in an index set $I$, and $\delta$ is a sum of terms of the form $\pi_{*,j}$ and $\pi_{i,*}$ with the same index sets, we therefore get that $d \leq 2 \delta$.

\proofwaypoint{Bounding $\delta$ and obtaining the result}
To bound $\delta$, we observe that because $G$ has at most $\ell - 1$ vertical lines and $k -1$ horizontal lines, we have
\[
\delta \leq \frac{\ell}{\ell n^{\ep/ 4}} + \frac{k}{k n^{\ep/4}} \leq \frac{2}{n^{\ep/4}} .
\]

This bound on $\delta$ allows us to bound the terms involving $d$ and $\delta$ by
\[
\delta + d \leq \O{ \frac{1}{n^{\ep/4}} },
\quad
\delta \log\left( \frac{1}{\delta} \right) + d \log\left( \frac{1}{d} \right) \leq \O{\frac{\log n}{n^{\ep/4}} } .
\]
Combining all of the bounds gives the desired result.
\end{proof}

Our final lemma shows that as long as $B(n)$ doesn't grow too fast, the bound from the previous lemma yields a uniform bound on the entire sample characteristic matrix. This is done by specifying an error threshold for which Lemma~\ref{lem:unionBoundOverGrids} yields a bound that holds with high probability, and then invoking a union bound.

\begin{lemma}
\label{lem:boundCharMatrixEntries}
Let $D_n$ be a sample of size $n$ from the distribution of a pair $(X,Y)$ of jointly distributed random variables. For every $B(n) = \O{n^{1-\ep}}$, there exists an $a > 0$ such that for sufficiently large $n$,
\[
\left| \widehat{M}(D_n)_{k, \ell} - M(X,Y)_{k, \ell} \right|
	\leq \O{ \frac{1}{n^a} }
\]
holds for all $k\ell \leq B(n)$ with probability $P(n) = 1 - o(1)$, where $\widehat{M}(D_n)_{k,\ell}$ is the $k, \ell$-th entry of the sample characteristic matrix and $M(X,Y)_{k,\ell}$ is the $k, \ell$-th entry of the population characteristic matrix of $(X,Y)$.
\end{lemma}

\begin{proof}
Fix $k, \ell$, and any $\alpha$ satisfying $0 < \alpha < \ep/(4 - 2\ep)$. Lemma~\ref{lem:unionBoundOverGrids} implies that with high probability the difference $|\widehat{M}(D_n)_{k, \ell} - M_{k,\ell}|$ is at most
\begin{eqnarray*}
\O{ \frac{\log(k\ell)}{C(n)^\alpha} + \frac{\log(k\ell n)}{n^{\ep / 4}} }
	&\leq& \O{ \frac{\log n}{C(n)^\alpha} + \frac{\log n}{n^{\ep / 4}} } \\
&\leq& \O{ \frac{\log n}{n^{\alpha \ep/2}} + \frac{\log n}{n^{\ep / 4}} }
\end{eqnarray*}
where the first inequality comes from $k\ell \leq B(n)$ and second is because $C(n) = k\ell n^{\ep/2} \geq n^{\ep/2}$. This bound is at most $\O{1/n^a}$ for every $a < \min\{\alpha \ep/2, \ep/4\}$, as desired. It remains only to show that the bound holds with high probability across all $k\ell \leq B(n)$.

Lemma~\ref{lem:unionBoundOverGrids} states that the probability our bound holds for one fixed pair $(k, \ell)$ is at least
\[ 1 - C(n)e^{-\Omega(n/C(n)^{1+2\alpha})} \geq 1 - \O{n}e^{-\Omega(n^u)} \]
for some positive $u$. This is because $C(n) \leq B(n)n^{\ep/2} \leq \O{n^{1 - \ep/2}}$ for large $n$, and so our choice of $\alpha$ ensures that $C(n)^{1 + 2\alpha} = \O{n^{1-u}}$ for some $u > 0$.

We can then perform a union bound over all pairs $k\ell \leq B(n)$: since the number of such pairs can be bounded by a  polynomial in $n$, we have that the desired condition is satisfied for all $k\ell \leq B(n)$ with probability approaching 1.
\end{proof}

We are now ready to prove the main result.

\vspace{0.3cm}
\noindent \textbf{Theorem}
\textit{
Let $f : m^\infty \rightarrow \R$ be uniformly continuous, and assume that $f \circ r_i \rightarrow f$ pointwise. Then for every random variable $(X,Y)$, we have
\[
\left( f \circ r_{B(n)} \right) \left(\widehat{M}(D_n) \right) \rightarrow f(M(X,Y))
\]
in probability where $D_n$ is a sample of size $n$ from the distribution of $(X,Y)$, provided $\omega(1) < B(n) \leq O(n^{1-\ep})$ for some $\ep > 0$.}
\vspace{0.3cm}
\begin{proof}
Let $N$ denote $B(n)$, let $M_N = r_N(M)$, and let $\widehat{M}_N(D_n) = r_N(\widehat{M}(D_n))$. We begin by writing
\begin{eqnarray*}
\left| f \left( \widehat{M}_N(D_n) \right) - f(M) \right| &\leq& \left| f \left( \widehat{M}_N(D_n) \right) - f \left( M_N \right) \right| + \left| f \left( M_N \right) - f(M) \right| \\
	&=& \left| f \left( \widehat{M}_N(D_n) \right) - f \left( M_N \right) \right| + \left| \left( f \circ r_N \right)(M) - f(M) \right|
\end{eqnarray*}
and observing that as $n \rightarrow \infty$, the second term vanishes by the pointwise convergence of $f \circ r_i$ and the fact that $B(n) > \omega(1)$. It therefore suffices to show that the first term converges to 0 in probability. Since $f$ is uniformly continuous, we can establish this via a simple adaptation of the continuous mapping theorem, which says that if the sequence of random variables $R_n \rightarrow R$ in probability, and $g$ is continuous, then $g(R_n) \rightarrow g(R)$ in probability. We replace $R$ with a second sequence, and replace continuity with uniform continuity.

Let $\| \cdot \|$ denote the supremum norm on $m^\infty$, and fix any $z > 0$. Then, for any $\delta > 0$, define
\[
C_{\delta} = \left\{ A \in m^\infty : \exists A' \in m^\infty\mbox{ s.t. } \| A - A' \| < \delta, \left| f(A) - fA') \right| > z \right\} .
\]
This is the set of matrices $A \in m^\infty$ for which it is possible to find, within a $\delta$-neighborhood of $A$, a second matrix that $f$ maps to more than $z$ away from $f(A)$. Because $f$ is uniformly continuous, there exists a $\delta^*$ sufficiently small so that $C_{\delta^*} = \emptyset$.

Suppose that $| f (\widehat{M}_N(D_n)) - f(M_N) | > z$. This means that either $\| \widehat{M}_N(D_n) - M_N \| > \delta^*$, or $M_N \in C_{\delta^*}$. The latter option is impossible since $C_{\delta^*} = \emptyset$, and Lemma~\ref{lem:boundCharMatrixEntries} tells us that $\Pr{\| \widehat{M}_N(D_n) - M_N \| > \delta^*} \rightarrow 0$ as $n$ grows. We therefore have that
\[ \left| f \left( \widehat{M}_N(D_n) \right) - f(M_N) \right| \rightarrow 0 \]
in probability, as desired.
\end{proof}

\section{Proof of Theorem~\ref{thm:continuityofM}}
\label{app:continuityofM}
\setcounter{figure}{0}
In this appendix we prove Theorem~\ref{thm:continuityofM}, reproduced below.

\vspace{0.3cm}
\noindent \textbf{Theorem}
\textit{
Let $\P(\R^2)$ denote the space of random variables supported on $\R^2$ equipped with the metric of statistical distance. The map from $\P(\R^2)$ to $m^\infty$ defined by $(X,Y) \mapsto M(X,Y)$ is uniformly continuous.
}
\vspace{0.3cm}

The proposition below begins our argument with the simple observation that the family of maps consisting of applying any finite grid to some $(X,Y) \in \P(\R^2)$ is uniformly equicontinuous. The reason this holds is that $(X,Y)|_G$ is a deterministic function of $(X,Y)$, and deterministic functions cannot increase statistical distance.

\begin{proposition}
\label{prop:equicontinuityOfGridding}
Let $\mathbb{G}$ be the set of all finite grids. The family $\{ (X,Y) \mapsto (X,Y) |_G : G \in \mathbb{G} \}$ is uniformly equicontinuous on $\P(\R^2)$.
\end{proposition}
\begin{proof}
To establish uniform equicontinuity, we need to show that, given some $(X,Y) \in \P(\R^2)$ and some $\ep > 0$, we can choose $\delta$ to satisfy the continuity condition in a way that does not depend on $G$ or on $(X,Y)$. But because deterministic functions cannot increase statistical distance, we have that if $(X,Y), (X', Y') \in \P$ are at most $\ep$ apart then
\[ \Delta \left( (X,Y)|_G, (X', Y')|_G \right) \leq \Delta \left( (X,Y), (X', Y') \right) = \ep \]
where $\Delta$ denotes statistical distance. Choosing $\delta = \ep$ therefore gives the result.
\end{proof}

At this point it is tempting to try to use continuity properties of discrete mutual information to obtain uniform continuity of the characteristic matrix. And indeed, this strategy does yield that each \textit{individual} entry of the characteristic matrix is a uniformly continuous function. However, to obtain continuity of the entire (infinite) characteristic matrix we need to make a statement about all grid resolutions simultaneously. This is not straightforward because mutual information is only uniformly continuous for a fixed grid resolution, and the family $\{ (X,Y) \mapsto I((X,Y)|_G) : G \in \mathbb{G}\}$ is in fact not even equicontinuous.

The normalization in the definition of $\popMIC$ is what allows us to establish the uniform continuity of the characteristic matrix despite this problem. To see why, suppose we have a distribution over a $k$-by-$\ell$ grid and we are allowed to move at most $\delta$ away in statistical distance for some small $\delta$. The largest change in discrete mutual information that this can cause indeed increases as we increase $k$ and $\ell$. However, it turns out that we can bound the extent of this ``non-uniformity'': the proposition below shows that as we move away from a distribution, the discrete mutual information can change only proportionally to the amount of mass we move, with the proportionality constant bounded by $\log \min \{k, \ell\}$. Because $\log \min \{k, \ell\}$ is the quantity by which we regularize the entries of the characteristic matrix, this is exactly enough to make the normalized matrix continuous. This proposition is the technical heart of our continuity result. And as we show in Corollary~\ref{cor:MIdiscontinuous} when we demonstrate the non-continuity of the non-normalized characteristic matrix mutual information, our bound is tight.

\begin{proposition}
\label{prop:boundedVariationOfI}
Let $I_{k,\ell} : \mathcal{P}(\{1,\ldots,k\} \times \{1, \ldots, \ell\}) \rightarrow \R$ denote the discrete mutual information function on $k$-by-$\ell$ grids. For $0 < \delta \leq 1/4$, the maximal change in $I_{k,\ell}$ over any subset of $\mathcal{P}(\{1,\ldots,k\} \times \{1, \ldots, \ell\})$ of diameter $\delta$ (in statistical distance) is
\[
O\left( \delta \log \left( \frac{1}{\delta} \right) + \delta \log \min \{k, \ell \}  \right).
\]
\end{proposition}
\begin{proof}
Without loss of generality, assume $k \leq \ell$, so that $\log \min \{k, \ell\} = \log k$. Let $(X,Y)$ and $(X', Y')$ be two random variables distributed over $\{1, \ldots, k\} \times \{1, \ldots, \ell\}$ that are at most $\delta$ apart in statistical distance. Using $I(X,Y) = H(Y) - H(Y | X)$, we can express the difference between the mutual information of these two pairs of random variables as
\[
\left| I(X,Y) - I(X',Y') \right| \leq \left| H(Y) - H(Y') \right| + \left| H(Y | X) - H(Y' | X') \right| .
\]

We now use Lemma~\ref{lem:entropyBoundForChangedMass}, which relates movement of probability mass to changes in entropy and is proven in Appendix~\ref{appendix:technical}, to separately bound each of the terms on the right hand side. Straightforward application of the lemma to $| H(Y) - H(Y') |$ shows that it is at most $2 H_b(2 \delta) + 3\delta \log k$, where $H_b(\cdot)$ is the binary entropy function. Since $H_b(x) \leq O(x\log(1/x))$ for $x$ small, this is $O(\delta \log(1/\delta) + \delta \log k)$.

Bounding the term with the conditional entropies is more involved. Let $p_x = \Pr{X = x}$, and let $p_x' = \Pr{X' = x}$. We have
\begin{eqnarray}
\left| H(Y | X) - H(Y' | X') \right| &=& \sum_x \left| p_x H(Y | X = x) - p_x' H(Y' | X' = x) \right| \nonumber \\
	&\leq& \sum_x \left( p_x \left| H(Y | X = x) - H(Y' | X' = x) \right| + \left| p_x' - p_x \right| H(Y' | X' = x) \right) \nonumber \\
	&=& \sum_x p_x \left| H(Y | X = x) - H(Y' | X' = x) \right| + \sum_x \left| p_x' - p_x \right| \log k \nonumber \\
	&\leq& \sum_x p_x \left| H(Y | X = x) - H(Y' | X' = x) \right| + \delta \log k \label{line:boundMutualInfoChange}
\end{eqnarray}
where the last line is because $\sum_x |p_x - p_x'| \leq \delta$ and $H(Y' | X' = x) \leq \log k$.

Now let $\delta_{x+}$ be the magnitude of all the probability mass entering any cell in column $x$, let $\delta_{x-}$ be the magnitude of all the probability mass leaving any cell in column $x$, and let $\delta_x = \delta_{x+} + \delta_{x-}$. Using this notation, we can again apply Lemma~\ref{lem:entropyBoundForChangedMass} to obtain
\begin{eqnarray*}
\sum_x p_x \left| H(Y | X = x) - H(Y' | X' = x) \right| &\leq& \sum_x p_x \left( 2H_b \left( \frac{2\delta_x}{p_x} \right) + 3 \frac{\delta_x}{p_x} \log k \right) \\
	&=& 2\sum_x p_x H_b \left( \frac{2\delta_x}{p_x} \right) + 3 \sum_x \delta_x \log k \\
	&\leq& 2\sum_x p_x H_b \left( \frac{2\delta_x}{p_x} \right) + 3 \delta \log k \\
	&\leq& 2 H_b(2 \delta) + 3 \delta \log k
\end{eqnarray*}
where the last line is by application of Lemma~\ref{lem:boundWeightedAverageOfEntropy} from the appendix, which bounds weighted sums of binary entropies.

Combining this with Line~\eqref{line:boundMutualInfoChange} gives that
\[ \left| H(Y | X) - H(Y' | X') \right| \leq 2H_b(2\delta) + 4\delta \log k \]
which, together with the bound on $\left| H(Y) - H(Y') \right|$ and the fact that $H_b(X) \leq O(x \log(1/x))$ for $x$ small, gives the result.
\end{proof}

Having bounded the extent to which variation in mutual information depends on grid resolution, we are now ready to show the uniform continuity of the characteristic matrix.

\vspace{0.3cm}
\noindent \textbf{Theorem}
\textit{
Let $\P(\R^2)$ denote the space of random variables supported on $\R^2$ equipped with the metric of statistical distance. The map from $\P(\R^2)$ to $m^\infty$ defined by $(X,Y) \mapsto M(X,Y)$ is uniformly continuous.
}
\vspace{0.3cm}
\begin{proof}
We complete the proof in three steps. First, we show that a certain family of functions $F$ is uniformly equicontinuous. Second, we use this to show that a different family $F'$ consisting of functions of the form $\sup_{g \in A} g$ with $A \subset F$ is uniformly equicontinuous. Finally, we argue that since the entries of $M(X,Y)$ consist of the functions in $F'$, this is sufficient to establish the result.

Define
\[
F = \left\{ (X,Y) \mapsto \frac{I_{k,\ell}((X,Y)|_G)}{\log \min \{k, \ell \} } : k, \ell \in \Z_{> 1}, G \in G(k, \ell) \right\} .
\]
$F$ is uniformly equicontinuous by the following argument. Given some $\ep > 0$, we know (Proposition \ref{prop:equicontinuityOfGridding}) that for any $(X', Y')$ in an $\ep$-ball around $(X,Y)$, $(X',Y')|_G$ will remain $\ep$ of $(X,Y)|_G$ for any $G$. Proposition~\ref{prop:boundedVariationOfI} then tells us that if $\ep$ is sufficiently small then the distance between $I_{k,\ell}((X',Y')|_G)$ and $I_{k,\ell}((X,Y)|_G)$ will be at most
\[
O\left( \ep \log(1/\ep) + \ep \log \min \{k,\ell\} \right) .
\]
After the normalization, this becomes at most $O(\ep(\log(1/\ep) + 1))$, which goes to 0 (uniformly with respect to $(X,Y)$) as $\ep$ approaches 0, as desired.

Next, define 
\[
F' = \left\{ (X,Y) \mapsto M(X,Y)_{k,\ell} : k, \ell \in \Z_{> 1} \right\} .
\]
Each map in $F'$ is of the form $\sup_{g \in A} g$ for some $A \subset F$. Therefore, for a given $\ep > 0$, whatever $\delta$ establishes the uniform equicontinuity for $F$ can be used to establish continuity of all the functions in $F'$. (To see this: $\sup_{g \in A} g$ can't increase by more than $\ep$ if no $g$ increases by more than $\ep$, and $\sup_{g \in A}g$ is also lower bounded by any of the $g$'s, so it can't decrease by more than $\ep$ either.) Since we can use the same $\delta$ for all of the maps in $F'$, they therefore form a uniformly equicontinuous family.

Finally, the $\delta$ provided by the uniform equicontinuity of $F'$ also ensures that $M(X',Y')$ is within $\ep$ of $M(X,Y)$ in the supremum norm, thus giving the uniform continuity of $(X,Y) \mapsto M(X,Y)$.
\end{proof}

\section{Proof of Proposition~\ref{prop:MIdiscontinuous}}
\label{app:MIdiscontinuous}
\setcounter{figure}{0}
\vspace{0.3cm}
\noindent \textbf{Theorem}
\textit{
For some function $N(k, \ell$), let $M^N$ be the characteristic matrix with normalization $N$, i.e.,
\[
M^N(X,Y) = \frac{I^*((X,Y), k, \ell)}{N(k, \ell)} .
\]
If $N(k, \ell) = o(\log \min \{k,\ell\})$ along some infinite path in $\mathbb{N} \times \mathbb{N}$, then $M^N$ and $\sup M^N$ are not continuous as functions of $\mathcal{P}([0,1]\times[0,1]) \subset \P(\R^2)$.
}
\vspace{0.3cm}
\begin{proof}
Consider a random variable $Z$ uniformly distributed on $[0,1/2]^2$. Because $Z$ exhibits statistical independence, $I^*(Z, k,\ell)$ is zero for all $k, \ell$. Now define $Z_\ep$ to be uniformly distributed on $[0,1/2]^2$ with probability $1-\ep$ and uniformly distributed on the line from $(1/2, 1/2)$ to $(1,1)$ with probability $\ep$.

We lower-bound $I^*(Z_\ep, k, \ell)$. Without loss of generality suppose that $k \leq \ell$, and consider a grid that places all of $[0,1/2]^2$ into one cell and uniformly partitions the set $[1/2,1]^2$ into $k-1$ rows and $k-1$ columns. By considering just the rows/columns in the set $[1/2,1]^2$ we see that this grid gives a mutual information of at least $\ep \log(k-1)$. Thus, we have that for all $k, \ell$,
\[
I^*(Z_\ep, k, \ell) \geq \ep \log \min \{k-1, \ell - 1\} .
\]

This implies that the limit of $M^N(Z_\ep)$ along $P$ is $\infty$, and so the distance between $M^N(Z)$ and $M^N(Z_\ep)$ in the supremum norm is infinite.
\end{proof}

\section{Proof of Theorem~\ref{thm:explicitValueOfBoundary}}
\label{app:explicitValueOfBoundary}
\setcounter{figure}{0}
\vspace{0.3cm}
\noindent \textbf{Theorem}
\textit{
Let $M$ be a population characteristic matrix. Then $M_{k, \uparrow}$ equals
\[ \max_{P \in P(k)} \frac{I(X, Y|_P)}{\log{k}} \]
where $P(k)$ denotes the set of all partitions of size at most $k$.
}
\vspace{0.3cm}
\begin{proof}
Define
\[ M_{k,\uparrow}^* = \max_{P \in P(k)} \frac{I(X, Y|_P)}{\log{k}} .\]
We wish to show that $M_{k,\uparrow}^*$ is in fact equal to $M_{k,\uparrow}$. To show that $M_{k,\uparrow} \leq M_{k,\uparrow}^*$, we observe that for every $k$-by-$\ell$ grid $G = (P,Q)$, where $P$ is a partition into rows and $Q$ is a partition into columns, the data processing inequality gives $I((X, Y)|_G) \leq I(X, Y|_P)$. Thus $M_{k,\ell} \leq M_{k,\uparrow}^*$ for $\ell \geq k$, implying that
\[ M_{k,\uparrow} = \lim_{\ell \rightarrow \infty} M_{k,\ell} \leq M_{k,\uparrow}^* .\]

It remains to show that $M_{k,\uparrow}^* \leq M_{k,\uparrow}$. To do this, we let $P$ be any partition into $k$ rows, and we define $Q_\ell$ to be an equipartition into $\ell$ columns. We let
\[ M_{k,\ell,P}^* = \frac{I(X|_{Q_\ell}, Y|_P)}{\log{k} } .\]
Since $M_{k,\ell,P}^* \leq M_{k,\ell}$ when $\ell \geq k$, we have that for all $P$
\[ \frac{I(X, Y|_P)}{\log{k}} = \lim_{\ell \rightarrow \infty} M_{k,\ell,P}^* \leq \lim_{\ell \rightarrow \infty} M_{k,\ell} = M_{k,\uparrow} \]
which gives that
\[ M_{k,\uparrow}^* = \sup_P \frac{I(X, Y|_P)}{\log{k}} \leq M_{k,\uparrow} \]
as desired.
\end{proof}

\section{Proof of Theorem~\ref{thm:alg_infinite_data}}
\label{app:alg_infinite_data}
\setcounter{figure}{0}
\vspace{0.3cm}
\noindent \textbf{Theorem}
\textit{
Given a random variable $(X,Y)$, $M_{k, \uparrow}$ (resp. $M_{\uparrow, \ell}$) is computable to within an additive error of $O(k \ep \log(1/(k\ep))) + E$ (resp. $O(\ell \ep \log(1/(\ell\ep))) + E$) in time $O(kT(E)/\ep)$ (resp. $O(\ell T(E) / \ep)$), where $T(E)$ is the time required to numerically compute the mutual information of a continuous distribution to within an additive error of $E$.
}
\vspace{0.3cm}
\begin{proof}
Without loss of generality we prove the claim only for $M_{k, \uparrow}$. Given $0 < \ep < 1$, we would like a partition into rows $P$ of size at most $k$ such that $I(X, Y|_P)$ is maximized. We would like to use \algname{OptimizeXAxis} for this purpose, but while our search problem is continuous, \algname{OptimizeXAxis} can only perform a discrete search over sub-partitions of some master partition $\Pi$. We therefore set $\Pi$ to be an equipartition into $1/\ep$ rows and show that this gets us close enough to achieve the desired result.

With $\Pi$ as described, the \algname{OptimizeXAxis} provides in time $O(kT(E)/\ep)$ a partition $P_0$ into at most $k$ rows such that $I \left(X,Y|_{P_0} \right)$ is maximized, subject to $P_0 \subset \Pi$, to within an additive error of $E$. To prove the claim then, we must show that the loss we incur by restricting to sub-partitions of $\Pi$ costs us at most $O(k \ep \log(1/(k\ep)))$. In other words, we must show that
\[ I \left( X, Y|_P \right) - I \left( X, Y|_{P_0} \right) \leq O(k \ep) \]
where $P$ is an optimal partition into rows. Note that we have omitted the absolute value above, since by the optimality of $P$, $I \left( X, Y|_P \right) \geq I \left( X, Y|_{P_0} \right)$ always.

We prove the desired bound by showing that there exists some $P' \subset \Pi$ such that the mutual information of $(X, Y|_{P'})$ is $O(k\ep\log(1/(k\ep)))$-close to that achieved with $(X, Y|_P)$. Since $P' \subset \Pi$ gives us that $I\left( X, Y|_{P_0} \right) \geq I\left( X, Y|_{P'} \right)$, we may then conclude that $I \left( X, Y|_P \right) - I \left( X, Y|_{P_0} \right)$ is at most $O(k\ep \log(1/(k\ep)))$.

We construct $P'$ by simply moving replacing every horizontal line in $P$ with the horizontal line in $\Pi$ closest to it. Since there are at most $k-1$ horizontal lines in $P$, and each such line is contained in a row of $\Pi$ containing $1/\ep$ probability mass, performing this operation moves at most $(k-1)\ep$ probability mass. In other words, the statistical distance between $(X, Y|_{P'})$ and $(X, Y|_P)$ is at most $(k-1)\ep \leq k\ep$. Thus, for sufficiently small $\ep$, Proposition~\ref{prop:boundedVariationOfI}, proven in Appendix~\ref{app:continuityofM}, can be used to show that
\[
\left| I\left(X, Y|_{P'} \right) - I\left(X, Y|_P\right) \right| \leq O \left( k\ep \log \left( \frac{1}{k\ep} \right) + k\ep \log \left( \frac{1}{\ep} \right) \right)
\]
which yields the desired result.
\end{proof}

\begin{remark}
We do not explore here the details of the numerical integration associated with the above theorem, since the error introduced by the numerical integration is independent of the algorithm being proposed. However, standard numerical integration methods can be used to make this error arbitrarily small with an understood complexity tradeoff (see, e.g., \cite{stoer1980numerical}).
\end{remark}

\section{Proof of Theorem~\ref{thm:same_boundary}}
\label{app:same_boundary}
\setcounter{figure}{0}
\vspace{0.3cm}
\noindent \textbf{Theorem}
\textit{
Let $(X,Y)$ be jointly distributed random variables. Then $\partial [M] = \partial M$.
}
\vspace{0.3cm}
\begin{proof}
Without loss of generality, we show that $[M]_{k, \uparrow} = M_{k, \uparrow}$. Fix any partition into rows $P$. If $Q_\ell$ is an equipartition into $\ell$ columns then
\[
\lim_{\ell \rightarrow \infty} I(X|_{Q_\ell}, Y|_P) = I(X, Y|_P) ,
\]
because the continuous mutual information equals the limit of the discrete mutual information with increasingly fine partitions. (See, e.g., Chapter 8 of \cite{Cover2006} for a proof of this.) This means that, letting $P(k)$ denote the set of all partitions of size at most $k$, we have
\[
[M]_{k, \uparrow} = \max_{P \in P(k)} \frac{I(X, Y|_P)}{\log k} = M_{k, \uparrow}
\]
where the second equality follows from Proposition~\ref{thm:explicitValueOfBoundary}.
\end{proof}

\section{Consistency of \texorpdfstring{$\MICestE$}{MICe} in estimating \texorpdfstring{$\popMIC$}{MIC*}}
\label{app:MICeconsistent}
\setcounter{figure}{0}
The consistency of $\MICestE$ for estimating $\popMIC$ can be established using the same technical lemmas that we used to show that $\MIC \rightarrow \popMIC$. Specifically, we can use Lemma~\ref{lem:unionBoundOverGrids}, which bounds the difference, for all $k$-by-$\ell$ grids $G$, between the sample quantity $I(D_n|_G)$ and the population quantity $I((X,Y)|_G)$ with high probability, where $D_n$ is a sample of size $n$ from $(X,Y)$. That lemma yields the following fact about the sample equicharacteristic matrix, whose proof is similar to that of Lemma~\ref{lem:boundCharMatrixEntries}.

\begin{lemma}
\label{lem:boundEquiCharMatrixEntries}
Let $D_n$ be a sample of size $n$ from the distribution of a pair $(X,Y)$ of jointly distributed random variables. For every $B(n) = \O{n^{1-\ep}}$, there exists an $a > 0$ such that for sufficiently large $n$,
\[
\left| \widehat{[M]}(D_n)_{k, \ell} - [M](X,Y)_{k, \ell} \right|
	\leq \O{ \frac{1}{n^{a}} }
\]
holds for all $k\ell \leq B(n)$ with probability $P(n) = 1 - o(1)$, where $\widehat{[M]}(D_n)_{k,\ell}$ is the $k, \ell$-th entry of the sample equicharacteristic matrix and $[M](X,Y)_{k,\ell}$ is the $k, \ell$-th entry of the population equicharacteristic matrix of $(X,Y)$.
\end{lemma}

In the case of $\MIC$, we proceeded to apply abstract continuity considerations to obtain our consistency theorem (Theorem~\ref{thm:estimator}) from a result analogous to the above lemma. A similar argument shows us that, in the case of the equicharacteristic matrix as well, we can estimate a large class of functions of the matrix in the same way. This is stated formally in the theorem below. As before, we let $m^\infty$ be the space of infinite matrices equipped with the supremum norm, and given a matrix $A$ the projection $r_i$ zeros out all the entries $A_{k,\ell}$ for which $k\ell > i$.

\vspace{0.3cm}
\noindent \textbf{Theorem}
\textit{
Let $f : m^\infty \rightarrow \R$ be uniformly continuous, and assume that $f \circ r_i \rightarrow f$ pointwise. Then for every random variable $(X,Y)$, we have
\[
\left( f \circ r_{B(n)} \right) \left(\widehat{[M]}(D_n) \right) \rightarrow f([M](X,Y))
\]
in probability where $D_n$ is a sample of size $n$ from the distribution of $(X,Y)$, provided $\omega(1) < B(n) \leq O(n^{1-\ep})$ for some $\ep > 0$.
}

\section{The \algname{EquicharClump} algorithm}
\label{app:computation}
\setcounter{figure}{0}
In Theorem~\ref{thm:algorithm_with_c}, we sketched an algorithm called \algname{EquicharClump} for approximating the sample equicharacteristic matrix that is more efficient than the naive computation. In this appendix, we describe the algorithm in detail, bound its runtime, and show that it indeed yields a consistent estimator of $\popMIC$ from finite samples as well as a consistent independence test when used to compute the total information coefficient. We then present some empirical results characterizing the sensitivity of the algorithm to its speed-versus-optimality parameter $c$.

The results in this section can be summarized as follows: let $(X,Y)$ be a pair of jointly distributed random variables, and let $D_n$ be a sample of size $n$ from the distribution of $(X,Y)$. For every $c \geq 1$, there exists a matrix $\{\widehat{M}\}^c(D_n)$ such that
\begin{enumerate}
\item There exists an algorithm \algname{EquicharClump} for computing $r_B(\{\widehat{M}\}^c(D_n))$ in time $O(n + B^{5/2})$, which equals $O(n + n^{5(1-\ep)/2})$ when $B(n) = O(n^{1-\ep})$.
\item The function
\[
\widetilde{\MICestE}_{,B}(\cdot) = \max_{k\ell \leq B(n)} \{\widehat{M}\}^c(\cdot)_{k,\ell}
\]
is a consistent estimator of $\popMIC$ provided $\omega(1) < B(n) \leq O(n^{1-\ep})$ for some $\ep >0$.
\item The function
\[\widetilde{\TICestE}_{,B}(\cdot) = \sum_{k\ell \leq B(n)} \{\widehat{M}\}^c(\cdot)_{k,\ell}
\]
yields a consistent right-tailed test of independence provided $\omega(1) < B(n) \leq O(n^{1-\ep})$ for some $\ep >0$
\end{enumerate}

We will prove these results in order.

\subsection{Algorithm description and analysis of runtime}
We begin by describing the algorithm and bound its runtime simultaneously. As in the proof of Theorem~\ref{thm:algorithm_naive}, we bound the runtime required to approximately compute only the $k,\ell$-th entries of $\{\widehat{M}\}^c(D_n)$ satisfying $k \leq \ell, k\ell \leq B$. To do this, we analyze two portions of $\{\widehat{M}\}^c(D_n)$ separately: we first consider the case $\ell \geq \sqrt{B}$, in which we must compute the entries corresponding to all the pairs $\{(2, \ell), \ldots, (B/\ell, \ell)\}$. We then consider $\ell < \sqrt{B}$, in which case we need only compute the entries $\{(2, \ell), \ldots, (\ell, \ell)\}$ since the additional pairs would all have $k > \ell$.

For the case of $\ell \geq \sqrt{B}$, as in the previous theorem we can simultaneously compute using \algname{OptimizeXAxis} the entries corresponding to all the pairs $\{ (2, \ell), \ldots, (B/\ell, \ell) \}$ in time $O(|\Pi|^2 (B/\ell) \ell) = O(|\Pi|^2 B)$, which equals $O(c^2B^3/\ell^2)$ when we set $\Pi$ to be an equipartition of size $cB/\ell$. Doing this for $\ell = \sqrt{B}, \ldots, B/2$ gives a contribution of the following order to the runtime.
\begin{align*}
O(c^2B^3) \sum_{\ell = \sqrt{B}}^{B/2} \frac{1}{\ell^2} &= O\left( c^2B^3 \right) O\left( \frac{1}{\sqrt{B}} \right) \\
    &= O(c^2 B^{5/2})
\end{align*}

For the case of $\ell < \sqrt{B}$, we can simultaneously compute using \algname{OptimizeXAxis} the entries corresponding to all the pairs $\{ (2, \ell), \ldots, (\ell, \ell) \}$ in time $O(|\Pi|^2 \ell^2)$ which equals $O(c^2 \ell^4) \leq O(c^2 B^2)$ when we set $\Pi$ to be an equipartition of size $c\ell$. Summing over the $O(\sqrt{B})$ possible values of $\ell$ with $\ell < \sqrt{B}$ gives an upper bound of $O(c^2 B^{5/2})$.

\subsection{Consistency}
Let $(X,Y)$ be a pair of jointly distributed random variables. For a sample $D_n$ of size $n$ from the distribution of $(X,Y)$ and a speed-versus-optimality parameter $c \geq 1$, let $\{\widehat{M}\}^c(D_n)$ denote the matrix computed by \algname{EquicharClump}. (Notice the use of curly braces to differentiate this from the sample equicharacteristic matrix $\widehat{[M]}$.) We show here that $\max_{k \ell \leq B(n)} \{\widehat{M}\}^c(D_n)_{k,\ell}$ is a consistent estimator of $\popMIC(X,Y)$, and correspondingly that $\sum_{k \ell \leq B(n)} \{\widehat{M}\}^c(D_n)_{k,\ell}$ yields a consistent independence test.

The key to both consistency results is that, though in calculating the $k,\ell$-th entry of $\{\widehat{M}\}^c(D_n)$ the algorithm only searches for optimal partitions that are sub-partitions of some equipartition, the size of the equipartition used always grows as $n$, $k$, and $\ell$ grow large. Therefore, in the limit this additional restriction does not hinder the optimization. We present this argument by introducing a population object called the {\em clumped equicharacteristic matrix}. We observe that this matrix is the limit of the \algname{EquicharClump} procedure as sample size grows, and then show that the supremum and partial sums of this matrix have the necessary properties.

\begin{definition}
Let $(X,Y)$ be jointly distributed random variables and fix some $c \geq 1$. Let
\[
I^{\{c*\}}((X,Y), k, \ell) = \max_G I((X,Y)|_G)
\]
where the maximum is over $k$-by-$\ell$ grids whose larger partition is an equipartition and whose smaller partition must be contained in an equipartition of size $c \cdot \max \{k, \ell\}$. The {\em clumped equicharacteristic matrix} of $(X,Y)$, denoted by $\{M\}^c(X,Y)$, is defined by
\[
\{M\}^c(X,Y)_{k,\ell} = \frac{I^{\{c*\}}((X,Y), k, \ell)}{\log \min \{k, \ell\}}
\]
\end{definition}
Notice that curly braces differentiate the quantities $I^{\{c*\}}$ and $\{M\}^c$ defined above from the corresponding equicharacteristic matrix quantities $I^{[*]}$ and $[M]$.

The following two results, which we state without proof, characterize the convergence of the output of \algname{EquicharClump} to the clumped equicharacteristic matrix. These lemmas can be shown using Lemma~\ref{lem:unionBoundOverGrids}, which simultaneously bounds the difference, for all $k$-by-$\ell$ grids $G$, between the sample quantity $I(D_n|_G)$ and the population quantity $I((X,Y)|_G)$ with high probability over the sample $D_n$ of size $n$ from $(X,Y)$.

\begin{lemma}
\label{lem:uniform_convergence_clump}
Let $D_n$ be a sample of size $n$ from the distribution of a pair $(X,Y)$ of jointly distributed random variables. For every $B(n) = \O{n^{1-\ep}}$, there exists an $a > 0$ such that for sufficiently large $n$,
\[
\left| \{\widehat{M}\}^c(D_n)_{k, \ell} - \{M\}^c(X,Y)_{k, \ell} \right|
	\leq \O{ \frac{1}{n^{a}} }
\]
holds for all $k, \ell \leq \sqrt{B(n)}$ with probability $P(n) = 1 - o(1)$, where $\{\widehat{M}\}^c(D_n)$ denotes the matrix computed by the \algname{EquicharClump} algorithm with parameter $c$ on the sample $D_n$.
\end{lemma}
Notice that the error bound provided by the above lemma holds not for $k\ell \leq B(n)$ as in the analogous Lemma~\ref{lem:boundCharMatrixEntries} and Lemma~\ref{lem:boundEquiCharMatrixEntries}, but rather for the smaller region defined by $k, \ell \leq \sqrt{B(n)}$. However, though we do not have uniform convergence outside the region $k, \ell \leq \sqrt{B(n)}$, we do nevertheless have pointwise convergence there, as stated below.

\begin{lemma}
\label{lem:pointwise_convergence_clump}
Fix $k, \ell \geq 2$. Let $D_n$ be a sample of size $n$ from the distribution of a pair $(X,Y)$ of jointly distributed random variables. For every $B(n) > \omega(1)$, we have that
\[
\{\widehat{M}\}^c(D_n)_{k,\ell} \rightarrow \{M\}^c(X,Y)_{k,\ell}
\]
in probability as $n$ grows, where $\{\widehat{M}\}^c(D_n)$ denotes the matrix computed by the \algname{EquicharClump} algorithm with parameter $c$ on the sample $D_n$.
\end{lemma}

\subsubsection{Consistency for estimating \texorpdfstring{$\popMIC$}{MIC*}}
The consistency of $\{\widehat{M}\}^c(D_n)$ for estimating $\popMIC$ follows from the following property of the clumped equicharacteristic matrix $\{M\}^c$, for which we state a proof sketch.
\begin{proposition}
\label{prop:clumped_equichar_mic}
Let $(X,Y)$ be a pair of jointly distributed random variables. Then we have $\sup \{M\}^c(X,Y) = \popMIC(X,Y)$.
\end{proposition}
\begin{proof}(Sketch) Let $\{M\}^c = \{M\}^c(X,Y)$, and let $M = M(X,Y)$ be the characteristic matrix. Fix $k$, and consider the limit $\{M\}^c_{k,\ell}$ as $\ell$ grows. The grid chosen for the $k, \ell$-th entry when $\ell > k$ will contain an equipartition $P_\ell$ of size $\ell$ on the x-axis, and a partition $Q_\ell$ of size $k$ on the y-axis that is optimal subject to the restriction that $Q_\ell$ be contained in an equipartition of size $c\ell$. As $\ell$ grows large, the equipartition $P_\ell$ on the first axis will become finer and finer until in the limit $X|_{P_\ell} \rightarrow X$. And the partition $Q_\ell$ will be chosen from a finer and finer equipartition, so that in the limit it approaches an unconditionally optimal partition $Q$ of size $k$. The convergence of $Q_\ell$ to the optimal partition $Q$ of size $k$ can be shown to be uniform using Proposition~\ref{prop:boundedVariationOfI}. This implies that
\[
\{M\}^c_{k,\uparrow} = \lim_{\ell\rightarrow\infty}\{M\}^c_{k,\ell} = \max_{P \in P(k)} \frac{I(X, Y|_P)}{\log k}
\]
where $P(k)$ denotes the set of all partitions of size at most $k$. Therefore, the boundary $\partial \{M\}^c$ of $\{M\}^c$ equals the boundary $\partial M$ of $M$. Since $\popMIC(X,Y) = \sup \partial M$ (Theorem~\ref{thm:MICinTermsOfBoundary}), this implies that
\[
\sup \{M\}^c \geq \sup \partial \{M\}^c = \sup \partial M = \popMIC(X,Y) .
\]
On the other hand, $\{M\}^c \leq M$ element-wise since the optimization for the $k,\ell$-th entry of $\{M\}^c$ is performed over a subset of the grids searched for the $k,\ell$-th entry of $M$. This means that $\sup \{M\}^c \leq \sup M = \popMIC(X,Y)$.
\end{proof}

This fact, together with the pointwise convergence of $\{\widehat{M}\}^c(D_n)$ to $\{M\}^c$, suffices to establish the consistency we seek via standard continuity arguments, which we give in the abstract lemma below. The lemma applies to a double-indexed sequence indexed by $i$ and $j$; in our argument, the index $i$ corresponds to position in the equicharacteristic matrix, and the index $j$ corresponds to sample size. The sequence $A$ corresponds to the output of the \algname{EquicharClump} algorithm, the sequence $a$ corresponds to the clumped equicharacteristic matrix, and the sequence $B$ corresponds to the sample equicharacteristic matrix.

\begin{lemma}
\label{lem:convergence_of_approx_equichar}
Let $\{A_{ij}\}_{i,j=1}^\infty$ and $\{B_{ij}\}_{i,j=1}^\infty$ be sequences of random variables, and let $\{a_i\}_{i=1}^\infty$ be a non-stochastic sequence. Assume that the following conditions hold.
\begin{enumerate}
\item $A_{ij} \leq B_{ij}$ almost surely
\item For every $i$, $A_{ij} \rightarrow a_i$ in probability
\item $B'_j = \max_{i \leq j} B_{ij}$ satisfies $B'_j \rightarrow \sup \{a_i\}$ in probability
\end{enumerate}
Then $A'_j = \max_{i \leq j} A_{ij}$ converges in probability to $\sup \{a_i\}$ as well.
\end{lemma}
\begin{proof}
Let $a = \sup \{a_i\}$. We give the proof for the case that $a < \infty$. However, it is easily adapted to the infinite case. We must show that for every $\ep > 0$ and every $0 < p \leq 1$, there exists some $N$ such that $\Prsm{|A'_j - a| < \ep} > p$ for all $j \geq N$. By the definition of $a$, we know that there exists some $k$ such that $|a_k - a| < \ep / 2$. Also, by the convergence of $A_{kj}$ to $a_k$, there exists some $m$ such that $\Prsm{|A_{kj} - a_k| < \ep / 2} > 1-p$ for all $j \geq m$. Thus, with probability at least $1-p$, we have
\begin{align*}
\left| A_{kj} - a \right| &\leq \left| A_{kj} - a_k \right| + \left| a_k - a \right| \\
    &\leq \ep
\end{align*}
for all $j \geq m$.

Next, we observe that since $A'_j \geq A_{kj}$ for $j \geq k$, the above inequality implies that for $j \geq \max\{m, k\}$ we have $\Prsm{A'_j > a - \ep} > 1-p$. It remains only to show that $A'_j$ doesn't get too large, but this follows from the fact that $A'_j \leq B'_j$ and $B'_j \rightarrow a$ in probability. Specifically, we are guaranteed some $N \geq \max\{m, k\}$ such that $\Prsm{B'_j < a + \ep} > 1-p$ for $j \geq N$. Since $B'_j < a + \ep$ implies $A'_j < a + \ep$, we have that $\Prsm{|A'_j - a| < \ep} > 1-p$ for $j \geq N$, as desired.
\end{proof}

\begin{proposition}
The function
\[
\widetilde{\MICestE}_{,B}(\cdot) = \max_{k\ell \leq B(n)} \{\widehat{M}\}^c(\cdot)_{k,\ell}
\]
is a consistent estimator of $\popMIC$ provided $\omega(1) < B(n) \leq O(n^{1-\ep})$ for some $\ep >0$, where $\{\widehat{M}\}^c(\cdot)$ is the output of the the \algname{EquicharClump} algorithm.
\end{proposition}
\begin{proof}
Let $(X,Y)$ be a pair of jointly distributed random variables, and let $D_n$ be a sample of size $n$ from the distribution of $(X,Y)$. Let $\{(k_i, \ell_i) \}_{i=1}^\infty \subset \Z^+ \times \Z^+$ be a sequence of coordinates with the property that for every number $B$ there exists an index $q(B)$ such that
\[
\left\{ (k_i, \ell_i) : i \leq q(B) \right\} = \left\{ (k, \ell) : k \ell \leq B \right\} .
\]

We define $B_{ij} = \widehat{[M]}(D_j)_{k_i, \ell_i}$, i.e., $B_{ij}$ is the $k_i, \ell_i$-th entry of the sample characteristic matrix evaluated on a sample of size $j$. We analogously define $A_{ij} = \{\widehat{M}\}^c(D_j)_{k_i, \ell_i}$, and we define $a_i = \{M\}^c(X,Y)_{k_i, \ell_i}$. We observe that by Proposition~\ref{prop:clumped_equichar_mic}, $\sup a_i = \sup \{M\}^c(X,Y) = \popMIC$.

It is straightforward to see that $A_{ij} \leq B_{ij}$. Additionally, Lemma~\ref{lem:pointwise_convergence_clump} shows that $A_{ij} \rightarrow a_i$ in probability, and Corollary~\ref{cor:mice_consistent}, which states that $\MICestE$ is a consistent estimator of $\popMIC$, shows that $B'_j = \max_{i \leq j} B_{ij} \rightarrow \popMIC(X,Y)$. In the notation of the lemma, it therefore follows that $A'_j = \max_{i \leq j} A_{ij}$ converges in probability to $\popMIC(X,Y)$ as well. But this means that the sub-sequence 
\[
A'_{q(B(n))} = \max_{i \leq q(B(n))} \{\widehat{M}\}^c(D_{q(B(n))})_{k_i, \ell_i} = \max_{k\ell \leq B(n)} \{\widehat{M}\}^c(D_{q(B(n))})_{k,\ell}
\]
converges in probability to $\popMIC(X,Y)$, which implies the result since the sequence $A'_j$ is monotone.
\end{proof}

\subsubsection{Consistency for total information coefficient}
Similarly to the consistency argument for $\popMIC$, we begin by exhibiting the relevant property of the population clumped equicharacteristic matrix.
\begin{proposition}
\label{prop:clumped_equichar_tic}
Let $(X,Y)$ be a pair of jointly distributed random variables. If $X$ and $Y$ are statistically independent, then $\{M\}^c(X,Y) \equiv 0$. If not, then there exists some $a > 0$ and some integer $\ell_0 \geq 2$ such that
\[
\{M\}^c(X,Y)_{k,\ell} \geq \frac{a}{\log \min \{k, \ell\}}
\]
either for all $k \geq \ell \geq \ell_0$, or for all $\ell \geq k \geq \ell_0$.
\end{proposition}
\begin{proof}(Sketch)
Let $\{M\}^c = \{\widehat{M}\}^c(X,Y)$. Under independence, every entry of $\{M\}^c$ is zero since $I((X,Y)|_G) = 0$ for any grid $G$. For the case of dependence, the argument is identical to that given in the proof of Proposition~\ref{prop:nonzero_submatrix}. Specifically, it can be shown that there exists some index $\ell_0$, taken without loss of generality to be a column index, and some $r > 0$ such that all but finitely many of the entries in the $\ell_0$-column are at least $r$. It can then be shown that for large $k$, the entries $(k,\ell_0), (k, \ell_0+1), \ldots, (k, k)$ have non-decreasing values of $I^{[c*]}$. This establishes the claim for $a = r \log \ell_0$.
\end{proof}

We now show that the above result, together with the uniform convergence of $\{\widehat{M}\}^c(D_n)$ to $\{M\}^c(X,Y)$, implies the consistency we seek.
\begin{proposition}
The function
\[
\widetilde{\TICestE}_{,B}(\cdot) = \sum_{k\ell \leq B(n)} \{\widehat{M}\}^c(\cdot)_{k,\ell}
\]
yields a consistent right-tailed test of independence provided $\omega(1) < B(n) \leq O(n^{1-\ep})$ for some $\ep >0$, where $\{\widehat{M}\}^c(\cdot)$ is the output of the the \algname{EquicharClump} algorithm.
\end{proposition}
\begin{proof}
Let $(X,Y)$ a pair of jointly distributed random variables, and let $D_n$ be a sample of size $n$ from the distribution of $(X, Y)$. It suffices to show consistency for any deterministic monotonic function of the statistic in question. We therefore choose to analyze $\widetilde{\TICestE}_{,B}(D_n) \log(B(n)) / B(n)$.

For the null hypothesis in which $X$ and $Y$ are independent, we observe that since $\{\widehat{M}\}^c(D_n) \leq \widehat{[M]}(D_n)$ element-wise, $0 \leq \widetilde{\TICestE}_{,B}(D_n) \leq \TICestE_{,B}(D_n)$ as well. Moreover, the argument given in Appendix~\ref{app:tic_consistent}, which shows that $\TICestE_{,B}(D_n)/B(n)$ converges to 0 in probability under the null hypothesis, can be adapted to show that $\TICestE_{,B}(D_n) \log(B(n))/B(n) \rightarrow 0$ as well. Thus, $\widetilde{\TICestE}_{,B}(D_n) \log(B(n))/B(n)$ converges to 0 in probability, as required.

For the case that $X$ and $Y$ are dependent, the proof is analogous to the argument given in Appendix~\ref{app:tic_consistent} for $\TICestE$. The only difference is that Lemma~\ref{lem:uniform_convergence_clump}, which guarantees the uniform convergence of $\{\widehat{M}\}^c(D_n)$ to $\{M\}^c(X,Y)$, applies only to the $k, \ell$-th entries for which $k, \ell \leq \sqrt{B(n)}$, rather than the entries over which we are summing, which are those for which $k\ell \leq B(n)$. However, since we require only a lower bound on $\widetilde{\TICestE}_{,B}(D_n)$, we may neglect these entries because
\[
\widetilde{\TICestE}_{,B}(D_n) = \sum_{k\ell \leq B(n)} \{\widehat{M}\}^c(D_n)_{k,\ell} \geq \sum_{k,\ell \leq \sqrt{B(n)}}\{\widehat{M}\}^c(D_n)_{k,\ell} .
\]
It can then be shown, following the argument from Appendix~\ref{app:tic_consistent}, that there exists some $a > 0$ depending only on $B$ such that, with probability $1-o(1)$,
\[
\frac{\log B(n)}{B(n)}\left( \sum_{k,\ell \leq \sqrt{B(n)}}\{\widehat{M}\}^c(X,Y)_{k,\ell} - \widetilde{\TICestE}_{,B}(D_n) \right) \leq \O{ \frac{\#_n \log B(n)}{B(n) n^a} } = \O{ \frac{\log B(n)}{n^a} }
\]
where $\#_n = B(n)$ represents the number of pairs $(k,\ell)$ such that $k, \ell \leq \sqrt{B(n)}$. To obtain the result, we note that this means that
\[
\frac{\log B(n)}{B(n)}\widetilde{\TICestE}_{,B}(D_n) \geq \frac{\log B(n)}{B(n)} \sum_{k,\ell \leq \sqrt{B(n)}}\{\widehat{M}\}^c(X,Y)_{k,\ell} -  \O{ \frac{\log B(n)}{n^a} }
\]
and then invoke Proposition~\ref{prop:clumped_equichar_tic}, which implies that for large $n$
\[
\sum_{k,\ell \leq \sqrt{B(n)}}\{M\}^c(X,Y) \geq \Omega \left( \frac{B(n)}{\log B(n)} \right) .
\]

\end{proof}

\subsection{Empirical characterization of the performance of \algname{EquicharClump}}
\label{app:empirical_char_of_c}
The \algname{EquicharClump} algorithm has a parameter $c$ that controls the fineness of the equipartition whose sub-partitions are searched over by the algorithm. To gain an empirical understanding of the effect of $c$ on performance, we computed $\MICestE$ on the set of relationships described in Section~\ref{sec:bias_variance} using \algname{EquicharClump} with different values of $c$. For each relationship, we compared the average $\MICestE$ across all 500 independent samples from that relationship with different values of $c$. We performed this analysis at sample sizes of $n=250$ (Figure~\ref*{fig:characterization_of_c_250}), $n=500$ (Figure~\ref*{fig:characterization_of_c_500}), and $5,000$ (Figure~\ref*{fig:characterization_of_c_5000}).

We summarize our findings as follows.
\begin{itemize}
\item At low ($n=250$) and medium ($n=500$) sample sizes, using $c=1$ introduces a downward bias for more complex relationships when $B(n) = n^{0.6}$ is used but not when $B(n) =n^{0.8}$ is used. This makes sense since the low sample size and low setting of $B(n)$ mean that the algorithm is searching over grids with relatively few cells, and so setting $c=1$ hinders its ability to find good grids in this limited search space. This bias is almost entirely alleviated by setting $c \geq 2$.
\item At high sample size ($n=5,000$), this effect is still observable but much reduced. This makes sense since when $n$ is large, $B(n)$ is large as well, and so the number of cells allowed in the grids being searched over is already large regardless of the exponent $\alpha$ used in $B(n) = n^\alpha$. Thus, there is less need for the robustness provided by searching for an optimal grid.
\end{itemize}

\begin{figure}
\centering
\includegraphics[clip=true, trim = 0in 0in 0in 0in, height=3.6in]{\pathToCommonFigs/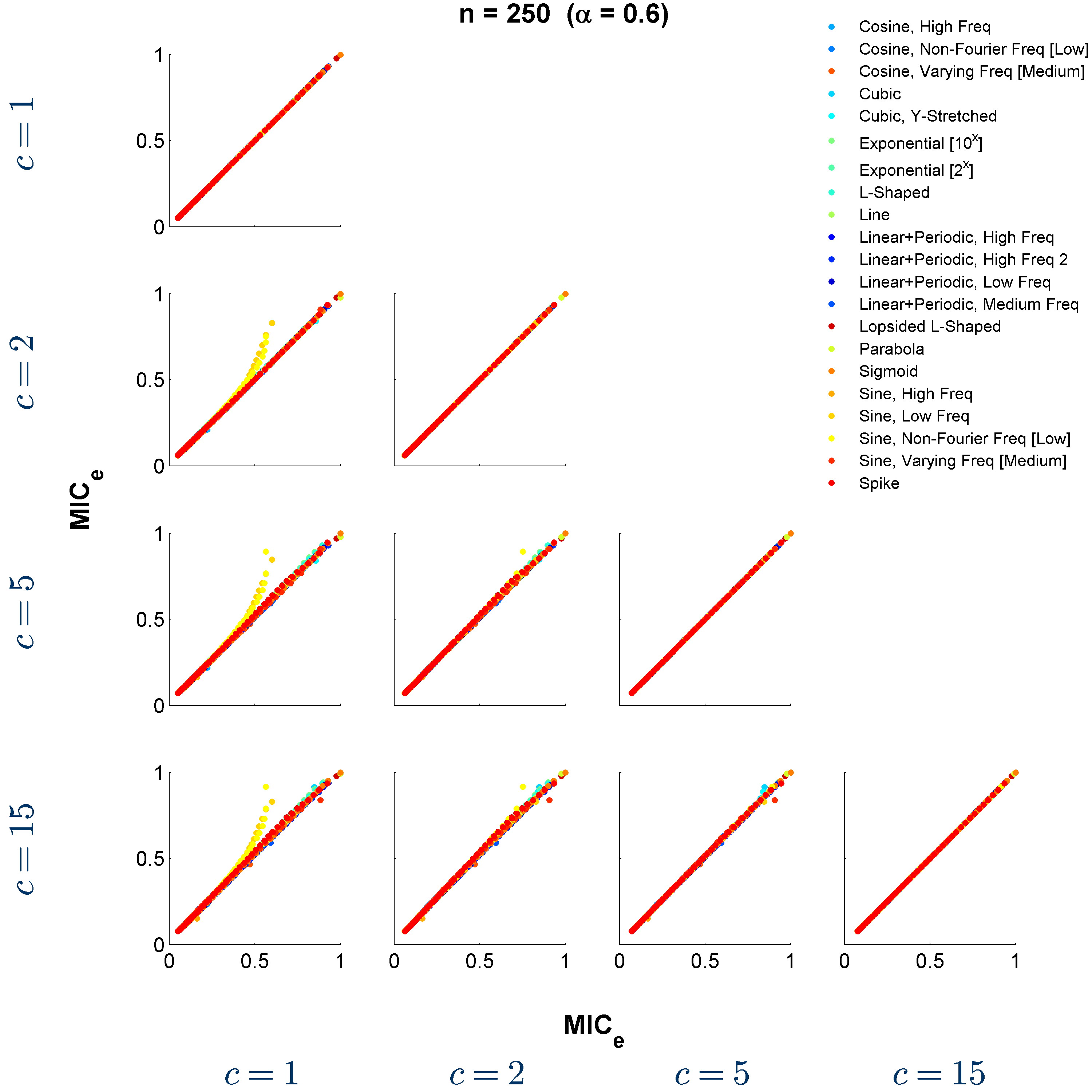} \\
\vspace{0.4in}
\includegraphics[clip=true, trim = 0in 0in 0in 0in, height=3.6in]{\pathToCommonFigs/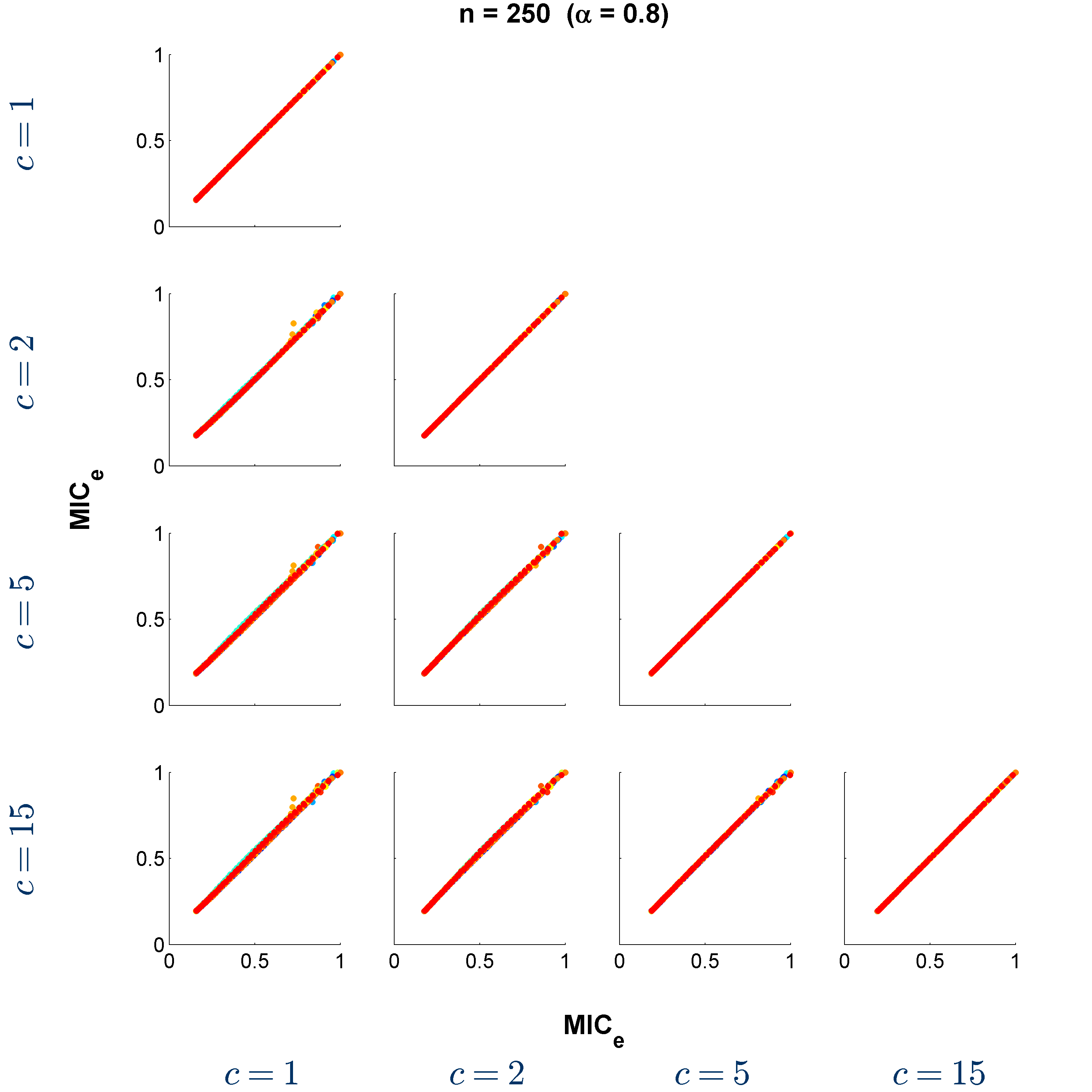}
\caption{The effect of the parameter $c$ on the performance of \algname{EquicharClump}, at $n=250$. See Section~\ref{app:empirical_char_of_c} for details.}
\label{fig:characterization_of_c_250}
\end{figure}

\begin{figure}
\centering
\includegraphics[clip=true, trim = 0in 0in 0in 0in, height=3.6in]{\pathToCommonFigs/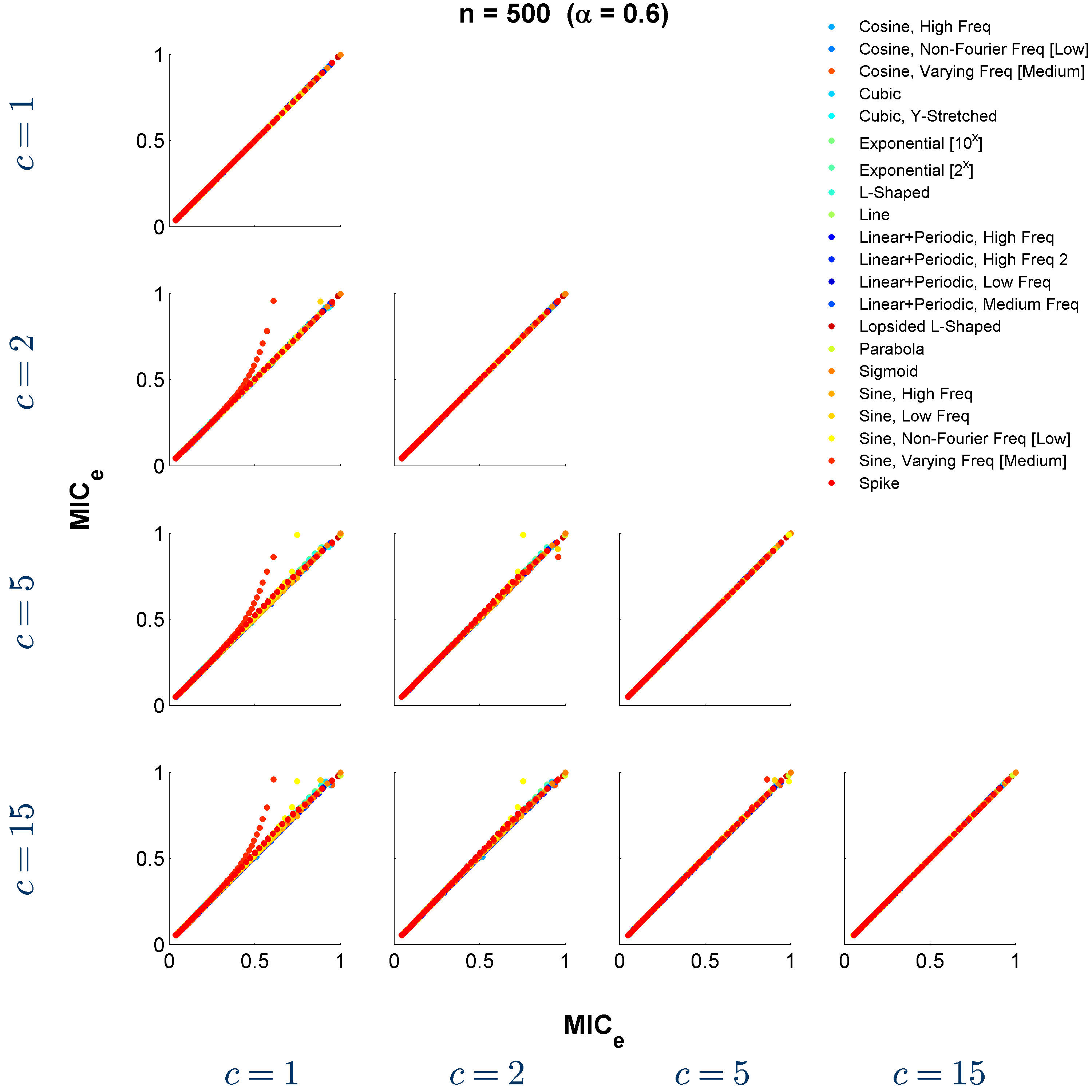} \\
\vspace{0.4in}
\includegraphics[clip=true, trim = 0in 0in 0in 0in, height=3.6in]{\pathToCommonFigs/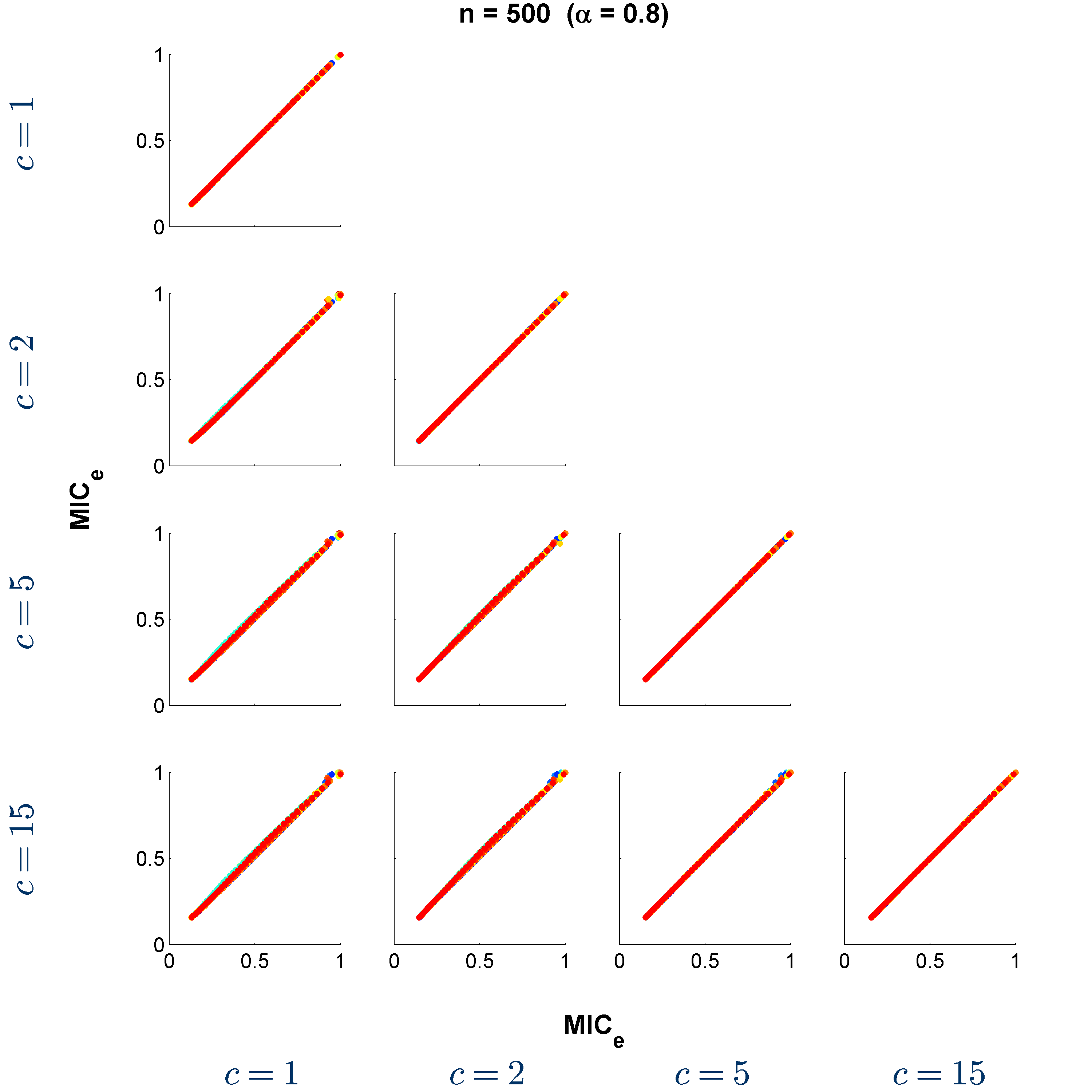}
\caption{The effect of the parameter $c$ on the performance of \algname{EquicharClump}, at $n=500$. See Section~\ref{app:empirical_char_of_c} for details.}
\label{fig:characterization_of_c_500}
\end{figure}

\begin{figure}
\centering
\includegraphics[clip=true, trim = 0in 0in 0in 0in, height=3.6in]{\pathToCommonFigs/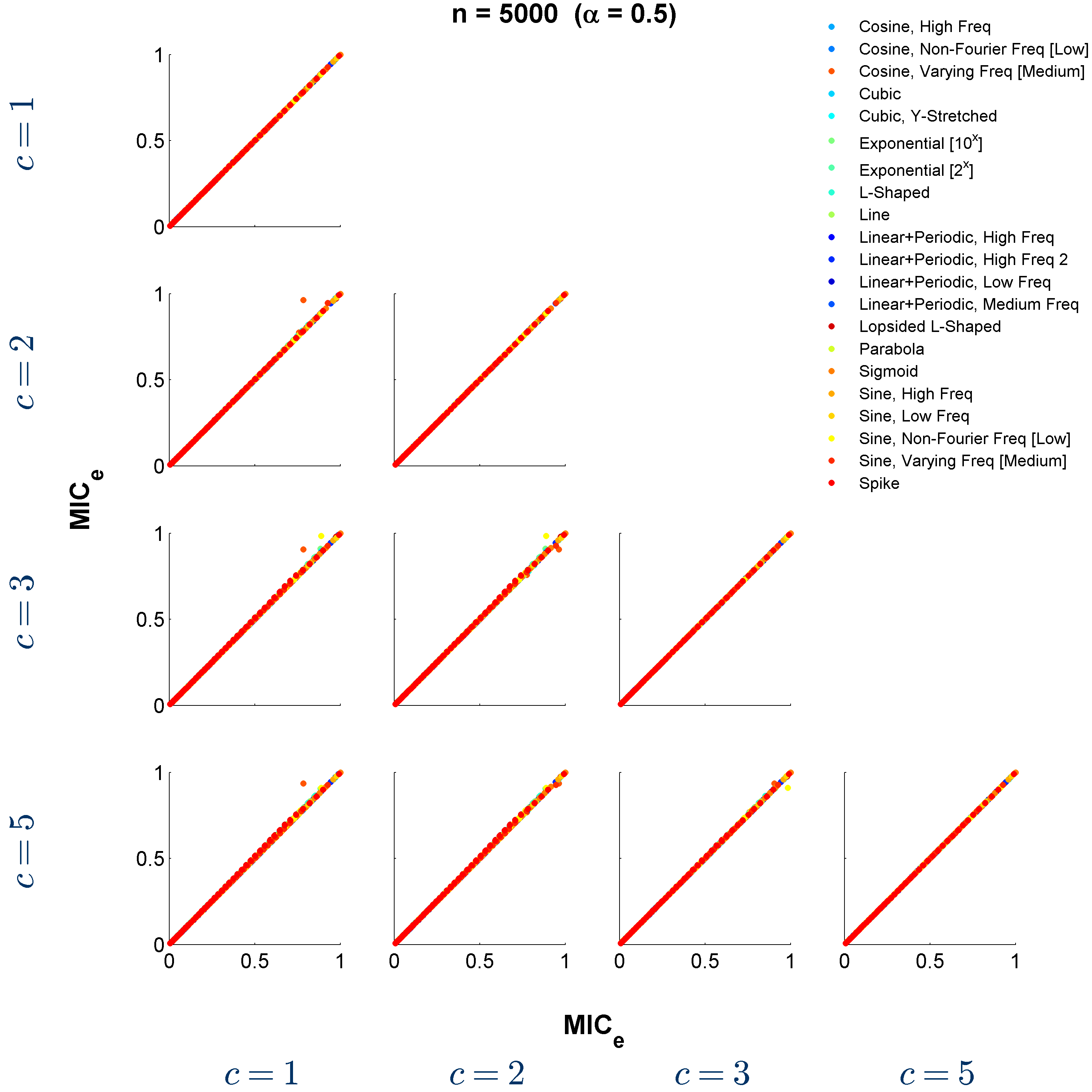} \\
\vspace{0.4in}
\includegraphics[clip=true, trim = 0in 0in 0in 0in, height=3.6in]{\pathToCommonFigs/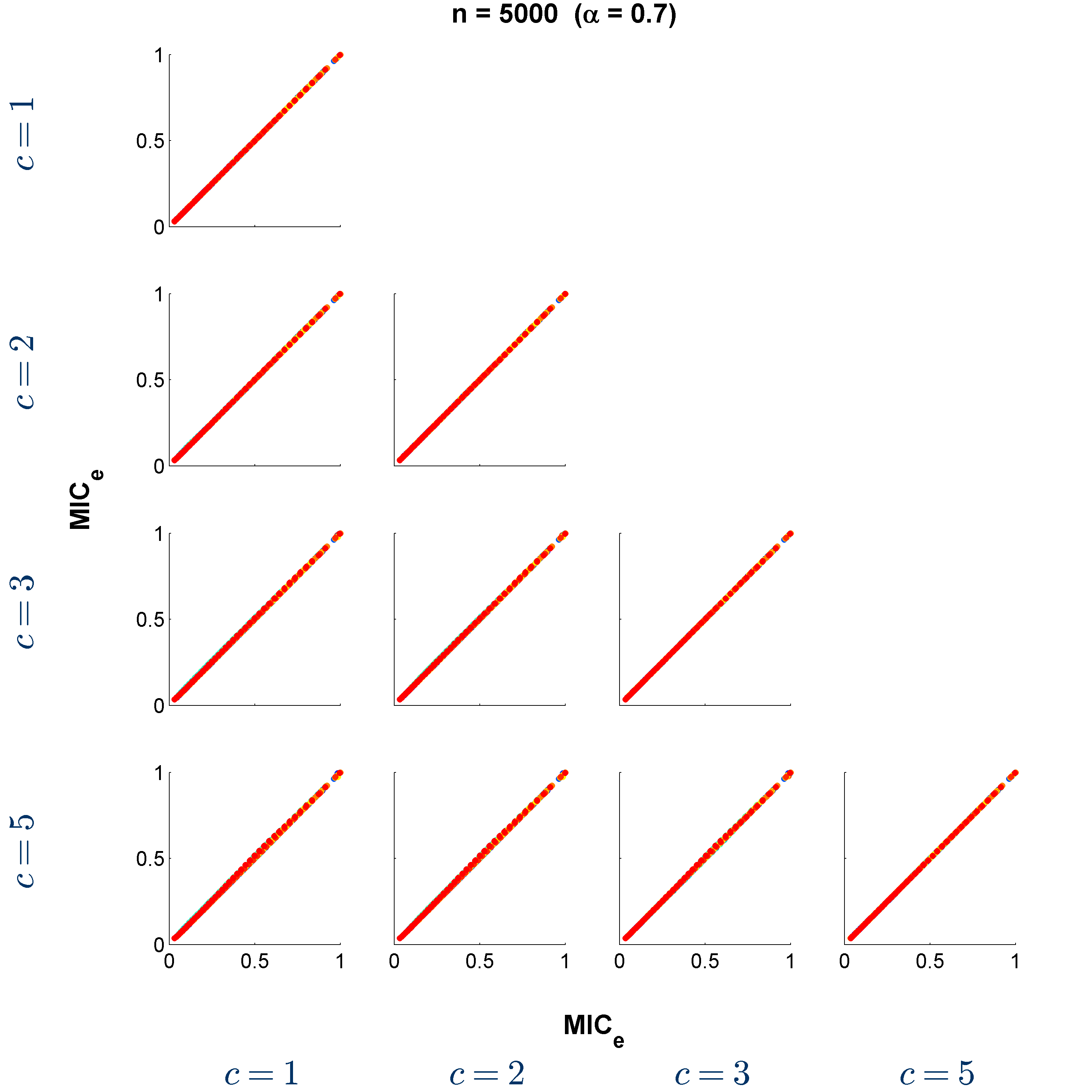}
\caption{The effect of the parameter $c$ on the performance of \algname{EquicharClump}, at $n=5,000$. See Section~\ref{app:empirical_char_of_c} for details.}
\label{fig:characterization_of_c_5000}
\end{figure}

\section{Equitability and power analyses from \texorpdfstring{\cite{reshef2015comparisons}}{the empirical paper}}
\setcounter{figure}{0}
Figure~\ref*{fig:equitabilityAnalysis_evenCurve_XYNoise} contains a representative equitability analysis from \cite{reshef2015comparisons}. Figure~\ref*{fig:indivRelPowerOptimalParam} contains power curves from \cite{reshef2015comparisons} for a large set of leading methods.

\label{app:equitabilityAnalysis_evenCurve_XYNoise}
\label{app:indivRelPowerOptimalParam}
\begin{figure}
	\centering
	\begin{minipage}[b][\textheight][t]{0.62\linewidth}
		\includegraphics[clip=true, trim = 0.15in 1in 0.7in 0.6in, width=\textwidth]{\pathToFigures/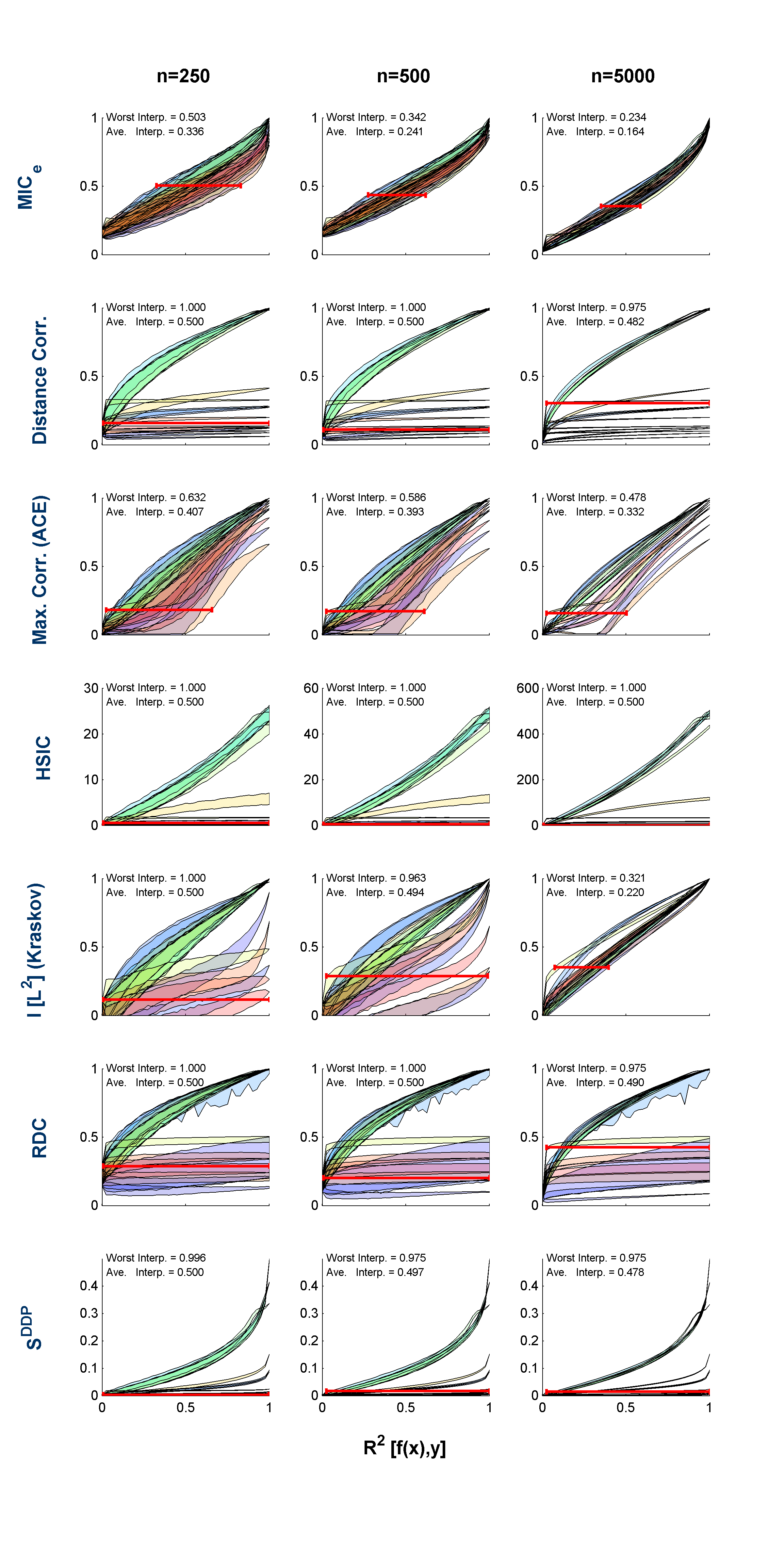}
	\end{minipage}
	\hfill
	\begin{minipage}[b][\textheight][t]{0.34\linewidth}
	    \includegraphics[clip=true, trim = 4.7in 0.25in 2.25in 0.35in, width=\textwidth]{\pathToFigures/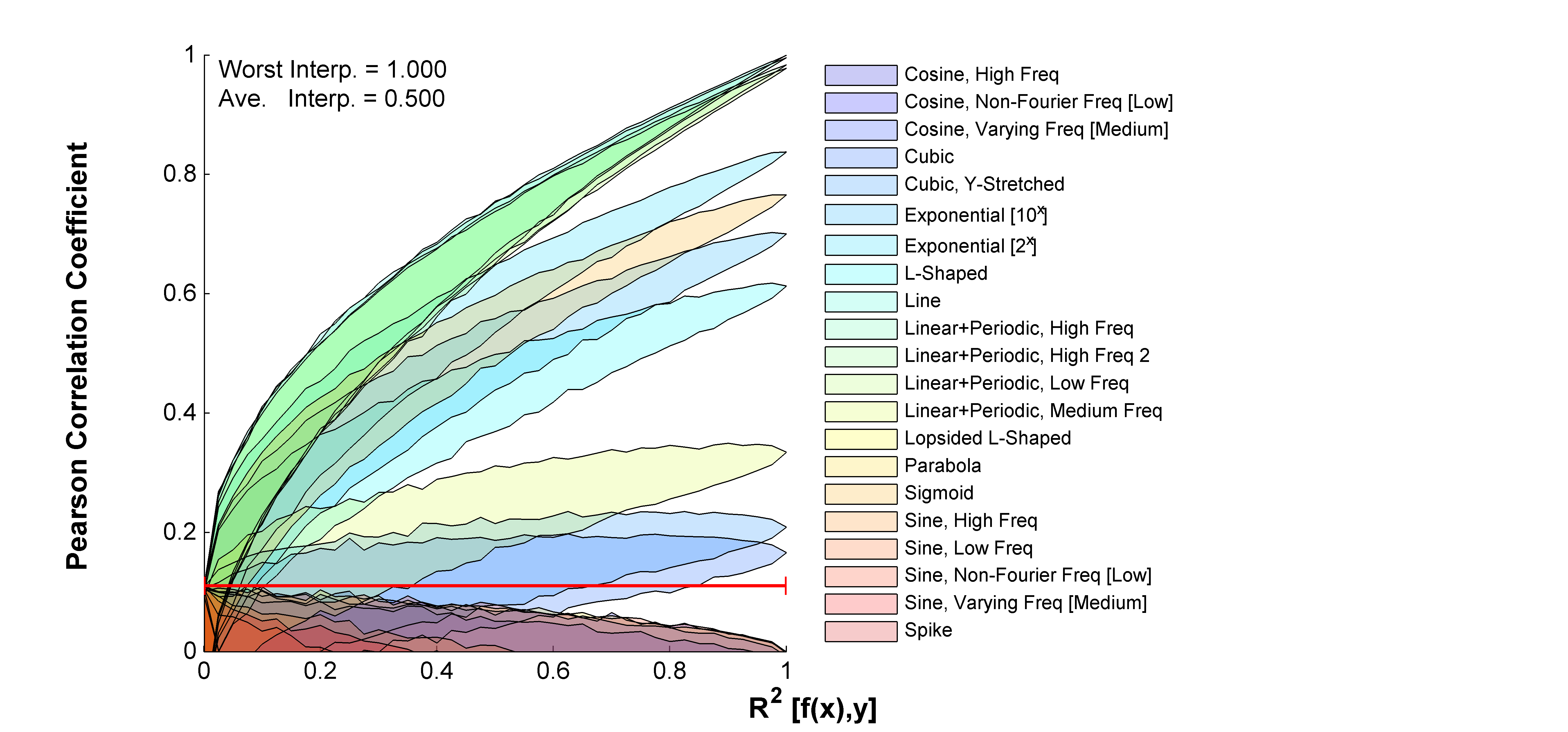}
	\end{minipage}
	\vspace{-1in}
	\caption[The equitability of measures of dependence on noisy functional relationships]{
		    The equitability of measures of dependence on a set of noisy functional relationships, reproduced from \cite{reshef2015comparisons}. \textit{[Narrower is more equitable.]} The plots were constructed as in Figure~\ref{fig:equitabilityAnalysis}.
		    Mutual information, estimated using the Kraskov estimator, is represented using the squared Linfoot correlation.}
	    \label{fig:equitabilityAnalysis_evenCurve_XYNoise}
\end{figure}

\begin{figure}
	\centering
	\includegraphics[clip=true, trim = 0.85in 0.8in 0.775in 0.6in, width=0.975\linewidth]{\pathToFigures/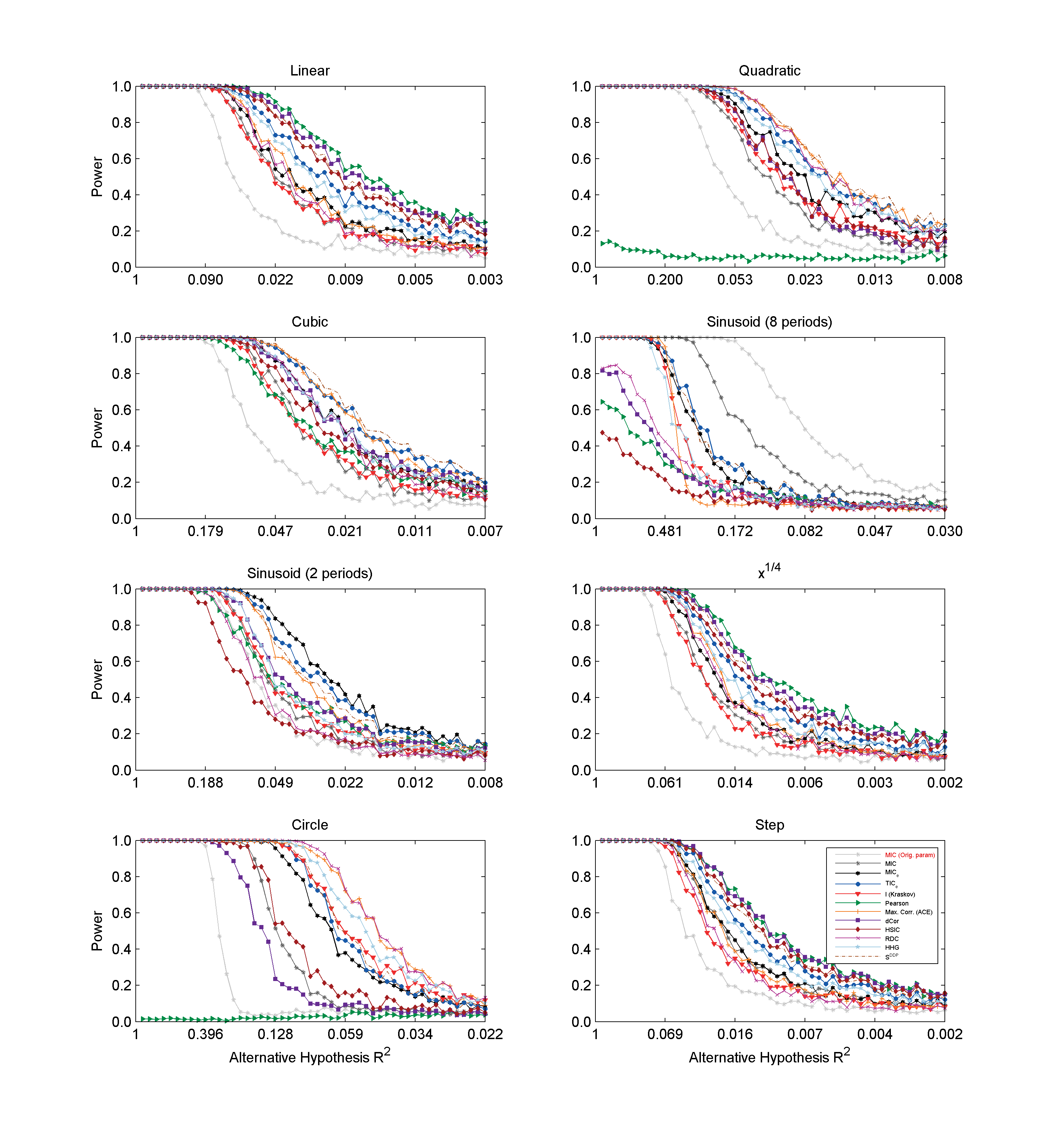}
	\caption[Power against statistical independence of tests based on measures of dependence across several relationship types]{
	    Power of independence testing using several leading measures of dependence, on the relationships chosen by~\cite{simon2012comment}, at $50$ noise levels with linearly increasing magnitude for each relationship and $n=500$. To enable comparison of power regimes across relationships, the x-axis of each plot lists $R^2$ rather than noise magnitude.
	    }
	\label{fig:indivRelPowerOptimalParam}
\end{figure}

\section{Equitability analysis of randomly chosen functions with additional noise model}
\setcounter{figure}{0}
Figure~\ref*{fig:equitabilityAnalysis_evenCurve_GP_YNoise} contains a version of the main text Figure~\ref{fig:equitabilityAnalysis_evenCurve_GP_XYNoise}, but where noise has been added only to the dependent variable in each functional relationship, rather to both the independent and dependent variables.

\begin{figure}
    \centering
    \includegraphics[clip=true, trim = 0in 0in 0in 0in, height=7in]{\pathToFigures/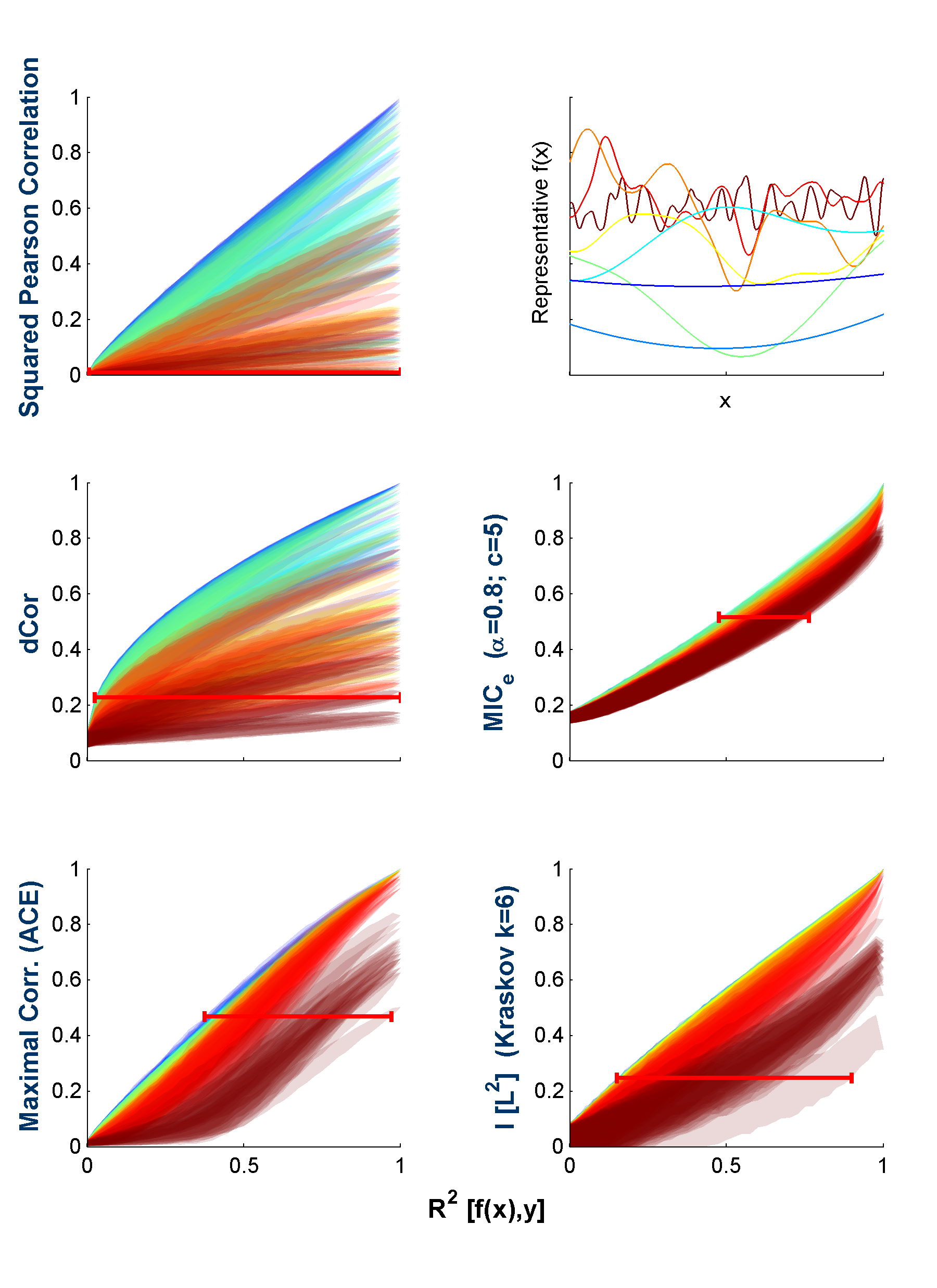}
    \caption{Equitability of methods examined on functions randomly drawn from a Gaussian process distribution, using a different noise model. This figure is identical to Figure~\ref{fig:equitabilityAnalysis_evenCurve_GP_XYNoise}, but with noise added only to the dependent variable in each relationship. Each method is assessed as in Figure~\ref{fig:equitabilityAnalysis_evenCurve_GP_XYNoise}, with a red interval indicating the widest range of $R^2$ values corresponding to any one value of the statistic; the narrower the red interval, the higher the equitability. Sample relationships for each Gaussian process bandwidth are shown in the top right with matching colors.
    }
    \label{fig:equitabilityAnalysis_evenCurve_GP_YNoise}
\end{figure}

\section{Consistency of independence testing based on \texorpdfstring{$\TICestE$}{TICe}}
\label{app:tic_consistent}
\setcounter{figure}{0}
Here we prove Propositions~\ref{prop:nonzero_submatrix} and~\ref{prop:growth_of_S} and then use those propositions to prove Theorem~\ref{thm:consistency_TICe}, which shows that $\TICestE$ can be used for independence testing.

\subsection{Proof of Proposition~\ref{prop:nonzero_submatrix}}
\label{app:nonzero_submatrix}
\vspace{0.3cm}
\noindent \textbf{Proposition}
\textit{
Let $(X,Y)$ be a pair of jointly distributed random variables. If $X$ and $Y$ are statistically independent, then $M(X,Y) \equiv [M](X,Y) \equiv 0$. If not, then there exists some $a > 0$ and some integer $\ell_0 \geq 2$ such that
\[
M(X,Y)_{k,\ell}, [M](X,Y)_{k,\ell} \geq \frac{a}{\log \min \{k, \ell\}}
\]
either for all $k \geq \ell \geq \ell_0$, or for all $\ell \geq k \geq \ell_0$.
}
\vspace{0.3cm}
\begin{proof}
We give the proof only for $[M] = [M](X,Y)$, with the understanding that all parts of the argument are either identical or similar for $M(X,Y)$. When $X$ and $Y$ are independent, then for any grid $G$, $(X,Y)|_G$ exhibits independence as well. Therefore $I((X,Y)|_G) = 0$ for all grids $G$, and so every entry of $[M]$, being a supremum over such quantities, is 0.

For the case that $X$ and $Y$ are dependent, our strategy is to first find, without loss of generality, a column of $[M]$ almost all of whose values are bounded away from zero, and then argue that this suffices.

The dependence of $X$ and $Y$ implies that $\popMIC(X,Y) > 0$. By Corollary~\ref{cor:boundary_equichar}, which states that $\sup \partial [M] = \popMIC(X,Y)$, we therefore know that there is at least one non-zero element of the boundary of $[M]$, as defined in Definition~\ref{def:boundary}. Without loss of generality, suppose that this element is $[M]_{\uparrow, \ell_0} = \lim_{k \rightarrow \infty} [M]_{k,\ell_0}$. The fact that this limit is strictly positive implies that there exists some $k_0 \geq \ell_0$ and some $r > 0$ such that $[M]_{k, \ell_0} \geq r$ for all $k \geq k_0$. That is, all but finitely many of the entries in the $\ell_0$-th column of $[M]$ are at least $r$.

We now show that the existence of such a column suffices to prove the claim. Fix some $k > k_0$ and note that this implies that $k > \ell_0$. We argue that for all $\ell$ in $\{\ell_0, \ldots, k\}$, the desired condition holds. Since $k > \ell_0$, the term $I^{[*]}((X,Y), k, \ell_0)$ in the definition of $[M]_{k,\ell_0}$ is a maximization over grids that have an equipartition of size $k$ on one axis and an optimal partition of size $\ell_0$ on the other. Since we allow empty rows/columns in the maximization, substituting any $\ell$ satisfying $\ell_0 \leq \ell \leq k$ therefore does not constrain the maximization in any way and so it cannot decrease $I^{[*]}$. In other words, for $\ell$ satisfying $\ell_0 \leq \ell \leq k$, we have
\[
I^{[*]}((X,Y), k, \ell) \geq I^{[*]}((X,Y), k, \ell_0) .
\]
Since $k \geq \ell, \ell_0$, the normalizations in the definition of $[M]_{k, \ell}$ and $[M]_{k, \ell_0}$ are $\log \ell$ and $\log \ell_0$ respectively. Therefore, we have that
\[
[M]_{k, \ell} \geq [M]_{k, \ell_0} \frac{\log \ell_0}{\log \ell} \geq \frac{r \log \ell_0}{\log \ell}
\]
where the last inequality is because $k > k_0$. Setting $a = r \log \ell_0$ then gives the result.
\end{proof}

\subsection{Proof of Proposition~\ref{prop:growth_of_S}}
\label{app:growth_of_S}
\vspace{0.3cm}
\noindent \textbf{Proposition}
\textit{
Let $(X,Y)$ be a pair of jointly distributed random variables. If $X$ and $Y$ are statistically independent, then $S_B(M(X,Y)) = S_B([M](X,Y)) = 0$ for all $B > 0$. If not, then $S_B(M(X,Y))$ and $S_B([M](X,Y))$ are both $\Omega(B \log \log B)$.
}
\vspace{0.3cm}
\begin{proof}
We give the argument for $M = M(X,Y)$ only, but the argument holds as stated for $[M](X,Y)$ as well.

The result follows from the guarantee given by the Proposition~\ref{prop:nonzero_submatrix} above. In the case of independence, the proposition tells us that $M \equiv 0$, which immediately gives that $S_B(M) = 0$ for all $B > 0$. For the case of dependence, the proposition implies that there is some $a > 0$ and some integer $\ell_0 \geq 2$ such that, without loss of generality, $M_{k, \ell} \geq a / \log \ell$ for all $k \geq \ell \geq \ell_0$. We convert this into a lower bound on $S_B(M)$.

The key is to write the sum one column at a time, counting how many entries in each column both satisfy $k \geq \ell \geq \ell_0$ and $k\ell \leq B$. For any $\ell$ satisfying $\ell_0 \leq \ell \leq \sqrt{B}$, the entries $(\ell, \ell), \ldots, (B/\ell, \ell)$ meet this criterion, and there are $B / \ell_0 - (\ell_0-1)$ of them. Moreover, since the guarantee of Proposition~\ref{prop:nonzero_submatrix} tells us that all of these entries are at least $a / \log \ell$, we can lower-bound $S_B(M)$ as follows.
\begin{align*}
S_B(A) &\geq \sum_{\ell = \ell_0}^{\sqrt B} \frac{a}{\log \ell} \left( \frac{B}{\ell} - (\ell - 1) \right)  \\
    &= a B \sum_{\ell = \ell_0}^{\sqrt B} \frac{1}{\ell \log \ell} - a\sum_{\ell = \ell_0}^{\sqrt B}\frac{\ell - 1}{\log \ell} \\
    &= a \left( B \sum_{\ell = \ell_0}^{\sqrt B} \frac{1}{\ell \log \ell} - O(B) \right) \\
    &= \Omega (B \log \log B)
\end{align*}
where the second-to-last equality is because $(\ell - 1)/\log \ell \leq \ell$, and the last equality is because $\sum_{i = i_0}^n 1/(i \log i)$ grows like $\log \log n$.
\end{proof}

\subsection{Proof of Theorem~\ref{thm:consistency_TICe}}
\label{app:consistency_TICe}
\vspace{0.3cm}
\noindent \textbf{Theorem}
\textit{
The statistics $\TIC_B$ and $\TICestE_{,B}$ yield consistent right-tailed tests of independence, provided $\omega(1) < B(n) \leq O(n^{1-\ep})$ for some $\ep > 0$.
}
\vspace{0.3cm}
\begin{proof}
We give the proof for $\TIC$ only; however, the argument holds as stated for $\TICestE$ as well.

Let $(X,Y)$ be jointly distributed random variables, and let $D_n$ be a sample of size $n$ from the distribution of $(X,Y)$. Let $M = M(X,Y)$ be the characteristic matrix of $(X,Y)$ and let $\widehat{M}(D_n)$ be the sample characteristic matrix. It suffices to establish the result for a deterministic monotonic function of $\TIC_B(D_n)$. We therefore show convergence of $\TIC_B(D_n)/B(n)$ to zero under the null hypothesis of independence and to $\infty$ under any alternative. Our general strategy for doing so is to translate known bounds on our error at estimating entries of $M$ into bounds on the difference between $\TIC_B(D_n)/B(n) = S_{B(n)}(\widehat{M}(D_n))/B(n)$ and $S_{B}(M)/B(n)$. We then obtain the result by invoking Proposition~\ref{prop:growth_of_S}, which implies that $S_B(M)/B(n)$ is zero under the null hypothesis but grows without bound under the alternative.

We know from Lemma~\ref{lem:boundCharMatrixEntries} (Lemma~\ref{lem:boundEquiCharMatrixEntries} for the equicharacteristic matrix) that there exists some $a >0$ depending only on $B$ such that
\[
\left| \widehat{M}(D_n)_{k,\ell} - M_{k,\ell} \right| \leq O \left( \frac{1}{n^a} \right)
\]
for all $k\ell \leq B(n)$ with probability $1 - o(1)$. This means that with probability $1-o(1)$ we have
\[
\frac{1}{B(n)}\left| \TIC_B(D_n) - S_{B(n)}(M) \right| \leq O\left( \frac{\#_n}{B(n) n^a} \right)
\]
where $\#_n$ is the number of pairs $(k, \ell)$ such that $k\ell \leq B(n)$. It can be shown by taking the integral of $B/x$ with respect to $x$ that $\#_n = O(B(n) \log B(n))$. Therefore, the error in the above bound is at most $O(\log B(n) / n^a) = O(1/\mbox{poly}(n))$ for our choice of $B(n)$.

We now use Proposition~\ref{prop:growth_of_S} to show that this bound gives the desired result. Under the null hypothesis of independence, the proposition says that $S_{B(n)}(M) = 0$ always, and so since $B$ is a growing function the bound implies that $\TIC_B(D_n)/B(n) \rightarrow 0$ in probability. Under the alternative hypothesis in which $(X,Y)$ exhibit a dependence, the proposition implies that $S_{B(n)}(M)/B(n) > \omega(1)$. Since $B$ is a growing function of $n$, this means that for any $r > 0$, the probability that $S_{B(n)}(M)/B(n) > r$ goes to 1 as $n$ grows. In other words, $\TIC_B(D_n)/B(n) \rightarrow \infty$ in probability.
\end{proof}

\section{Information-theoretic lemmas}
\label{appendix:technical}
\setcounter{figure}{0}
\begin{lemma}
\label{lem:boundEntropy}
Let $\Pi$ and $\Psi$ be random variables distributed over a discrete set of states $\Gamma$, and let $(\pi_i)$ and $(\psi_i)$ be their respective distributions. Let $P = f(\Pi)$ and $Q = f(\Psi)$ for some function $f$ whose image is of size $B$. Define
\[ \ep_i = \frac{\psi_i - \pi_i}{\pi_i} .\]
Then for every $0 < a < 1$ there exists some $A > 0$ such that
\[ \left| H(Q) - H(P) \right| \leq \left( \log B \right) A \sum_i |\ep_i| \]
when $|\ep_i| \leq 1 - a$ for all $i$.
\end{lemma}
\begin{proof}
We prove the claim with entropy measured in nats. A rescaling then gives the general result.

Let $(p_i)$ and $(q_i)$ be the distributions of $P$ and $Q$ respectively, and define
\[ e_i = \frac{q_i - p_i}{p_i} \]
analogously to $\ep_i$. Before proceeding, we observe that
\[ e_i = \sum_{j \in f^{-1}(i)} \frac{\pi_j}{p_i} \ep_j .\]

We now proceed with the argument. We have from \cite{roulston1999estimating} that
\begin{eqnarray}
\left| H(Q) - H(P) \right| &\leq& \left| \sum_i \left( e_i p_i (1 + \ln p_i) + \frac{1}{2} e_i^2 p_i + \O{e_i^3} \right) \right| \\
	&\leq& \left| \sum_i e_i p_i \right| + \left| \sum_i e_i p_i \ln p_i \right| + \frac{1}{2} \left| \sum_i e_i^2 p_i \right| + \left| \sum_i \O{e_i^3} \right| \\
	&=& \left| \sum_i e_i p_i \ln p_i \right| + \frac{1}{2} \sum_i e_i^2 p_i + \left| \sum_i \O{e_i^3} \right| \label{eq:toBound}
\end{eqnarray}
where the final equality is because $\sum_i e_i p_i = \sum_i q_i - \sum_i p_i = 0$. We proceed by bounding each of the terms in Equation~\ref{eq:toBound} separately.

To bound the first term, we write
\[ \left| \sum_i e_i p_i \ln p_i \right| \leq -\sum_i | e_i | p_i \ln p_i .\]
We then note that $- \sum_i p_i \ln p_i \leq \ln B$, and since each of the summands has the same sign this means that $-p_i \ln p_i \leq \ln B$. We also observe that
\[ \left| e_i \right| \leq \left| \sum_{j \in f^{-1}(i)} \frac{\pi_j}{p_i} \ep_j \right| \leq \sum_j \frac{\pi_j}{p_i} \left| \ep_j \right| \leq \sum_j |\ep_j| \]
since $\pi_j/p_i \leq 1$. Together, these two facts give
\begin{eqnarray*}
- \sum_i |e_i| p_i \ln p_i &\leq& (\ln B) \sum_i |e_i| \\
	&\leq& (\ln B) \sum_i |\ep_i|
\end{eqnarray*}
The second inequality is because each $e_i$ is a weighted average of a set of $\ep_i$ and each $\ep_i$ enters into the expression of exactly one $e_i$.

To bound the second term, we use the fact that $p_i \leq 1$ for all $i$, and so
\[ \sum_i e_i^2 p_i \leq \sum_i e_i^2 . \]
We then write
\begin{eqnarray*}
\sum_i e_i^2 &=& \sum_i \left( \sum_{j \in f^{-1}(i)} \frac{\pi_j}{p_i} \ep_j \right)^2 \\
	&\leq& \sum_i \sum_{j \in f^{-1}(i)} \frac{\pi_j}{p_i} \ep_j^2 \\
	&\leq& \sum_j \ep_j^2 \\
	&=& \sum_j \O{\left| \ep_j \right|}
\end{eqnarray*}
where the second line is a consequence of the convexity of $f(x) = x^2$ and the third line is because the sets $f^{-1}(i)$ partition $\Gamma$.

To bound the third term, we write
\[ \left| \sum_i \O{e_i^3} \right| \leq \sum_i \O{|e_i|^3} \]
and then proceed as we did with the second term, using the fact that $f(x) = x^3$ is convex for $x \geq 0$. This gives
\[ \sum_i \O{|e_i|^3} \leq \sum_i \O{|\ep_i|^3} = \sum_i \O{|\ep_i|} \]
completing the proof.
\end{proof}

\begin{lemma}
\label{lem:boundWeightedAverageOfEntropy}
Let $\{ w_i \} \subset [0,1]$ be a set of size $n$ with $\sum_i w_i \leq 1$, and let $\{u_i\}$ be a set of $n$ non-negative numbers satisfying $\sum_i u_i = a$ and $u_i \leq w_i$. Then
\[
\sum_{i=1}^n w_i H_b \left( \frac{u_i}{w_i} \right) \leq H_b \left( a \right)
\]
where $H_b$ is the binary entropy function.
\end{lemma}
\begin{proof}
Consider the random variable $X$ taking values in $\{0, \ldots, n\}$ that equals 0 with probability $1-\sum_i w_i$ and equals $i$ with probability $w_i$ for $0 < i \leq n$. Define the random variable $Y$ taking values in $\{0,1\}$ by
\[ \Pr{Y = 0 | X = i} = \left\{
        \begin{array}{ll}
            0 & \quad i = 0 \\
            u_i / w_i & \quad 0 < i \leq n
        \end{array}
    \right. .
 \]
The function we wish to bound equals $H(Y | X) \leq H(Y)$. We therefore observe that
\[ \sum_{i=1}^n w_i H_b \left( \frac{u_i}{w_i} \right) \leq H(Y) . \]
The result follows from the observation that
\[ \Pr{Y = 0} = \sum_i \Pr{X = i} \frac{u_i}{w_i} = \sum_i u_i \leq a  .\]
\end{proof}

\begin{lemma}
Let $X$ be a random variable distributed over $k$ states, with $\Pr{X = x} = p_x$. Let $\alpha_x \geq 0$ be such that $\sum \alpha_x = \delta$, and define the random variable $X'$ by $\Pr{X' = x} = (p_x + \alpha_x)/(1 + \delta)$. We have
\[
\left| H(X') - H(X) \right| \leq H_b(\delta) + \delta \log k
\]
where $H_b$ is the binary entropy function.
\end{lemma}
\begin{proof}
Define a new random variable $Z$ by
\[ \Pr{Z = 0 | X' = x} = \frac{p_x}{p_x+\alpha_x}, \quad \Pr{Z = 1 | X' = x} = \frac{\alpha_x}{p_x + \alpha_x} .\]
We will use the fact that $H(X' | Z = 0) = H(X)$ to obtain the required bound.

To upper bound $H(X') - H(X)$, we write
\begin{eqnarray*}
H(X') - H(X) &\leq& H(X', Z) - H(X)  \\
	&=& H(Z) + \Pr{Z = 0} H(X' | Z = 0) + \Pr{Z = 1} H(X' | Z = 1) - H(X) \\
	&\leq& H_b(\delta) + (1-\delta) H(X) + \delta H(X' | Z = 1) - H(X) \\
	&=& H_b(\delta) - \delta H(X) + \delta \log k \\
	&\leq& H_b(\delta) + \delta \log k
\end{eqnarray*}
where in the fourth line we have used that $H(X' | Z = 1) \leq \log k$.

To upper bound $H(X) - H(X')$, we write
\begin{eqnarray*}
H(X') + H(Z) &\geq& H(X', Z) \\
	&\geq& \Pr{Z=0} H(X' | Z = 0) \\
	&=& (1-\delta) H(X)
\end{eqnarray*}
which yields
\[ H(X') \geq (1-\delta) H(X) - H_b(\delta) \]
since $H(Z) = H_b(\delta)$. Thus, we have
\[ H(X) - H(X') \leq \delta H(X) + H_b(\delta) \leq \delta \log k + H_b(\delta) . \]
\end{proof}

\begin{lemma}
Let $X$ be a random variable distributed over $k$ states, with $\Pr{X = x} = p_x$. Let $\alpha_x \leq 0$ be such that $\sum |\alpha_x| = \delta$, and define the random variable $X'$ by $\Pr{X' = x} = (p_x + \alpha_x)/(1 - \delta)$. We have
\[
\left| H(X') - H(X) \right| \leq H_b \left( \frac{\delta}{1-\delta} \right) + \frac{\delta}{1-\delta} \log k
\]
where $H_b$ is the binary entropy function. In particular, when $\delta \leq 1/3$ we have
\[
\left| H(X') - H(X) \right| \leq H_b(2\delta) + 2\delta \log k .
\]
\end{lemma}
\begin{proof}
We observe that we can get from $X'$ to $X$ by adding $\delta / (1-\delta)$ probability mass and rescaling. The previous lemma then gives the result.
\end{proof}

\begin{lemma}
\label{lem:entropyBoundForChangedMass}
Let $X$ be a random variable distributed over $k$ states, with $\Pr{X = x} = p_x$. Let $\alpha_x$ be such that $\sum |\alpha_x| = \delta$, and define the random variable $X'$ by $\Pr{X' = x} = (p_x + \alpha_x)/(1 - \sum \alpha_x)$. That is, $X'$ is the result of changing the probability of state $x$ by $\alpha_x$ and then re-normalizing to obtain a valid distribution. If $\delta \leq 1/4$, we have
\[
\left| H(X') - H(X) \right| \leq 2H_b(2\delta) + 3\delta \log k
\]
where $H_b$ is the binary entropy function.
\end{lemma}
\begin{proof}
Let $\delta_+$ be the total magnitude of all the positive $\alpha_x$, and let $\delta_-$ be the total magnitude of all the negative $\alpha_x$. We first add all the mass we're going to add, and apply the first of the previous two lemmas. Then we remove all the mass we are going to remove, and apply the second of the two previous lemmas. This yields a bound of
\begin{eqnarray*}
&&	H_b(\delta_+) + \delta_+ \log k + H_b\left(2\frac{\delta_-}{1 + \delta_+} \right) + 2 \frac{\delta_-}{1 + \delta_+} \log k \\
&\leq& H_b(\delta_+) + \delta_+ \log k + H_b(2 \delta_-) + 2\delta_- \log k \\
&\leq& H_b(2\delta) + \delta \log k + H_b(2 \delta) + 2 \delta \log k \\
&\leq& 2H_b(2 \delta) + 3 \delta \log k
\end{eqnarray*}
where the first inequality is because $1 + \delta_+ \leq 1 + \delta < 2$ and $2 \delta_- \leq 2 \delta \leq 1/2$, and the second inequality is because $\delta_+ \leq \delta < 2\delta \leq 1/2$.
\end{proof}

\newpage
\bibliographystyle{ieeetr}
\bibliography{\pathToCommon/References}


\end{document}

%% file: commands.tex
\usepackage[T1]{fontenc}
\newcommand{\algname}[1]{\textsc{#1}}

\newcommand{\R}{\mathbb{R}}
\newcommand{\Z}{\mathbb{Z}}
\newcommand{\ep}{\varepsilon}
\renewcommand{\Pr}[1]{\mathbf{P} \left( {#1} \right)}
\newcommand{\Prsm}[1]{\mathbf{P} ( {#1} )}
\newcommand{\E}[1]{\mathbf{E} \left( {#1} \right)}
\renewcommand{\O}[1]{O \left( {#1} \right) }

\renewcommand{\P}{\boldsymbol{\mathcal{P}}}
\newcommand{\Q}{\boldsymbol{\mathcal{Q}}}

\newcommand{\mcZ}{\mathcal{Z}}
\newcommand{\hvphi}{{\hat{\varphi}}}

\newcommand{\MIC}{\mbox{MIC}}
\newcommand{\popMIC}{{\mbox{MIC}_*}}
\newcommand{\MICestE}{{\mbox{MIC}_e}}

\newcommand{\TIC}{\mbox{TIC}}
\newcommand{\TICestE}{{\mbox{TIC}_e}}

\newcommand{\proofwaypoint}[1]{ {\em {#1}:} }

\newtheorem{thm}{Theorem}[section]
\newtheorem*{thm*}{Theorem}

\newtheorem{lemma}[thm]{Lemma}

\theoremstyle{definition}
\newtheorem{definition}[thm]{Definition}

\newcommand{\figpart}[1]{\textbf{({#1})}}